\documentclass[11pt]{article}%
\usepackage{amssymb}
\usepackage{amsfonts}
\usepackage{amsmath}
\usepackage{amssymb}
\usepackage{color}
\usepackage{graphicx}
\usepackage{caption}%
\usepackage{hyperref}
\usepackage{float}
\usepackage{subcaption}
\usepackage{multirow}
\usepackage{authblk}

\usepackage{tabularx}
\usepackage{algorithm}
\usepackage{algpseudocode}
\providecommand{\U}[1]{\protect\rule{.1in}{.1in}}
\newtheorem{theorem}{Theorem}[section]

\newtheorem{corollary}[theorem]{Corollary}
\newtheorem{lemma}[theorem]{Lemma}

\newtheorem{remark}[theorem]{Remark}

\newenvironment{proof}[1][Proof]{\noindent \textbf{#1.} }{\  \rule{0.5em}{0.5em}}
\begin{document}
	
	\title{Robust adaptive variable selection in ultra-high dimensional linear regression models}
	\author{ Abhik Ghosh \textsuperscript{a}, María Jaenada \textsuperscript{b}		and Leandro Pardo \textsuperscript{b} 
	\footnote{A. Ghosh. Email: abhik.ghosh@isical.ac.in\\ 
		M. Jaenada. Email: mjaenada@ucm.es\\
		L. Pardo. Email: lpardo@mat.ucm.es}} 
	\affil{\textsuperscript{a} Indian Statistical Institute, Kolkata, India. \textsuperscript{b} Complutense University of Madrid, Spain.}
	\maketitle

	\begin{abstract}
		We consider the problem of simultaneous variable selection and parameter estimation in an ultra-high dimensional linear regression model.
		The adaptive penalty functions are used in this regard to achieve the oracle variable selection property with simpler assumptions and lesser computational burden. Noting the non-robust nature of the usual adaptive procedures (e.g., adaptive LASSO) based on the squared error loss function against data contamination, quite frequent with modern large-scale data sets (e.g., noisy gene expression data, spectra and spectral data), in this paper, we present a new adaptive  regularization procedure using a robust loss function based on the density power divergence (DPD) measure under a general class of error distributions. We theoretically prove that  the proposed adaptive DPD-LASSO estimator of the regression coefficients is highly robust, consistent, asymptotically normal and leads to robust oracle-consistent
		variable selection under easily verifiable assumptions. Numerical illustrations are provided for the mostly used normal and heavy-tailed error densities. Finally, the proposal is applied to analyze an interesting spectral dataset, in the field of chemometrics, regarding the electron-probe X-ray microanalysis (EPXMA) of archaeological glass vessels from the 16th and 17th centuries.
	\end{abstract}
	
	\textbf{keywords:}
		High-dimensional linear regression models; 
		Adaptive LASSO estimator; 
		non-polynomial dimensionality; 
		oracle property; 
		density power divergence.

	\textbf{AMS Code:}
		MSC primary 62F35; MSEC secondary 62J07
		
	\section{Introduction}\label{SEC:intro}
	
	Let us consider the standard linear regression model (LRM) given by 
	\begin{equation}
		\boldsymbol{y}=\boldsymbol{X\beta }+\boldsymbol{\varepsilon },  \label{1.1}
	\end{equation}%
	where $\boldsymbol{y}=\left( y_{1},...,y_{n}\right) ^{T}$ denotes the vector of observations
	from a response variable ${Y}$, $\boldsymbol{X}=\left( 
	\boldsymbol{x}_{1},...,\boldsymbol{x}_{n}\right) ^{T}$ is the design matrix
	containing associated values of the explanatory variable $\mathbb{X}=(X_{1},...,X_{p})\mathbb{\in R}^{p}$, $\boldsymbol{\beta }=(\beta_{1},...,\beta _{p})$ is the regression coefficient vector,  and $\boldsymbol{\varepsilon =}\left( \varepsilon
	_{1},...,\varepsilon _{n}\right) ^{T}$ follows $n$-variate normal distribution with mean vector $\boldsymbol{0}_n$ and
	variance-covariance matrix $\sigma ^{2}\boldsymbol{I}_{n}.$
	High-dimensional statistics under the LRM (\ref{1.1}) refers to the situation where the number of unknown parameters ($p$) is of much larger order than the sample size $n.$ 	Here, we consider the case where  $p$ grows exponentially with $n,$ i.e., $p=O(e^{n^{l}})$ for some $l\in \left( 0,1\right) $; such case is  often referred to as ultra-high dimensional or of non-polynomial dimensionality.
	
	In recent years, several statistical methods, algorithms and theories have been
	developed to perform high-dimensional data analysis. Penalized least square
	methods have become popular for the high-dimensional LRM since Tibshirami  \cite{ten3} introduced the least absolute shrinkage and selection operator (LASSO) estimate. Under the LRM (\ref{1.1}), the LASSO estimate of $\boldsymbol{\beta}$ is defined as 
	\begin{equation} 
		\widehat{\boldsymbol{\beta }}_{\mathrm{L}}=\widehat{\boldsymbol{\beta }}_{
			\mathrm{L}}(\lambda )=\arg \min_{\boldsymbol{\beta }\in \boldsymbol{\ 
				\mathbb{R}
			}^{p}}\left( \left\Vert \boldsymbol{y}-\boldsymbol{X\beta }\right\Vert
		_{2}^{2}+\lambda \left\Vert \boldsymbol{\beta }\right\Vert _{1}\right) ,
		\label{1.2}
	\end{equation}%
	where $\left\Vert \boldsymbol{\beta }\right\Vert_{1}=\sum\limits_{i=1}^{p}\left\vert \beta _{i}\right\vert $ is the $\ell_{1}$-norm, 
	$\left\Vert y-\boldsymbol{X\beta }\right\Vert_{2}^{2}=\sum\limits_{i=1}^{n}\left( y_{i}-\boldsymbol{x}_{i}^{T}\boldsymbol{\beta }\right)^{2}$ 
	is the least-squares loss ($\ell _{2}$-norm), and $\lambda > 0$ is a (penalty) regularization parameter. This estimator is proved to perform  variable selection as some components of $\widehat{\boldsymbol{\beta}}$ become exactly zero depending on the choice of $\lambda$. 
	
	Let $\boldsymbol{\beta}_0 = (\beta_{01}, \ldots, \beta_{0p})^T$  be the true
	parameters in the LRM  (\ref{1.1}). We denote by $S_0=\left\{ j : \beta_{0j} \neq 0\right\} $ the subset of true non-zero coefficients and 
	 assume that the cardinality of $S_0$ is $s_0 \ll n$ (sparsity condition). 
	Let $\widehat{\boldsymbol{\beta }}$ be an estimator of $\boldsymbol{\beta },$ we say that a procedure (or estimator) enjoys variable selection consistency, 
	if $\left\{ j: \widehat{\beta }_{j}\neq 0\right\} =S_0$ with probability tending to one. It is well known that the LASSO tends to  provide biased estimate of regression coefficients and also requires strong conditions (see, e.g., \cite{buhlmann}) for model selection consistency, which often do not hold in practice \cite{fan1, zou}. 
	
	To overcome the drawback with the oracle properties,  Fan and Li  \cite{fan1} proposed the use of some general non-concave penalty function, specifically the SCAD penalty, instead of  LASSO. However, such non-concave penalties increase the computational burden significantly as the dimension ($p$) increases. Alternatively, Zou  \cite{zou} introduced a computationally faster adaptive version of the LASSO estimator, which assigns different weights to different coefficients within the $\ell_1$-penalty. 
	Precisely, the adaptive LASSO estimator of $\boldsymbol\beta$ in the LRM (\ref{1.1}) is defined as

	\begin{equation}
		\widehat{\boldsymbol{\beta }}_{\mathrm{AL}}=\widehat{\boldsymbol{\beta }}_{%
			\mathrm{AL}}(\lambda )=\arg \min_{\boldsymbol{\beta }\in \boldsymbol{\ 
				\mathbb{R}
			}^{p}}\left( \left\Vert \boldsymbol{y}-\boldsymbol{X\beta }\right\Vert
		_{2}^{2}+\lambda \mathop{\textstyle \sum }_{j=1}^{p}\frac{|\beta_j|}{|\widetilde{\beta}_j| + \delta_n\operatorname{I}(\widetilde{\beta}_j = 0)}\right),
		\label{1.21}
	\end{equation}
	where $\widetilde{\beta }_{j}$ is any (initial) consistent estimator of $\beta _{j}$ for $j=1, \ldots, p$, black and  $\delta_n$ is a  very small positive number chosen to avoid division by zero.
	Note that, for any fixed $\lambda $, the components having nonzero initial estimates get relatively  lower penalty than the zero (initial) components (for which the penalty weight goes to infinity). Consequently, the adaptive LASSO estimator is able to reduce the estimation bias and improves
	variable selection accuracy.
	For fixed $p,$  Zou  \cite{zou} proved that the adaptive LASSO has the
	oracle property and for $p>n,$ Huang et al. \cite{Huang} established that, under the partial orthogonality and certain other conditions, 
	the adaptive LASSO estimators are consistent and also efficient when the marginal regression estimators are used as the initial estimators. 
	The convexity of adaptive LASSO criterion ensures that these desired properties are global instead of local
	and further this estimator can be easily computed using the same efficient algorithms 
	that are used for LASSO, namely the  least angle regression (LARS Efron et al. \cite{efron}).
	
	However, the above methods are based on the squared-error loss function whose lack of robustness is well known. 
	Outlying values of $\boldsymbol{x}_{i}$ (leverage point) or extreme values of $\left(\boldsymbol{x}_{i},Y_{i}\right)$ (influence points) 
	have significant influence in the regularization procedures based on $\ell _{2}$ loss function. 
	For this purpose, the penalized M-estimators have been developed, which replaces the loss function by a convex function of the form
	$
	\mathop{\textstyle \sum }\limits_{i=1}^{n}\rho \left( y_{i} - \boldsymbol{x}_{i}^{T}\boldsymbol{\beta }\right).  
	$
	Wang et al. \cite{wang} proposed the least absolute deviation (LAD) loss with  
	$\rho \left( x\right) =\left\vert x\right\vert $ along with the LASSO penalty (LAD-LASSO method). 
	This LAD loss does not work well for small errors as it strongly penalizes the small residuals 
	(Owen \cite{Owen} and Lambert-Lacroix and Zwald \cite{La}). 
	Lambert-Lacroix and Zwald \cite{La} alternatively  proposed using the so called Huber loss, given by
	\begin{equation*}
		\rho \left( x\right) =\left\{ 
		\begin{array}{cc}
			x^{2} & \left\vert x\right\vert \leq M \\ 
			2M\left\vert x\right\vert -M^{2} & \left\vert x\right\vert >M%
		\end{array}%
		\right. .
	\end{equation*}%
	This function is quadratic in small values of $x$ but grows linearly for large values of $x$, and a data-adaptive choice of $M$ was discussed in Fan et al. \cite{Fan2017}. 
	Robust LASSO with Huber's loss function (Huber \cite{Huber}) 
	is robust to outliers in the response variable but not  for leverage points  
	(outliers in the covariates).  Arslan \cite{ar} presented a weighted version of the LAD-LASSO method which is also robust with respect to the leverage points. Other approaches that are robust with respect to leverage point were presented in Khan et al. \cite{khan}, Li et al. \cite{Li} and Alfons et al. \cite{alfons}. In Chang et al. \cite{Le}, 
	a penalized method is presented using Tukey's biweight loss that is resistant to outliers in both the response and the covariates.
	The theory of adaptive LASSO along with the quantile regression loss function has been developed in Fan et al. \cite{Fan/etc:2014}
	for the ultra-high dimensional set-up.  Qi et al. \cite{qi} considered, to gain robustness and efficiency simultaneously, 
	a data-driven convex combination of adaptive LASSO  and LAD-LASSO methods. 
	
	Recently,  Zhang et al. \cite{zhang} and Ghosh and Majumdar  \cite{ghosh2} have considered a robust loss function  
	based on the density power divergence (DPD) of   Basu et al. \cite{basu0}  along with grouped LASSO penalty  
	and the general class of non-concave penalty functions, respectively. 
	This DPD based loss function is seen to provide significantly improved performance, with better trade-offs between efficiency and robustness,	under classical low-dimensional set-ups  (see, e.g., Basu et al. \cite{basu1}, Ghosh and Basu \cite{ghosh1}, among many others)	as well as in high-dimensional set-up with non-concave penalties (Ghosh and Majumdar \cite{ghosh2}).
	In brief, the great properties of the DPD that led us to choose this particular loss function in the present paper include 
	its high-robustness with only a small loss of efficiency, bounded influence functions with decreasing gross-error sensitivities,
	simple interpretation as a intuitive robust generalization of the classical MLE and	the possibility to avoid the complications of non-parametric smoothing while still using a density based divergence; 	see Section \ref{SEC:def} for some brief descriptions and refer to Basu et al. \cite{basu1} for further details and examples.	Further, it has also been shown that the minimum DPD estimator of the regression coefficient can achieve 50\% asymptotic breakdown point 
	under appropriate conditions for the low-dimensional linear regression models and are robust with respect to both outliers and leverage points 	(Ghosh and Basu \cite{ghosh1}). 	Such strong robustness properties of the DPD based loss function are also seen to translate for high-dimensional regression models 	in Zhang et al. \cite{zhang}, Ghosh and Majumdar  \cite{ghosh2} and Ghosh \cite{ghosh}.

	Although the general non-concave penalized DPD-based procedure (DPD-ncv) of  Ghosh and Majumdar \cite{ghosh2} is seen to yield significantly improved robust 	and efficient performance  both in terms of variable selection and parameter estimation and satisfies the oracle variable selection property 	under appropriate conditions, the method appears to be computationally challenging in ultra-high dimensional set-ups with growing number of covariates. 	A typical way to address the computational challenges in ultra-high dimensional scenarios involves employing a two-step process that combines a screening procedure with variable selection. However, this approach sacrifices the inherent advantages of regularization methods, namely the simultaneous robust variable selection and parameter estimation, and may discard truly related variables due to sample contamination in the screening step (suitable robust screening method, e.g., \cite{GhoshThoresen2021}, may be used to avoid the effects of contamination in the screening stage but it would add own computational challenges over the usual simple screening). More importantly, screening methods may have a higher bias than regularization methods, affecting the asymptotic properties of the final parameter estimates if the screening procedure is employed before applying the regularized estimation and selection procedure.

	Therefore, the main purpose of the present paper is to develop a robust yet simpler theoretically sound  regularization procedure for ultra high-dimensional  LRM (\ref{1.1})  that would be able to produce competitive estimation and variable selection performances as the DPD-ncv procedure with lower computational costs. For this purpose, we would stick with the DPD-based loss function but considering an appropriately adaptive penalty function. 
	This is motivated by the fact that the classical adaptive procedures significantly reduce the computational burden compared to a non-concave penalized procedures.	So, in this paper, we combine the advantages from these two fronts, the adaptive penalty and the robust DPD-based loss function, to present a new robust adaptive regularization procedure under the ultra-high dimensional LRM which will have the desired oracle property. In the following, we summarize the main contributions of the present paper, highlighting its difference from the closely related work of \cite{ghosh2}, as follows: 

\begin{itemize}
	
	\item Firstly, the present work differs from the work of \cite{ghosh2} mainly in the choice of the penalty function, which greatly simplifies the assumptions required for theoretical derivations as well as the practical computations. 	We use a general class of data-driven adaptive penalty functions instead of non-concave penalties so the stochastic nature of the penalty made our procedures much simple	without significance compromise in the model selection and estimation performances. As a particular case of our proposed robust adaptive procedures, we additionally provide a computationally faster approximation to the DPD-ncv method of \cite{ghosh2}. 
	\item Secondly, although the methodological idea is a simple modification from \cite{ghosh2}, the theoretical derivations of the oracle consistency, asymptotic distributions and the robustness of the proposed procedure required extensive non-trivial extensions of the existing literature, which are the major contributions of the present paper. The required assumptions in our theoretical derivations are much weaker than those used in \cite{ghosh2}  for LASSO penalty	and, hence, they hold more easily in real-life applications.
	\item Finally, we have provided extensive empirical illustrations comparing the performance and computational costs	(run times) of our proposed procedures 	with the DPD-ncv from \cite{ghosh2} and other existing robust and non-robust model penalized procedures. 	It has been illustrated in  detail that our simple and computationally efficient proposal of adaptively penalized DPD based estimates of the regression coefficients performs well under data contaminations (competitive to the computationally extensive and
	complex DPD-ncv method of \cite{ghosh2}).
	\noindent
	Therefore, in terms of both theoretical and computational aspects, the present work provides a significant improvement over the existing work of \cite{ghosh2} by considering (data) adaptive penalty function in defining the objective function along with the robust DPD-based loss function. 
\end{itemize}

	At this point, we would like to emphasis that we are considering robustness against pre-specified parametric family of model distributions
	and not non-parametric robustness with fully unspecified model distribution. That is, we always assume a known model distribution,
	the error distribution in the LRM (\ref{1.1}) in the present paper, that is followed by the majority of the observed data
	and talk about robust estimation and variable selection against contamination in parts of the data (e.g., outliers). 
	Such parametrically robust procedures provide significantly improved efficiency compared to the fully non-parametric methods
	when a parametric model can be assumed for the majority of the data except possibly for a part of it being misspecified or containing outliers;	we refer readers to Basu et al.~\cite{basu1} for more details about this approach of robust and (highly) efficient parametric inference. 	Accordingly, our proposed estimation and variable section procedure would be dependent on the pre-assumed form of the error density in the  LRM (\ref{1.1});	we will first start the description of our proposal by assuming normal error density but subsequently extend the definitions for general error distributions from a location-scale family (where the density still have to be pre-specified) in Section \ref{SEC:def}. As the DPD loss function is already seen to provide better trade-off between efficiency and robustness compared to other such parametrically robust procedures	\cite{zhang, ghosh2, ghosh}, in the present paper, we mainly compare the performance of DPD-based loss function with different choices of penalty function to propose and justify a computationally simpler approximation to the general non-concave penalties as mentioned above.

	The rest of the paper is organized as follows: We start with the definition of the DPD loss function and the corresponding robust regularized estimator with adaptively weighted penalty functions in the LRM (Section \ref{SEC:def}). The robustness of the proposed adaptive procedure is examined theoretically via its influence function analysis (Section \ref{SEC:if}). 
	Further,  large sample oracle model selection property of our proposed estimator is derived along with its consistency and asymptotic normality  (Section \ref{SEC:oracle}). 	Subsequently, an efficient computational algorithm has been discussed (Section \ref{SEC:computation})	and the finite-sample performance of the proposal is illustrated through simulation studies, 	comparing its performance with that of other existing robust and non-robust regularization  methods (Section \ref{sec:simulation}). 
	Finally, the proposed method is applied to a real high-dimensional data set from the field of chemometrics, 	highlighting the advantage of our proposal in the presence of leverage points (Section \ref{SEC:realdata}). For brevity, the proofs of all the results, details of the computational algorithm 	and additional numerical results are presented in the Appendix, available in the Online Supplementary Material.    
	
	\section{The proposed estimation procedure based on the density power divergence loss and the	adaptive LASSO penalty}
	\label{SEC:def}
	
	The DPD family  represents a rich class of density-based divergences.  
	Given two densities $g$ and  $f$ with respect to some common dominating measure, the DPD between them is defined, as a
	function of a nonnegative tuning parameter $\gamma $, as 
	\begin{equation}
		d_{\gamma }(g,f)=\left\{ 
		\begin{array}{ll}
			\int \left\{ f^{1+\gamma }(x)-\left( 1+\frac{1}{\gamma }\right) f^{\gamma
			}(x)g(x)+\frac{1}{\gamma }g^{1+\gamma }(x)\right\} dx, & \text{for}\mathrm{~}%
			\gamma >0, \\[2ex] 
			\int g(x)\log \frac{g(x)}{f(x)} dx, & \text{for}\mathrm{~}%
			\gamma =0.%
		\end{array}%
		\right.  \label{1.5}
	\end{equation}
	The quantities defined in Equation (\ref{1.5}) are genuine divergences in the sense that 
	$d_{\gamma }(g,f)\geq 0$ for any densities $g$ and $f$ and all $\gamma \geq 0$, 
	and $d_{\gamma }(g,f)$ is equal to zero if and only if the densities $g$ and $f$ are identically equal.

	For a brief background, let us consider the parametric model family of densities $\{f_{\boldsymbol{\theta }}:%
	\boldsymbol{\theta }\in \Theta \subset {\mathbb{R}}^{p}\}$as the model for a population having true density $g$ and  
	true distribution function $G$. The minimum DPD
	functional $T_{\gamma }(G)$ at $G$ is  defined  as $d_{\gamma
	}(g,f_{T_{\gamma }(G)})=\min_{\boldsymbol{\theta }\in \Theta }d_{\gamma
	}(g,f_{\boldsymbol{\theta }})$. Clearly, the third  term $\int g^{1+\gamma }(x)dx$ has
	no role in the minimization of $d_{\gamma }(g,f_{\boldsymbol{\theta }})$
	over $\boldsymbol{\theta }\in \Theta $. Thus, the essential objective
	function for the minimum DPD functional $%
	T_{\gamma }(G)$ reduces to 
	\begin{equation}
		\int \left\{ f_{\boldsymbol{\theta }}^{1+\gamma }(x)-\left( 1+\frac{1}{%
			\gamma }\right) f_{\boldsymbol{\theta }}^{\gamma }(x)g(x)\right\} dx=\int f_{%
			\boldsymbol{\theta }}^{1+\gamma }(x)dx-\left( 1+\frac{1}{\gamma }\right)
		\int f_{\boldsymbol{\theta }}^{\gamma }(x)dG(x).
		\label{EQ:DPDF_obj}
	\end{equation}%
	Thus, given a random sample $y_{1},\ldots ,y_{n}$ from the distribution $G$, 
	we can estimate the above objective function by replacing $G$ with its empirical estimate $G_{n}$. 
	For a given tuning parameter $\gamma $, therefore, the minimum DPD estimator (MDPDE) 
	$\widehat{\boldsymbol{\theta }}_{\gamma }$ of $\boldsymbol{\theta }$ can be obtained by minimizing the objective function 
	\begin{equation}
		\begin{aligned}
			H_{n,\gamma }(\boldsymbol{\theta )} &= \int f_{\boldsymbol{%
					\theta }}^{1+\gamma }(x)dx-\left( 1+\frac{1}{\gamma }\right) \int f_{%
				\boldsymbol{\theta }}^{\gamma }(x)dG_{n}(x) 
			\\
			&=\int f_{\boldsymbol{\theta }}^{1+\gamma }(x)dx
			-\left( 1+\frac{1}{\gamma }%
			\right) \frac{1}{n}\sum_{i=1}^{n}f_{\boldsymbol{\theta }}^{\gamma }(y_{i}) 
		\end{aligned}  \label{1.6} 
	\end{equation}
	over $\boldsymbol{\theta }\in \Theta $. 
	Importantly, unlike many other divergence based estimation procedures, 
	the  minimization problem associated with the MDPDE does not require the use of a non-parametric density estimate. 
	
	Durio and Isaia \cite{da} extended the concept of the MDPDE to the problem of robust estimation in the LRM with $n> p$ (low dimensional settings).  
	For such classical LRM, $f_{\boldsymbol{\theta}}(y_i)$ is a normal density with mean $\boldsymbol{x}_i^T\boldsymbol{\beta}$
	and variance $\sigma^2$. It is a simple exercise to establish that, in this case 
	the corresponding DPD loss function  has the form
	\begin{equation}
		L_{n,\gamma }(\boldsymbol{\beta},\sigma)=\frac{1}{\left( 2\pi \right)
			^{\gamma /2}\sigma ^{\gamma }}\left( \frac{1}{\sqrt{\gamma +1}}-\frac{\gamma
			+1}{\gamma }\frac{1}{n}\sum\limits_{i=1}^{n}\exp \left\{ -\gamma \frac{%
			\left( y_{i}-\boldsymbol{x}_{i}^{T}\boldsymbol{\beta })\right) ^{2}}{2\sigma
			^{2}}\right\} \right) +\frac{1}{\gamma }.  \label{1.7}
	\end{equation}%
	The term $1/\gamma$ has been introduced to get the log-likelihood function as a limiting case as $\gamma \downarrow 0 $,	and therefore the corresponding MDPDE (at $\gamma=0$) is nothing but the usual maximum likelihood estimator (MLE).	When both $Y$ and $\mathbb{X}$ are random but we only assume the parametric model for the conditional distribution of $Y$ given $\mathbb{X}$, 	the loss function $L_{n,\gamma }(\boldsymbol{\beta},\sigma)$ in (\ref{1.7}) can be seen as an empirical estimate of the expectation (with respect to the unknown covariate distribution) of the DPD objective function (\ref{EQ:DPDF_obj}) 	between the conditional data and model densities of $Y$ given $\mathbb{X}$. 	However, under the fixed design setup with non-stochastic covariates $\boldsymbol{x}_{i}$, $i=1,...,n,$ the only random observations $y_{1},...,y_{n}$ are independent but non-homogeneous. 
	Ghosh and Basu \cite{ghosh1}  studied this problem and suggested to minimize the average DPD measure between 	the data and the model densities for each given covariate value. 	Interestingly, this approach also leads to the same loss function given in (\ref{1.7}). 
	Therefore the loss function $L_{n,\gamma }(\boldsymbol{\beta},\sigma)$ in (\ref{1.7}), referred to as the DPD loss function (with tuning parameter $\gamma $), 
	can be used to obtain robust MDPDE under the LRM (\ref{1.1}) with normal errors for both stochastic and fixed design matrices.

	In this paper, we first consider the penalized objective function constructed with the DPD-loss function and the standard adaptive lasso penalty from (\ref{1.21}), as given by
	\begin{equation}
		Q_{n,\gamma ,\lambda }(\boldsymbol{\beta},\sigma)=L_{n,\gamma }(\boldsymbol{\beta},\sigma)+ \lambda \sum_{j=1}^p \frac{|\beta_j|}{|\widetilde{\beta}_j| + \delta_n\operatorname{I}(\widetilde{\beta}_j = 0)},  \label{EQ:DPD_AL}
	\end{equation}
	where $\lambda$ is the usual regularization parameter and $\gamma>0$  is a robustness tuning parameter.
	As we will show later, we need this initial estimator $\widetilde{\beta }_{j}$ to be robust to achieve the robustness of the final estimator, 
	say $(\widehat{\boldsymbol{\beta}}_{\gamma,	\lambda}, \widehat{\sigma}_{\gamma, \lambda})$, 
	obtained by minimizing $Q_{n,\gamma ,\lambda }(\boldsymbol{\beta},\sigma)$. 
	We will refer to these final estimators $(\widehat{\boldsymbol{\beta}}_{\lambda,\gamma}, \widehat{\sigma}_{\lambda,\gamma})$ 
	as the adaptive DPD-LASSO (Ad-DPD-LASSO) estimator of $(\boldsymbol{\beta}, \sigma)$. 

	Note that, we consider the minimization of $Q_{n,\gamma ,\lambda }(\boldsymbol{\beta},\sigma)$ 
	with respect to both $\boldsymbol{\beta}$ and $\sigma$ for their simultaneous estimation; 
	the additional adaptive LASSO penalty on $\boldsymbol{\beta}$ helps sparse selection of the components of $\boldsymbol{\beta}$ 
	with oracle (variable selection) consistency. At $\gamma = 0$, 	the Ad-DPD-LASSO  leads to the non-robust penalized MLEs of	$\boldsymbol{\beta}$ and $\sigma$. Furthermore, for known $\sigma$, it leads to the classical adaptive LASSO  in (\ref{1.21}). Thus, for $\gamma>0$, the minimization of $Q_{n,\gamma ,\lambda}(\boldsymbol{\beta}, \sigma )$ provides 
	a generalization of the adaptive LASSO estimator $\widehat{\boldsymbol{\beta }}_{\mathrm{AL}}$ 
	with the additional benefit of robustness against data contamination. 
 
	Before starting the theoretical analysis of our proposal, 
	let us further extend its definition  by considering the more general class of adaptively weighted LASSO penalty  and a location-scale family of error distribution, as in  Ghosh and Majumdar \cite{ghosh2}. 	Letting the random errors $\epsilon_i$s in the model (\ref{1.1}) be independent and identically distributed 	having a location-scale density of the form $\frac{1}{\sigma} f\left(\frac{\epsilon}{\sigma}\right)$ with location $0$ and scale $\sigma$ ($f$ being an univariate distribution with mean 0 and variance 1), 	the corresponding DPD loss function is given by (generalizing from (\ref{1.7})) 
	\begin{align}
		L_{n, \gamma}(\boldsymbol{\beta}, \sigma) = \frac{1}{\sigma^{\gamma}}%
		\left(M_f^{(\gamma)} - \frac{1+\gamma}{\gamma} \frac{1}{n} \sum_{i=1}^n
		f^{\gamma}\left(\frac{y_i - \boldsymbol{x}_i^T\boldsymbol{\beta}}{\sigma}%
		\right)\right) + \frac{1}{\gamma},  \label{EQ:DPD_lossGen}
	\end{align}
	where $M_f^{(\gamma)} = \int f(\epsilon)^{1+\gamma}d\epsilon$ is assumed to exist finitely. 
	Then, the generalized adaptively weighted DPD-LASSO (AW-DPD-LASSO) estimators of $(\boldsymbol{\beta}, \sigma)$ 
	is the minimizer of
	\begin{equation}
		Q_{n,\gamma,\lambda }(\boldsymbol{\beta }, \sigma)=L_{n,\gamma }(\boldsymbol{%
			\beta }, \sigma) + \lambda\sum\limits_{j=1}^{p} w\left(\left\vert \widetilde{\beta}_{j}\right\vert\right)\left\vert \beta_{j}\right\vert,
		\label{EQ:DPD_AL_Gen}
	\end{equation}
	with $L_{n,\gamma }(\boldsymbol{\beta }, \sigma)$ being now given by (\ref{EQ:DPD_lossGen}) and $w$ is a suitable weight function. 	Note that, the general objective function (\ref{EQ:DPD_AL_Gen}) coincides with (\ref{EQ:DPD_AL})	for a hard-thresholding weight function $w(s) = (s + \delta_nI(s =  0))^{-1}$ for some very small positive number $\delta_n$	and standard normal error density. 	Although the weights are generally assumed to be stochastic depending on the initial estimators,	they could also be non-stochastic; for example, if $w(s)=1$ for all $s$, the general objective function (\ref{EQ:DPD_AL_Gen})	coincides with that of  the DPD-LASSO.
	
	In the rest of the paper, unless otherwise mentioned, we will denote by $L_{n, \gamma}(\boldsymbol{\beta},\sigma)$ 
	and $Q_{n,\gamma ,\lambda }(\sigma ,\boldsymbol{\beta })$ the generalized quantities in (\ref{EQ:DPD_lossGen}) and (\ref{EQ:DPD_AL_Gen}),
	respectively, and derive the theoretical results for the general class of AW-DPD-LASSO estimators under the ultra-high dimensional set-up.
	The simplification of all the results for the Ad-DPD-LASSO estimator (and some others) will also be provided as special cases.
	All the theoretical results presented here are valid for any general location-scale family of error distributions as long as it satisfies the associated assumptions.

	\begin{remark}[Connection with the Non-concave penalized DPD estimators]\label{REM:DPD-NCV}
		{{\rm   
				Given a penalty function $p_{\lambda}(\cdot)$, with  regularization parameter $\lambda$, 
				the DPD-ncv estimator is the minimizer of the penalized objective function 
				$L_{n,\gamma }(\boldsymbol{\beta}, \sigma) + \lambda\sum_{j=1}^{p} p_{\lambda}\left(| \beta_{j}|\right),$
				which is computationally challenging in higher dimensions. 
				However, in view of our general AW-DPD-LASSO estimator and its objective function (\ref{EQ:DPD_AL_Gen}),
				we can obtain an easily computable approximation of the DPD-ncv estimator. 
				Following the idea from Zou and Li \cite{Zou3} and Fan et al. \cite{Fan/etc:2014},
				 given a good initial estimator $\widetilde{\boldsymbol{\beta}}$,
				we can use the following approximation:
				$$
				p_{\lambda}(|\beta_j|) \approx p_{\lambda}(|\widetilde{\beta}_{j}|) 
				+ p_{\lambda}'(|\widetilde{\beta}_j|) (|\beta_j|- |\widetilde{\beta}_j|).
				$$
				Therefore, an AW-DPD-LASSO estimator with weight function $w = p_{\lambda}'$,
				which is much easier to compute even in ultra-high dimension, 
				is expected to work as a substitute for  the corresponding DPD-ncv estimator. 
				We will verify its performance both theoretically and empirically  in the subsequent sections for the popular SCAD penalty function, 
				for which 
				\begin{eqnarray}
					w(s) = p_{\lambda}'(s) =I(s\leq \lambda) + \frac{(a\lambda - s)_{+}}{(a-1)\lambda_n}I(s>\lambda), 
					\label{EQ:SCAD}
				\end{eqnarray}
				with $a>2$ being a tuning constant suggested from the common literature as $a=3.7$. We will refer to (\ref{EQ:SCAD}) as the SCAD weight function.
		}}
	\end{remark}

	\section{Robustness: Influence Function Analyses}
	\label{SEC:if}
	
	The influence function (IF) is a classical tool for measuring (local) robustness 	which indicates the possible asymptotic bias in the estimation functional due to an infinitesimal	contamination (Hampel et al. \cite{Hampel/etc:1986}). The concept has been extended	rigorously for the penalized estimators by  Avella-Medina \cite{Avella-Medina:2017}, 	and this extended IF has been further used in the high-dimensional context by Ghosh and Majumdar \cite{ghosh2}.	Here, we will derive the IF for our AW-DPD-LASSO estimators 	$(\widehat{\boldsymbol{\beta}}_{\gamma, \lambda}, \widehat{\sigma}_{\gamma, \lambda})$, 	to examine their theoretical (local) robustness against data contamination.
	
	In order to define the IF, we first extend the definition of the AW-DPD-LASSO estimator as a statistical functional.	Assuming the true joint distribution of $(Y, \mathbb{X})$, to be $G(y, \boldsymbol{x})$, 	the statistical functional $\boldsymbol{T}_{\gamma,\lambda}(G)=(\boldsymbol{T}_{\gamma,\lambda}^{\beta}(G), T_{\gamma, \lambda}^{\sigma}(G))$ 	corresponding to $(\widehat{\boldsymbol{\beta}}_{\gamma, \lambda}, \widehat{\sigma}_{\gamma, \lambda})$	is defined as the minimizer of 
	\begin{eqnarray}
		Q_{\gamma, \lambda}(\boldsymbol{\beta}, \sigma) = \int L_{\gamma}^\ast((y, 
		\boldsymbol{x});\boldsymbol{\beta}, \sigma) dG(y,\boldsymbol{x}) +
		\lambda\sum\limits_{j=1}^{p} w\left(\left\vert
		U_{j}(G)\right\vert\right)\left\vert \beta_{j}\right\vert,
		\label{EQ:penalDPD_lossFunc}
	\end{eqnarray}
	with respect to $\boldsymbol{\theta} = (\boldsymbol{\beta}, \sigma)$, where $%
	\boldsymbol{U}(G) = (U_1(G), \ldots, U_p(G))$ is the statistical functional
	corresponding to the initial estimator $(\widetilde{\beta }_{j})_{j=1,
		\ldots, p}$ and 
	\begin{eqnarray}
		L_{\gamma}^\ast((y, \boldsymbol{x});\boldsymbol{\beta}, \sigma) &=&  \frac{1%
		}{\sigma^{\gamma}}\left(M_f^{(\gamma)} - \frac{1+\gamma}{\gamma}
		f^{\gamma}\left(\frac{y - \boldsymbol{x}^T\boldsymbol{\beta}}{\sigma}%
		\right)\right) + \frac{1}{\gamma}.
	\end{eqnarray}
	It is straightforward that the objective function in (\ref{EQ:penalDPD_lossFunc})
	coincides with 
	the empirical objective function (\ref{EQ:DPD_AL_Gen}) when $G$ is substituted by 
	$G_n$, and
	hence $(\boldsymbol{T}_{\gamma,\lambda}^{\beta}(G_n), T_{\gamma,
		\lambda}^{\sigma}(G_n)) =(\widehat{\boldsymbol{\beta}}_{\gamma, \lambda}, 
	\widehat{\sigma}_{\gamma, \lambda})$.
	
	Note that, by definition, our AW-DPD-LASSO estimator also	belongs to the class of M-estimators considered in  Avella-Medina \cite{Avella-Medina:2017},	with their $L(Z, \theta)$ function coinciding with our $L_{\gamma}^\ast((Y, \boldsymbol{X});\boldsymbol{\theta})$. 	We will apply their  extended definition of IF  to derive the IF of our estimator. 
	The major problem for extending the theory of IFs  is the non-differentiability of the penalty function at zero 
	(for most common weight function including the one for adaptive LASSO); 
	Following Avella-medina \cite{Avella-Medina:2017}, we can rather consider a sequence of continuous and infinitely differentiable penalties, $\{p_{m, \lambda}(s, t(G))\}_{m\geq 1}$,	which converges to our adaptively weighted LASSO penalty $\lambda w\left(\left\vert t(G)\right\vert\right)\left\vert s\right\vert $ 	in the Sobolev space as $m\rightarrow\infty$; 
	note that they possibly depend on the functional $t(G) = {U}_j(G)$, the $j$-th component of the initial estimator $\boldsymbol{U}(G)$. 
	Correspondingly, we define the statistical functionals $\boldsymbol{T}_{m, \gamma,\lambda}(G) 
	= (\boldsymbol{T}_{m,\gamma,\lambda}^{\beta}(G), T_{m, \gamma, \lambda}^{\sigma}(G))$, for $m=1, 2, \ldots$, 
	as the minimizer of 
	\begin{eqnarray}
		Q_{m,\gamma, \lambda}(\boldsymbol{\beta}, \sigma) = \int L_{\gamma}^\ast((y, 
		\boldsymbol{x});\boldsymbol{\beta}, \sigma) dG(y,\boldsymbol{x}) +
		\sum\limits_{j=1}^{p} p_{m, \lambda}(\beta_j, U_j(G)),
		\label{EQ:penalDPD_lossFunc_m}
	\end{eqnarray}
	with respect to $\boldsymbol{\theta} = (\boldsymbol{\beta}, \sigma)$. 
	One can then define the IF of the AW-DPD-LASSO estimator, at the contamination points $(y_t, \boldsymbol{x}_t)$, as the limit of the IFs of 
	$\boldsymbol{T}_{m, \gamma,\lambda}(G)$ as 
	\begin{eqnarray}
		\mathcal{IF}({(y_t, \boldsymbol{x}_t)}, \boldsymbol{T}_{\gamma,\lambda}, G) =
		\lim_{m\rightarrow \infty}\mathcal{IF}({(y_t, \boldsymbol{x}_t)}, 
		\boldsymbol{T}_{m, \gamma, \lambda}, G).  \label{EQ:IF_DefLim}
	\end{eqnarray}

	Now, to compute the IFs of $\boldsymbol{T}_{m, \gamma,\lambda}(G)$ at the contamination points $(y_t, \boldsymbol{x}_t)$, 
	we consider its estimating equation obtained by equating the derivatives of $Q_{m,\gamma, \lambda}(\boldsymbol{\beta}, \sigma)$, 
	with respect to the parameters $\boldsymbol{\theta}=(\boldsymbol{\beta}, \sigma)$, to zero. 
	Then, by some standard calculations, these estimating equations can be simplified as 
	\begin{eqnarray}
		\left. 
		\begin{array}{lll}
			\frac{(1+\gamma)}{\sigma^{\gamma+1}} \int \psi_{1, \gamma}\left(\frac{y - 
				\boldsymbol{x}^T\boldsymbol{\beta}}{\sigma}\right)\boldsymbol{x}dG(y, 
			\boldsymbol{x}) & + \boldsymbol{P}_{m,\lambda}^{\ast}(\boldsymbol{\beta}, 
			\boldsymbol{U}(G)) & = \boldsymbol{0}_{p}, \\ 
			\frac{(1+\gamma)}{\sigma^{\gamma+1}} \int \psi_{2, \gamma}\left(\frac{y - 
				\boldsymbol{x}^T\boldsymbol{\beta}}{\sigma}\right)dG(y, \boldsymbol{x}) &  & 
			= 0.%
		\end{array}
		\right\}  \label{EQ:Est_Eqnm}
	\end{eqnarray}
	where $\psi_{1, \gamma}(s)=u(s)f^{\gamma}(s)$, $\psi_{2, \gamma}(s) = \{s
	u(s)+1\} f^{\gamma}(s) -\frac{\gamma}{{\gamma+1}} M_f^{(\gamma)}$, $%
	u=f^{\prime }/f$ with $f^{\prime }$ denoting the derivative of $f$ and $\boldsymbol{P}_{m,\lambda}^{\ast}(\boldsymbol{\beta}, \boldsymbol{U}(G))$ is
	a $p$-vector having $j$-th element as $\frac{\partial}{\partial\beta_j}p_{m,
		\lambda}(\beta_j, U_j(G))$. Now, we substitute the contaminated distribution 
	$G_{\epsilon}=(1-\epsilon)G+\epsilon\wedge_{(y_t, \boldsymbol{x}_t)}$, with $%
	\epsilon$ and $\wedge_{(y_t, \boldsymbol{x}_t)}$ being the contamination
	proportion and the degenerate contamination distribution at $(y_t, 
	\boldsymbol{x}_t)$, respectively, in place of $G$ in the above estimating
	equations (\ref{EQ:Est_Eqnm}) and differentiate with respect to $\epsilon$
	at $\epsilon=0$. Collecting terms after some algebra, and assuming all the
	relevant integrals exist finitely, we get the influence function of $%
	\boldsymbol{T}_{m, \gamma,\lambda}$ as given by 
	\begin{eqnarray}
		&& \mathcal{IF}({(y_t, \boldsymbol{x}_t)}, \boldsymbol{T}_{m, \gamma, \lambda}, G)
		= - (\boldsymbol{J}_{\gamma}^{\ast})^{-1} 
		\begin{bmatrix}
			\begin{array}{c}
				\hspace{-0.2cm} \frac{(1+\gamma)}{\sigma^{\gamma+1}}\psi_{1, \gamma}\left(\frac{y_t - 
					\boldsymbol{x}_t^T\boldsymbol{\beta}}{\sigma}\right)\boldsymbol{x}_t + \boldsymbol{P}^{\ast }_{m,\lambda}(\boldsymbol{\beta},\boldsymbol{U}(G))\\
				~~~~~~~~~~~~~~ +  \boldsymbol{P}_{m,\lambda}^{\ast(2)}(\boldsymbol{\beta}, \boldsymbol{U}(G))%
				\mathcal{IF}({(y_t, \boldsymbol{x}_t)}, \boldsymbol{U}, G) \\ 
				\hspace{-3cm}\frac{(1+\gamma)}{\sigma^{\gamma+1}}\psi_{2,\gamma}\left(\frac{y_t - 
					\boldsymbol{x}_t^T\boldsymbol{\beta}}{\sigma}\right)%
			\end{array}%
		\end{bmatrix}%
		.  \notag  \label{EQ:IF_DPDLASSOm}
	\end{eqnarray}
	where $\mathcal{IF}({(y_t, \boldsymbol{x}_t)}, \boldsymbol{U}, G)$ denote
	the IF of the initial estimator $\boldsymbol{U}$, $\boldsymbol{P}%
	_{m,\lambda}^{\ast(k)}(\boldsymbol{\beta}, \boldsymbol{U}(G))$ is a $p\times
	p$ diagonal matrix with $j$-th diagonal being the derivative of the $j$-th
	component of $\boldsymbol{P}_{m,\lambda}^{\ast}(\boldsymbol{\beta}, 
	\boldsymbol{U}(G))$ with respect to its $k$-th argument for $k=1,2$, and the 
	$(p+1)\times (p+1)$ matrix $\boldsymbol{J}_{\gamma}^{\ast} = \boldsymbol{J}_{\gamma}^{\ast}(G;\boldsymbol{%
		\beta}, \sigma) = \boldsymbol{J}_{\gamma}(G;\boldsymbol{\beta}, \sigma) + %
	\mbox{diag}\left\{\boldsymbol{P}_{m,\lambda}^{\ast(1)}(\boldsymbol{\beta}, 
	\boldsymbol{U}(G)), 0 \right\}$ with 
	\begin{align*}
		\boldsymbol{J}_{\gamma}(G;\boldsymbol{\beta}, \sigma) 
		=E_G\left[\frac{\partial^2L_{\gamma}^\ast((y, \boldsymbol{x});\boldsymbol{\beta}, \sigma) }{\partial\boldsymbol{\beta}\partial\boldsymbol{\beta}^T}\right]
		= - \frac{(1+\gamma)}{\sigma^{\gamma+2}} E_G%
		\begin{bmatrix}
			\begin{array}{cc}
				J_{11, \gamma}\left(\frac{y - \boldsymbol{x}^T\boldsymbol{\beta}}{\sigma}%
				\right)\boldsymbol{x}\boldsymbol{x}^T & J_{12, \gamma}\left(\frac{y - 
					\boldsymbol{x}^T\boldsymbol{\beta}}{\sigma}\right)\boldsymbol{x} \\ 
				J_{12, \gamma}\left(\frac{y - \boldsymbol{x}^T\boldsymbol{\beta}}{\sigma}%
				\right)\boldsymbol{x}^T & J_{22, \gamma}\left(\frac{y - \boldsymbol{x}^T%
					\boldsymbol{\beta}}{\sigma}\right)%
			\end{array}%
		\end{bmatrix}%
		, 
	\end{align*}
	where $J_{11, \gamma}(s) = \{\gamma u^2(s) + u^{\prime}(s)\}f^{\gamma}(s)$, $J_{12,
		\gamma}(s) = \{(1+\gamma) u(s) - \gamma s u^2(s) +su^{\prime \gamma}(s)$ and 
	$J_{22, \gamma}(s) = \{(1+\gamma)(1+2su(s)) + s^2 u^{\prime 2
	}u^2(s)\}f^{\gamma}(s) - \gamma M_f^{(\gamma)}$.
	Throughout this paper, we will assume that $f$ is such that 
	$J_{ii, \gamma}(s)>0$ for all $s$ and $i=1,2$.
	
	It has been shown in Proposition 1 of Avella-Medina \cite{Avella-Medina:2017} that under certain conditions including the existence of the above IF of $\boldsymbol{T}_{m, \gamma,\lambda}$, compactness of the parameter space $\Theta$ and the
	continuity of the relevant functions in $\boldsymbol{\theta}=(\boldsymbol{\beta}, \sigma)$, the limit of $\mathcal{IF}({(y_t, \boldsymbol{x}_t)},	\boldsymbol{T}_{m, \gamma,\lambda}, G)$ exists as $m\rightarrow\infty$ and	the limit is also independent of the choice of the differentiable penalty sequence $p_{m, \lambda}(s).$ Therefore, we can uniquely define the IF of our AW-DPD-LASSO estimator 
	$\boldsymbol{T}_{\gamma,\lambda}$ by \eqref{EQ:IF_DefLim} with any appropriate penalty sequence and the resulting
	IF will, in fact, be the distributional derivative of $\boldsymbol{T}_{\gamma,\lambda}(G_{\epsilon})$ with respect to $\epsilon$ at $\epsilon=0.$	So, we  take a particular choice of differentiable penalty functions as $p_{m, \lambda}(s, U_j(G)) =
	\lambda w\left(h_m(U_j(G))\right)h_m(s) $ where the infinitely	differentiable function 
	\begin{equation*}
		h_m(s) = \frac{2}{m}\log(e^{sm}+1) - s \rightarrow |s|, ~~~~\mbox{ as }
		m\rightarrow \infty. 
	\end{equation*}
	This particular choice of $p_m$ satisfies all the required regularity conditions for convergence, stated in Avella-Medina \cite{Avella-Medina:2017}, of the sequence of IFs  (of $\boldsymbol{T}_{m,\gamma,\lambda}$) to the unique limiting IF (of $\boldsymbol{T}_{\gamma, \lambda}$). Therefore, calculating the IFs of $\boldsymbol{T}_{m, \gamma,\lambda}$ with this
	particular penalty function for each $m$ and taking limit as $%
	m\rightarrow\infty$, we have derived the IF of our AW-DPD-LASSO estimator $\boldsymbol{T}_{\gamma,\lambda}$ which is presented in
	the following theorem.
	
	\begin{theorem}
		\label{THM:IFlim_MNPDPDE} Consider the above-mentioned set-up with the
		general error density $f$ and the true parameter value $\boldsymbol{\theta}^g = (\boldsymbol{\beta}^g, \sigma^g) = T_{\gamma, \lambda}(G)$,
		where $\boldsymbol{\beta}^g$ is \textit{sparse} with only $s (<n)$ non-zero
		components (recall $p>>n$). Without loss of generality, assume $%
		\boldsymbol{\beta}^g=(\boldsymbol{\beta}_1^{gT}, \boldsymbol{0}_{p-s}^T)^T$,
		where $\boldsymbol{\beta}_1^{g}$ contains all and only $s$-non-zero elements
		of $\boldsymbol{\beta}^{g}$. Let us denote $\boldsymbol{x}_1$, $\boldsymbol{x%
		}_{1,t}$ and $\boldsymbol{\beta}_1$ to be the $s$-vectors of the first $s$
		elements of the $p$-vectors $\boldsymbol{x}$, $\boldsymbol{x}_t$ and $%
		\boldsymbol{\beta}$, respectively, and the corresponding partition of our
		functional as $\boldsymbol{T}_{\gamma,\lambda}(G) = (\boldsymbol{T}%
		_{\gamma,\lambda}^{\beta_1}(G)^T, \boldsymbol{T}_{\gamma,\lambda}^{%
			\beta_2}(G)^T, {T}_{\gamma,\lambda}^{\sigma}(G))^T$. Then, whenever the
		associated quantities exist finitely,  the influence function of $%
		\boldsymbol{T}_{\gamma,\lambda}^{\beta_2}$ is identically zero at the true
		distribution $G$ and  that of $(\boldsymbol{T}_{\gamma,\lambda}^{\beta_1}, {T%
		}_{\gamma}^{\sigma})$ at $G$ is given by  
		\begin{align}
		&	\mathcal{IF}((y_t,  \boldsymbol{x}_t), (\boldsymbol{T}_{\gamma,\lambda}^{%
				\beta_1}, {T}_{\gamma,\lambda}^{\sigma}), G)  \notag =  \\
			& \hspace{1cm}-\boldsymbol{S}_{\gamma}^{-1}
			\begin{bmatrix}
				\begin{array}{c}
					\hspace{-0.3cm}\frac{(1+\gamma)}{(\sigma^g)^{\gamma+1}}\psi_{1, \gamma}\left(\frac{y_t - 
						\boldsymbol{x}_{1,t}^T\boldsymbol{\beta}^g}{\sigma^g}\right)\boldsymbol{x}_{1,t}
					+\lambda \boldsymbol{P}^{\ast }(\boldsymbol{\beta},\boldsymbol{U}(G))\\
					~~~~~~~~~~~~~~~~~~ + \lambda \boldsymbol{P}^{\ast(2)}(\boldsymbol{\beta}, \boldsymbol{U}(G))\mathcal{IF}({%
						(y_t, \boldsymbol{x}_t)}, \boldsymbol{U}^{(1)}, G) \\ 
					\hspace{-3.2cm} \frac{(1+\gamma)}{(\sigma^g)^{\gamma+1}}\psi_{2,\gamma}\left(\frac{y_t - 
						\boldsymbol{x}_{1,t}^T\boldsymbol{\beta}^g}{\sigma^g}\right) 
				\end{array}%
			\end{bmatrix}	,  \label{EQ:IF_MNPDPDE_Sparse}
		\end{align}
		where $\boldsymbol{U}^{(1)}$ is the first $s$ elements of $\boldsymbol{U}$,
		$\boldsymbol{P}^{\ast}(\boldsymbol{\beta}, \boldsymbol{U}(G))$ is an $s-$dimensional vector  having $j$-th element as $ w(|U_j(G)|)\operatorname{sign}(\beta_j)$ for $j=1, \ldots, s$  and
		$\boldsymbol{P}^{\ast(2)}(\boldsymbol{\beta}, \boldsymbol{U}(G))$ is an $s\times s$
		diagonal matrix  having $j$-th diagonal entry as $w^{\prime }(|U_j(G)|)%
		\mbox{sign}(U_j(G)\beta_j)$ for $j=1, \ldots, s$, with $w^{\prime }(s)$
		denoting the derivative of $w(s)$ in $s$,  and the $(s+1)\times(s+1)$ matrix 
		$\boldsymbol{S}_{\gamma}= \boldsymbol{S}_{\gamma}(G;\beta, \sigma)$ is defined as   
		\begin{align}
			\boldsymbol{S}_{\gamma}(G;\beta, \sigma) = - \frac{(1+\gamma)}{%
				\sigma^{\gamma+2}} E_G%
			\begin{bmatrix}
				\begin{array}{cc}
					J_{11, \gamma}\left(\frac{y - \boldsymbol{x}^T\boldsymbol{\beta}}{\sigma}%
					\right)\boldsymbol{x}_1\boldsymbol{x}_1^T & J_{12, \gamma}\left(\frac{y - 
						\boldsymbol{x}^T\boldsymbol{\beta}}{\sigma}\right)\boldsymbol{x}_1 \\ 
					J_{12, \gamma}\left(\frac{y - \boldsymbol{x}^T\boldsymbol{\beta}}{\sigma}%
					\right)\boldsymbol{x}_1^T & J_{22, \gamma}\left(\frac{y - \boldsymbol{x}^T%
						\boldsymbol{\beta}}{\sigma}\right)%
				\end{array}%
			\end{bmatrix}%
			. 
		\end{align}
	\end{theorem}
	
	For the LRM in (\ref{1.1}) with conditional density of $Y$ given $\mathbb{X}=\boldsymbol{x}$ given by $\frac{1}{\sigma}%
	f\left(\frac{y - \boldsymbol{x}^T\boldsymbol{\beta}}{\sigma}\right)$, let us
	denote the corresponding joint distribution  $ F_{(\boldsymbol{\beta},
		\sigma)}$, which also contains the marginal distribution (say $H$) of $%
	\boldsymbol{x}$ along with the above model conditional distribution. In this
	case, the matrix $\boldsymbol{S}_{\gamma}(G;\boldsymbol{\beta}, \sigma)$ can
	be further simplified as 
	\begin{equation*}
		\boldsymbol{S}_{\gamma}^{(0)}(G;\beta, \sigma) = - \frac{(1+\gamma)}{%
			\sigma^{\gamma+2}} 
		\begin{bmatrix}
			\begin{array}{cc}
				J_{11, \gamma}^{(0)}E_H\left(\boldsymbol{x}_1\boldsymbol{x}_1^T\right) & 
				J_{12, \gamma}^{(0)}E_H\left(\boldsymbol{x}_1\right) \\ 
				&  \\ 
				J_{12, \gamma}^{(0)}E_H\left(\boldsymbol{x}_1\right)^{T} & J_{22,
					\gamma}^{(0)}%
			\end{array}%
		\end{bmatrix}%
		, 
	\end{equation*}
	where 
	\begin{eqnarray}
		J_{11, \gamma}^{(0)} &=& \gamma M_{f,0,2}^{(\gamma)} +
		M_{f,0}^{(\gamma)\ast},  \notag \\
		J_{12, \gamma}^{(0)} &=& (1+\gamma) M_{f,0,1}^{(\gamma)} - \gamma
		M_{f,1,2}^{(\gamma)} + M_{f,1}^{(\gamma)\ast},  \notag \\
		\mbox{and }~~~~ J_{22, \gamma}^{(0)} &=& 2(1+\gamma)M_{f,1,1}^{(\gamma)} +
		M_{f,2}^{(\gamma)\ast} + \gamma M_{f,2,2}^{(\gamma)} + M_f^{(\gamma)}, 
		\notag
	\end{eqnarray}
	with $M_{f, i, j}^{(\gamma)} = \int s^i u(s)^j f(s)^{1+\gamma}ds$ and $M_{f,
		i}^{(\gamma)\ast} = \int s^i u^{\prime 1+\gamma}ds$ for $i,j = 0, 1, 2$. 
	\newline
	In particular, if the error density $f$ satisfies $J_{12, \gamma}^{(0)}=0$
	or $E(\boldsymbol{x}) =\boldsymbol{0}_p$, then we can separately write down
	the IFs of $\boldsymbol{T}_{\gamma,\lambda}^{\beta_1}$ and $%
	T_{\gamma,\lambda}^{\sigma}$ at the model $G=F_{(\boldsymbol{\beta},\sigma)}$
	from Theorem \ref{THM:IFlim_MNPDPDE} as given by 
	\begin{align}
		&\mathcal{IF}({(y_t, \boldsymbol{x}_t)}, \boldsymbol{T}_{\gamma,\lambda}^{%
			\beta_1}, F_{(\boldsymbol{\beta}, \sigma)})  \label{EQ:IF_MNPDPDE_beta} = \frac{\sigma}{J_{11, \gamma}^{(0)}} \left[E_H\left(\boldsymbol{x}_1%
		\boldsymbol{x}_1^T\right)\right]^{-1} \big[\psi_{1, \gamma}\left(\frac{y_t
			- \boldsymbol{x}_{1,t}^T\boldsymbol{\beta}_1}{\sigma}\right)\boldsymbol{x}_{1,t}  \\
		& \hspace{2cm} + \frac{\lambda\sigma^{\gamma+1}}{(1+\gamma)}\left(  \boldsymbol{P}^{\ast}(\boldsymbol{\beta}
		, \boldsymbol{\beta}) + \boldsymbol{P}^{\ast (2)}(\boldsymbol{\beta}%
		, \boldsymbol{\beta}) \mathcal{IF}({(y_t, \boldsymbol{x}_t)}, \boldsymbol{U}%
		^{(1)}, F_{(\boldsymbol{\beta},\sigma)})\right)\big],  \notag \\
		&\mathcal{IF}({(y_t, \boldsymbol{x}_t)}, \boldsymbol{T}_{\gamma,\lambda}^{%
			\sigma}, F_{\boldsymbol{\theta}}) = \frac{\sigma}{J_{22, \gamma}^{(0)}}%
		\psi_{2,\gamma}\left(\frac{y_t - \boldsymbol{x}_{1,t}^T\boldsymbol{\beta}_1}{\sigma%
		}\right),  \label{EQ:IF_MNPDPDE_sigma}
	\end{align}
	since $(\boldsymbol{\beta}_1^g, \sigma^g) = (\boldsymbol{\beta}_1, \sigma)$ at  
	$G=F_{(\boldsymbol{\beta},\sigma)}$ and $\boldsymbol{U}(F_{(\boldsymbol{\beta},\sigma)})=\boldsymbol{\beta}$ by its consistency.
	
	It is clear from the above formulas of the IF of the AW-DPD-LASSO estimator that it is expected to be robust having
	bounded IF if the quantities $\psi_{1, \gamma} ((y_t - 
	\boldsymbol{x}_{1,t}^T \boldsymbol{\beta}_1)/\sigma ) \boldsymbol{x}_{1,t}$, $%
	\psi_{2, \gamma} ((y_t - \boldsymbol{x}_{1,t}^T \boldsymbol{\beta}_1)/\sigma )$
	and $\mathcal{IF}({(y_t, \boldsymbol{x}_t)}, \boldsymbol{U}^{(1)}, F_{(%
		\boldsymbol{\beta},\sigma)})$ are bounded in either or both of the contamination point $%
	(y_t, \boldsymbol{x}_t)$. The first two quantities depend directly on the
	model assumption on the error density $f$ and the tuning parameter $\gamma$;
	for most common densities having exponential structure they can be seen to
	be bounded in $(y_t, \boldsymbol{x}_t)$ for any $\gamma>0$. The last
	quantity, namely the IF of the initial estimator $\boldsymbol{U}$, needs
	also to be bounded and a robust starting point is needed for the generalized adaptive DPD-LASSO. One
	possible choice could be the DPD-LASSO estimator 
	or other robust M-estimators of regression coefficient from the existing literature.

	\begin{remark}
		If the weights are fixed (non-stochastic) in the definition of AW-DPD-LASSO estimator,
		there is no involvement of $\boldsymbol{U}(G)$ in (\ref{EQ:Est_Eqnm}). 
		Hence the corresponding IF will not depend on the IF of $\boldsymbol{U}$
		but will coincide with the results of Ghosh and Majumdar \cite{ghosh2}.
	\end{remark}

	Finally, we simplify Theorem \ref{THM:IFlim_MNPDPDE} to obtain the	influence function of the Ad-DPD-LASSO estimator, defined as a
	minimizer of the simpler objective function in (\ref{EQ:DPD_AL}),	where the error density $f$ is the standard normal density (so that $J_{12,	\gamma}^{(0)}=0$) and $w(s) = (s + \delta_nI(s = 0))$ (so that $w^{\prime}=-s^{-2}$ whenever $s\neq 0$). 
	The simplified results are presented in the following corollary; 	clearly the IF of the Ad-DPD-LASSO estimator is bounded in the contamination points for all $\gamma>0$ indicating the claimed	robustness of our proposal against (infinitesimal) data contamination in both the response and covariates spaces, provided the initial estimator is chosen robustly.
	
	\begin{corollary}
		\label{CORR:IFlim_MNPDPDE} Consider the Ad-DPD-LASSO estimation by the	minimization of the  objective function  in (\ref{EQ:DPD_AL})
		for the LRM in (\ref{1.1}) under normal error density.	 Suppose that true distribution $F_{(\beta, \sigma)}$ underlying the data satisfies the LRM and the true regression	coefficient $\boldsymbol{\beta}$ is \textit{sparse} with only $s (<n)$
		non-zero components (but $p>>n$). Without loss of generality,	assume $\boldsymbol{\beta}=(\boldsymbol{\beta}_1^{T}, \boldsymbol{0}_{p-s}^T)^T$, where $\boldsymbol{\beta}_1$ contains all and only $s$-non-zero elements of $\boldsymbol{\beta}$. Also assume that the initial estimator is consistent in the sense $\boldsymbol{U}(F_{(\boldsymbol{\beta},\sigma)})=\boldsymbol{\beta}$ and denote $\boldsymbol{P}^{\ast}_0(\boldsymbol{\beta}	) = \left(|\beta_1|^{-1}, \ldots, |\beta_s|^{-1}\right)$ and $\boldsymbol{P}^{\ast(2)}_0(\boldsymbol{\beta}) = \mbox{diag}\{\beta_1^{-2}, \ldots, \beta_s^{-2}\}$. If the functional
		corresponding to the Ad-DPD-LASSO estimator is partitioned as $\boldsymbol{T}_{\gamma,\lambda}(G) = (\boldsymbol{T}_{\gamma,\lambda}^{\beta_1}(G)^T, \boldsymbol{T}_{\gamma,\lambda}^{\beta_2}(G)^T, {T}_{\gamma,\lambda}^{\sigma}(G))^T$, with $\boldsymbol{T}_{\gamma,\lambda}^{%
			\beta_1}$ being of length $s$, then their influence functions at the model $F_{(\boldsymbol{\beta}, \sigma)}$ are given by 
		\begin{align}
			&\mathcal{IF}({(y_t, \boldsymbol{x}_t)}, \boldsymbol{T}_{\gamma,\lambda}^{%
				\beta_1}, F_{(\boldsymbol{\beta}, \sigma)}) = - \sigma(\gamma+1)^{3/2} \left[%
			E_H\left(\boldsymbol{x}_1\boldsymbol{x}_1^T\right)\right]^{-1} \left[{\left({%
					y_t - \boldsymbol{x}_t^T\boldsymbol{\beta}}\right)} e^{-\frac{%
					\gamma\left(y_t - \boldsymbol{x}_t^T\boldsymbol{\beta}\right)^2}{2\sigma^2}} 
			\boldsymbol{x}_{1,t} \right.  \notag \\
			& ~~~~~~~~~~~~~~~~~~~~~~~~~~~ \left.+ \frac{\lambda\sigma^{2%
					\gamma+3}(2\pi)^{\gamma/2}}{(1+\gamma)} \left(\boldsymbol{P}^{\ast}_0(\boldsymbol{\beta}) +\boldsymbol{P}^{\ast(2)}_0(\boldsymbol{\beta}) 
			\mathcal{IF}({(y_t, \boldsymbol{x}_t)}, \boldsymbol{U}^{(1)}, F_{(%
				\boldsymbol{\beta},\sigma)})\right)\right],  \notag \\
			&\mathcal{IF}({(y_t, \boldsymbol{x}_t)}, \boldsymbol{T}_{\gamma,\lambda}^{%
				\beta_2}, F_{(\boldsymbol{\beta}, \sigma)})=0,  \notag \\
			&\mathcal{IF}((y_t, \boldsymbol{x}_t), \boldsymbol{T}_{\gamma,\lambda}^{%
				\sigma}, F_{\boldsymbol{\theta}}) = \frac{\sigma(1+\gamma)^{5/2}}{%
				2+\gamma^2} \left[\left(1 - \left(\frac{y_t - \boldsymbol{x}_t^T%
				\boldsymbol{\beta}}{\sigma}\right)^2\right) e^{-\frac{\gamma\left(y_t - 
					\boldsymbol{x}_t^T\boldsymbol{\beta}\right)^2}{2\sigma^2}} - \frac{\gamma}{%
				(1+\gamma)^{1/2}}\right],  \notag
		\end{align}
		whenever the associated quantities exist finitely. Here, all notations
		are the same as in Theorem \ref{THM:IFlim_MNPDPDE}. 
	\end{corollary}
	
	\begin{remark}
		Note that, at $\gamma\downarrow 0$, the Ad-DPD-LASSO
		coincides with the usual adaptive LASSO and hence the above results also provide
		its influence function which is new in the literature of
		adaptive LASSO. The IF at $\gamma=0$ is unbounded at any	contamination point, even if we start with a robust initial estimator,
		indicating the non-robust nature of the usual adaptive LASSO against data 	contamination.
	\end{remark}

	\section{Oracle Consistency}\label{SEC:oracle}
	
	We now  study the asymptotic properties of the AW-DPD-LASSO estimators under 
	the ultra-high dimensional set-up of non-polynomial order, following  Fan et al. \cite{Fan/etc:2014}. 
	With the notation of Sections \ref{SEC:intro} and \ref{SEC:def}, recall that,
	the number $p$ of the available covariates is assumed to grow exponentially with the sample size $n.$
	However, only few of them are significantly (linearly) associated with the response under the true model	so that the true value $\boldsymbol{\beta}_0=(\beta_{10}, \ldots, \beta_{p0})$ of the regression coefficient is sparse having only 
	$s\ll n$ non-zero entries. Without loss of generality,   $\boldsymbol{\beta}_0 = (\boldsymbol{\beta}_{10}, \boldsymbol{0}_{p-s})^T$ 
	and then $S_0  = \{1, 2, \ldots, s \}$.
	Let us allow $s=s_n=o(n)$ to slowly diverge with the sample size $n$, but the subscript will be suppress unless required to avoid any confusion.
	The oracle property refers to the fact that any estimator can correctly identify this true model $S_0$, 
	i.e., the first $s$ components of the estimated regression coefficient vector are consistent estimators of the components of $\boldsymbol{\beta}_{10}$
	whereas the remaining components are zero asymptotically  with probability tending to one.
	We will now show that, under certain conditions, the general AW-DPD-LASSO estimator of $\boldsymbol{\beta}$ 
	enjoys the oracle property and subsequently simplify the required conditions for the Ad-DPD-LASSO estimator.

	For simplicity, we here assume that the design matrix $\boldsymbol{X}$ is fixed with each column being standardized to have $\ell_1$-norm $\sqrt{n}$
	and the response is also standardized so that the error variance $\sigma^2$ is assumed to be known and equal to one. 
	The case of unknown $\sigma^2$ can be tackled by similar arguments with slightly modified assumptions as described later on;
	note that the objective function is convex in $\sigma$ and hence its minimizer can be shown to be consistent and asymptotically normal 
	through standard arguments and will be independent of the penalty used (see Ghosh and Majumdar \cite{ghosh2}).
	Further, in consistence with the high-dimensional literature, we will assume that the regularization parameter $\lambda=\lambda_n$ 
	in our objective function (\ref{EQ:DPD_AL_Gen}) depends of the sample size $n$ 
	but their explicit relation is given later on following the required assumptions.  
	In the following, given any $S\subseteq\{1, 2, \ldots, p\}$ 
	and any $p$-vector $\boldsymbol{v}=(v_1, \ldots, v_p)^T$, 
	we will denote $S^c = \{1, 2, \ldots, p\} \setminus S$, 
	$\boldsymbol{v}_S=(v_j : j\in S)$ and $\boldsymbol{v}_{S^c} = (v_j : j \notin S)$
	whereas $Supp(\boldsymbol{v}) = \{ j : v_j \neq 0 \}$.
	Note that, for the true model $S_0$, we have $\boldsymbol{\beta}_{0S_0} =\boldsymbol{\beta}_{10}$, $\boldsymbol{\beta}_{0S_0^c} =\boldsymbol{0}_{p-s}$
	and $Supp(\boldsymbol{\beta}_0) = S_0$.
	Further, let $\boldsymbol{X}_S$ consists of the $j$-th column of $\boldsymbol{X}$ for all  $j\in S$ for any $S$, 
	and put $\boldsymbol{X}_1=\boldsymbol{X}_{S_0}$ and $\boldsymbol{X}_2=\boldsymbol{X}_{S_0^c}$ 
	(so that $\boldsymbol{X}=[\boldsymbol{X}_1 : \boldsymbol{X}_2]$);
	the corresponding partition of the $i$-th row of $\boldsymbol{X}$ would be denoted by $\boldsymbol{x}_i = (\boldsymbol{x}_{1i}^T, \boldsymbol{x}_{2i}^T)^T$. 
	Also define 
	\begin{eqnarray}
		\boldsymbol{H}_{\gamma}^{(1)}(\boldsymbol{\beta}) &=& ((1+\gamma)\psi_{1, \gamma}(y_i - \boldsymbol{x}_{i}^T\boldsymbol{\beta}) : i=1, \ldots, n)^T
		\nonumber\\
		~~~~\mbox{and}~~ 
		\boldsymbol{H}_{\gamma}^{(2)}(\boldsymbol{\beta}) &=& \mbox{Diag}\{ (1+\gamma)J_{11, \gamma}(y_i - \boldsymbol{x}_{i}^T\boldsymbol{\beta}) : i=1, \ldots, n\}.
	\end{eqnarray}
	
	Since the AW-DPD-LASSO estimator of $\boldsymbol{\beta}$ coincides with the least-squares adaptive LASSO estimator at $\gamma=0$,
	we here focus on deriving their asymptotic properties for $\gamma>0$ only. 
	Under a general error distribution $f$ and a fixed $\gamma>0$,
	we need the following basic assumptions on the corresponding DPD loss function $L_{n,\gamma}(\boldsymbol{\beta}) = L_{n,\gamma}(\boldsymbol{\beta}, 1)$
	in (\ref{EQ:DPD_lossGen}) along with the boundedness of  the design matrix $\boldsymbol{X}=((x_{ij}))_{i=1, \ldots, n; j = 1, \ldots, p}$.
	All our assumptions and results are given in terms of a fixed design matrix, 
	but they can be easily extended for the random design matrix by showing that the required assumptions holds for the random design 
	asymptotically with probability one.
	
	\begin{itemize}
		\item[(A1)] The error density $f$ is such that $f^{\gamma}$ is Lipschitz with the Lipschitz constant $L_{\gamma}$.
		\item[(A2)]  The eigenvalues of the matrix 
		$n^{-1}(\boldsymbol{X}_1^T\boldsymbol{X}_1)$ are bounded below and above by positive constants $c_0$ and $c_0^{-1}$, respectively.
		Also $ \kappa_n := \max_{i,j}|x_{ij}| = o(n^{1/2}s^{-1}).$
	\end{itemize}

	Note that Assumption (A1) is implied by the boundedness of the function $\psi_{1, \gamma}(s)$ whenever it is differentiable
	and this holds for most common exponential family of distributions at any $\gamma>0$;
	in particular, it holds for the usual normal error distribution. 
	The second assumption, on the other hand, is the same as used in Fan et al. \cite{Fan/etc:2014};
	the first part is quite standard in high-dimensional literature whereas the second part holds for appropriate fixed design matrix
	as well as for common stochastic designs with asymptotic probability one (see Fan et al. \cite{Fan/etc:2014}  for some example).
	
	\subsection{General AW-DPD-LASSO estimator with Fixed Non-stochastic Weights}
	\label{SEC:Oracle_fixedW}
	
	We first consider the AW-DPD-LASSO estimators with fixed non-stochastic weights in the corresponding objective function in (\ref{EQ:DPD_AL_Gen}).
	With $\sigma=1$, the objective function in $\boldsymbol{\beta}$ can now be re-expressed as 
	$$
	Q_{n,\gamma,\lambda }(\boldsymbol{\beta })=L_{n,\gamma }(\boldsymbol{\beta }) + \lambda_n\sum\limits_{j=1}^{p} w_{j}\left\vert \beta_{j}\right\vert,
	$$
	where $w_j$ is the fixed weights corresponding to the $j$-th penalty term and the DPD loss $L_{n,\gamma }(\boldsymbol{\beta })$
	has the form $L_{n, \gamma}(\boldsymbol{\beta}) = \frac{1}{n} \sum_{i=1}^n \rho_{\gamma}(\boldsymbol{x}_i^T\boldsymbol{\beta}, y_i)$ with 
	\begin{align}
		\rho_{\gamma}(\boldsymbol{x}_i^T\boldsymbol{\beta}, y_i) = M_f^{(\gamma)} - \frac{1+\gamma}{\gamma} 
		f^{\gamma}\left({y_i - \boldsymbol{x}_i^T\boldsymbol{\beta}}\right) + \frac{1}{\gamma}.  
		\label{EQ:DPD_lossGen_beta}
	\end{align}
	Note that $\nabla L_{n,\gamma }(\boldsymbol{\beta }) = n^{-1}\left[\boldsymbol{X}^T\boldsymbol{H}_{\gamma}^{(1)}(\boldsymbol{\beta})\right]$
	and $\nabla^2 L_{n,\gamma }(\boldsymbol{\beta }) = n^{-1}
	\left[\boldsymbol{X}^T\boldsymbol{H}_{\gamma}^{(2)}(\boldsymbol{\beta})\boldsymbol{X}\right]$,
	where $\nabla$ and $\nabla^2$ denote the first and second order gradient with respect to $\boldsymbol{\beta}$, respectively.
	Let us also denote $\boldsymbol{w}=(w_1, \ldots, w_p)$, $\boldsymbol{w}_0=\boldsymbol{w}_{S_0}$, $\boldsymbol{w}_1=\boldsymbol{w}_{S_0^c}$,
	and define $\delta_n = \sqrt{{s(\log n)}/{n}} + {\lambda_n}||\boldsymbol{w}_0||_2$.
	
	We first study the properties of an oracle estimator obtained by minimizing the above objective function 
	$Q_{n,\gamma,\lambda }(\boldsymbol{\beta })$ with fixed weights over the restricted oracle parameter space 
	$\Theta^o= \big\{ \boldsymbol{\beta}=(\boldsymbol{\beta}_1^T, \boldsymbol{\beta}_2^T)^T\in \mathbb{R}^p : \boldsymbol{\beta}_2 = \boldsymbol{0}_{p-s} \big\}
	\equiv \mathbb{R}^s\times \{0\}^{p-s}$; let us denote the corresponding minimizer as 
	$\widehat{\boldsymbol{\beta}}^o = (\widehat{\boldsymbol{\beta}}_1^o, \boldsymbol{0}_{p-s})$, the oracle estimator for our model.
	
	We would like to point out that the oracle estimator defined above are exactly in line with the classical definitions from B\"{u}hlmann and van de Geer \cite{buhlmann}; it is also used by the pioneer paper by Fan and Li \cite{fan1} and the large pool of subsequent works build upon this paper. To see their alignments, note that the classical oracle estimate is the one  obtained based on a low dimensional sub-model of any high dimensional models which includes only the $s$ important covariates given to us by an oracle. Once it is assumed that we know the correct sub-model from Oracle, the parameter space can then be restricted only to a particular $s$-dimensional subspace of the whole parameter space corresponding to the $s$ important covariates. Assuming that the first $s$ covariates are important (without loss of generality), this oracle restricted parameter space is then given by our $\Theta^o$, and the estimate of the regression coefficient vector over this space can then be called as the oracle estimate, which only estimate the coefficients associated with $s$ important covariates (first $s$ in this case) and put zero for the other coefficient, i.e., $\widehat{\beta}^o$ as defined above. For good estimation procedures, such an oracle estimate should be close to the true parameter values (while considering estimating accuracy) and also any general estimate of the regression coefficients obtained over the whole parameter space would asymptotically be  very close to this oracle estimator with high-probability (while considering model selection accuracy); the later property is often known as the Oracle property in the literature. In the following, we will also prove the same oracle properties for our proposed adaptive DPD-based procedures under suitable assumptions.
	In particular,  we need the following additional assumption on the DPD loss function; from now on, all expectations are taken with respect to 
	the true model density $f$ of $\epsilon_i=y_i-\boldsymbol{x}_i^T\boldsymbol{\beta}_0$.
	
	\begin{itemize}
		\item[(A3)] The diagonal elements of $E[\boldsymbol{H}_{\gamma}^{(2)}(\boldsymbol{\beta}_0)]$ are all finite and 
		bounded from below by a constant $c_1>0$.
		\item[(A4)] Expectation of third order partial derivatives of  $\rho_{\gamma}(\boldsymbol{x}_i^T\boldsymbol{\beta}, y_i)$, $i=1, \ldots, n$, 
		with respect to all components of $\boldsymbol{\beta}_{S_0}$ are uniformly bounded in a neighborhood of $\boldsymbol{\beta}_{10}$. 
	\end{itemize}
	Note that these two assumptions are very common in the statistical inference using the DPD loss function even in low-dimensional 
	and they hold for most common (regular) models for the error; see Basu et al. \cite{basu1}, Ghosh and Basu \cite{ghosh1}. 
	Then, we have the following result about the $\ell_2$-consistency of the oracle estimator  $\widehat{\boldsymbol{\beta}}^o$.
	
	\begin{theorem}\label{THM:FixW_Oracle}
		If Assumptions (A1)--(A4) hold and $\lambda_n||\boldsymbol{w}_0||_2\sqrt{s}\kappa_n \rightarrow 0$,
		then, given any  constant $C_1>0$ and $\delta_n = \sqrt{{s(\log n)}/{n}} + {\lambda_n}||\boldsymbol{w}_0||_2$, 
		there exists some $c>0$ such that 
		\begin{eqnarray}
			P\left(\left|\left|\widehat{\boldsymbol{\beta}}_1^o - \boldsymbol{\beta}_{10}\right|\right|_2 \leq C_1\delta_n \right)\geq 1 - n^{-cs}.
		\end{eqnarray}
		Further, if $\delta_n^{-1}\min\limits_{1\leq j \leq s}|\beta_{j0}|\rightarrow \infty$, 
		then the sign of each component of $\widehat{\boldsymbol{\beta}}_1^o$ matches with that of $\boldsymbol{\beta}_{10}$.
	\end{theorem}
	
	Next we will show that the oracle estimator in Theorem \ref{THM:FixW_Oracle} is indeed an asymptotic  global minimizer of the 
	objective function $Q_{n,\gamma,\lambda }(\boldsymbol{\beta })$ over the whole parameter space 
	with probability tending to one. For this purpose, we need the following additional assumptions controlling the correlation between 
	the important and unimportant covariates in the same spirit as Condition 3 in Fan et al. \cite{Fan/etc:2014}.
	\begin{itemize}
		\item[(A5)]  $\left|\left|n^{-1}(\boldsymbol{X}_2^TE[\boldsymbol{H}_{\gamma}^{(2)}(\boldsymbol{\beta}_0)]\boldsymbol{X}_1)\right|\right|_{2, \infty}
		< \frac{\lambda_n\min(|\boldsymbol{w}_1|)}{2C_1\delta_n} $, 
		for some constant $C_1 >0$, where we denote 
		$||\boldsymbol{A}||_{2, \infty} =\sup\limits_{\boldsymbol{x}\in\mathbb{R}^q \setminus\{0\}} 
		\frac{||\boldsymbol{A}\boldsymbol{x}||_{\infty}}{||\boldsymbol{x}||_2}$
		for any $p\times q$ matrix $\boldsymbol{A}$ and $\min(|\boldsymbol{w}_1|)=\min\limits_{j>s}|w_j|$.
		
	\end{itemize}

	\begin{theorem}
		\label{THM:FixW_Main}
		Suppose that Assumptions (A1)--(A5) hold with $\lambda_n > 2\sqrt{(c+1)\log p / n}$ 
		and $\min(|\boldsymbol{w}_1|) >c_3$ for some constant $c, c_3>0$, and
		$$
		\lambda_n||\boldsymbol{w}_0||_2\kappa_n\max\{\sqrt{s}, ||\boldsymbol{w}_0||_2\} \rightarrow 0,
		~~~~
		\delta_ns^{3/2}\kappa_n^2(\log_2 n)^2 = o(n\lambda_n^2).
		$$
		Then, with probability at least $1 - O(n^{-cs})$, 
		there exists a global minimizer $\widehat{\boldsymbol{\beta}} = \left((\widehat{\boldsymbol{\beta}}_1^o)^T, \widehat{\boldsymbol{\beta}}_2^T\right)^T$ 
		of the AW-DPD-LASSO objective function $Q_{n,\gamma,\lambda }(\boldsymbol{\beta })$ such that
		$$
		\left|\left|\widehat{\boldsymbol{\beta}}_1^o - \boldsymbol{\beta}_{10}\right|\right|_2 \leq C_1\delta_n,
		~~~~\mbox{and }~~~~~
		\widehat{\boldsymbol{\beta}}_2 = \boldsymbol{0}_{p-s}.
		$$
	\end{theorem}

	Note that Theorem \ref{THM:FixW_Main} presents consistency of both the parameter estimates and oracle variable selection
	of our AW-DPD-LASSO estimators under appropriate assumptions. Next, we will derive the conditions under which 
	these estimators are asymptotically  normal after suitable normalizations.
	Let us define $\boldsymbol{V}_n = \big[\boldsymbol{X}_1^T E[\boldsymbol{H}_{\gamma}^{(2)}(\boldsymbol{\beta}_{0})] \boldsymbol{X}_1\big]^{-1/2}$,
	$\boldsymbol{Z}_n=\boldsymbol{X}_1\boldsymbol{V}_n$ and $\boldsymbol{\Omega}_n = Var\left[\boldsymbol{H}_{\gamma}^{(1)}(\boldsymbol{\beta}_{0})\right]$,
	and consider the following assumption.
	\begin{itemize}
		\item[(A6)]  $\boldsymbol{Z}_n\boldsymbol{\Omega}_nZ_n$ is positive definite, $\lambda_n ||\boldsymbol{w}_0||_2 = O(\sqrt{s/n})$,
		and $\sqrt{n/s} \min\limits_{1\leq j \leq s} |\beta_{j0}| \rightarrow \infty$.
	\end{itemize}
	
	These Assumptions (A5)-(A6) are similar to the conditions used in Fan et al. \cite{Fan/etc:2014}.
	while developing the properties of the adaptive quantile regressions. 
	Assumption (A5) provides the maximum correlation between the important and unimportant covariates 
	allowed for the proposed method to have oracle model selection consistency.
	This maximum allowed correlation further depends proportionally 
	on the minimum of the weights attached to the truly zero coefficients.
	Assumption (A6)  put further restrictions on these weights 
	ensuring that all the weights for unimportant covariates are penalized
	but those for the important covariates to be small enough depending on the rate of convergences of $\lambda_n$ and $s$.
	We now derive the following theorem on asymptotic distribution of the AW-DPD-LASSO.
		
	\begin{theorem}\label{THM:FixW_AsymNormal}
		Suppose that the assumptions of Theorem \ref{THM:FixW_Main} hold along with Assumption (A6).  
		Then, with probability tending to one, 
		there exists a global minimizer $\widehat{\beta} = (\widehat{\boldsymbol{\beta}}_1^o, \widehat{\boldsymbol{\beta}}_2)^T$ 
		of the AW-DPD-LASSO objective function $Q_{n,\gamma,\lambda }(\boldsymbol{\beta })$ such that
		$\widehat{\boldsymbol{\beta}}_2 = \boldsymbol{0}_{p-s}$ and 
		\begin{eqnarray}
			\boldsymbol{u}^T[\boldsymbol{Z}_n^T\boldsymbol{\Omega}_n\boldsymbol{Z}_n]^{-1/2} \boldsymbol{V}_n^{-1}
			\left[\left(\widehat{\boldsymbol{\beta}}_1^o - \boldsymbol{\beta}_{10}\right) + {n\lambda_n}\boldsymbol{V}_n^2\widetilde{\boldsymbol{w}_0}\right]
			\mathop{\rightarrow}^\mathcal{L} N(0,1),
			\label{EQ:AW-DPD-LASSO_AN}
		\end{eqnarray}
		for any arbitrary $s$-dimensional vector $\boldsymbol{u}$ satisfying $\boldsymbol{u}^T\boldsymbol{u} = 1$, where  
		$\widetilde{\boldsymbol{w}_0}$ is an $s$-dimensional vector with $j$-th element $w_j sign(\beta_{j0})$.
	\end{theorem}

	Assumption (A6) is crucial for the above asymptotic normality result of the AW-DPD-LASSO estimator,
	since it ensures that the bias term, namely ${n\lambda_n}\boldsymbol{V}_n\widetilde{\boldsymbol{w}_0}$,
	does not diverges asymptotically ensuring a legitimate asymptotic distribution of 
	$\left(\widehat{\boldsymbol{\beta}}_1^o - \boldsymbol{\beta}_{10}\right)$.
	It is controlled by the choice of weight and regularization sequence depending on $s$. 
	Further, the asymptotic variance of the AW-DPD-LASSO estimator depends crucially on the tuning parameter $\gamma>0$;
	it can be seen that that there is a slight increase in these asymptotic variances
	with increasing $\gamma>0$ in consistence with the classical theory of DPD based estimation  Basu et al. \cite{basu0}.

	As a special case of the above results, we get the properties of the DPD-LASSO estimators
	by using the weights $w_j = 1$ for all $j=1, \ldots, p$.  Then, we have $||\boldsymbol{w}_0||_2 = \sqrt{s}$,
	$\min(|\boldsymbol{w}_1|)=1>0$ and hence $\delta_n= \sqrt{s(\log n)/n} + \lambda_n\sqrt{s}$. 
	We can then obtain the consistency and 
	model selection property of the DPD-LASSO estimators, which is presented in the following corollary.

	\begin{corollary}[Properties of the DPD-LASSO Estimator]
		Under the set-up of this section, assume that Assumption (A1)-(A4) hold true.
		Then, for a given constant $C_1>0$, 
		we have the following results for the DPD-LASSO estimator with regularization parameter $\lambda_n=O(n^{-1/2})$.
		\begin{itemize}
			\item[a)] There exists some $c>0$ such that, with probability at least $1- n^{-cs}$,
			the corresponding oracle estimator $\widehat{\boldsymbol{\beta}}_1^o$ satisfies
			\begin{eqnarray}
				\left|\left|\widehat{\boldsymbol{\beta}}_1^o - \boldsymbol{\beta}_{10}\right|\right|_2 \leq 
				C_1\left[\sqrt{s(\log n)/n} + \lambda_n\sqrt{s}\right].
				\label{EQ:DPD-LASSO-1}
			\end{eqnarray}
			
			\item[b)] If $\sqrt{s(\log n)/n} = o\left(\min\limits_{1\leq j \leq s}|\beta_{j0}|\right)$, 
			then the sign of each component of $\widehat{\boldsymbol{\beta}}_1^o$ matches with that of the true $\boldsymbol{\beta}_{10}$.
			
			\item[c)] If Assumption (A5) holds with $\lambda_n > 2\sqrt{(c+1)\log p / n}$ for some constant $c>0$ and
			$n^{1/2}(\log_2 n)^{5/2} = O(\log p)$, then, with probability at least $1 - O(n^{-cs})$, 
			there exists a global minimizer $\widehat{\boldsymbol{\beta}} 
			= ((\widehat{\boldsymbol{\beta}}_1^o)^T, \widehat{\boldsymbol{\beta}}_2^T)^T$ 
			of the DPD-LASSO objective function such that
			$\widehat{\boldsymbol{\beta}}_1^o$ satisfies (\ref{EQ:DPD-LASSO-1}) and $\widehat{\boldsymbol{\beta}}_2 = \boldsymbol{0}_{p-s}$.
			
			\item[d)] In addition to the assumptions of item (c), 
			if 	$\sqrt{n/s} \min\limits_{1\leq j \leq s} |\beta_{j0}| \rightarrow \infty$
			and $\boldsymbol{Z}_n\boldsymbol{\Omega}_nZ_n$ is positive definite,
			then the DPD-LASSO estimator $\widehat{\boldsymbol{\beta}}_1^o$ obtained in item (c) 
			further satisfies (\ref{EQ:AW-DPD-LASSO_AN}), but the associated bias 
			${n\lambda_n}\boldsymbol{V}_n\widetilde{\boldsymbol{w}_0}$	is non-diminishing. 
		\end{itemize}
	\end{corollary}
	
	It is important to note that the good properties of the DPD-LASSO estimator 
	requires $\lambda_n=O(n^{-1/2})$, which is known to be quite low for a thresholding level 
	even for Gaussian noise Fan et al.  \cite{Fan/etc:2014}. This is in consistent with any LASSO estimators
	that motivated researchers to consider adaptive LASSO. In general, 
	the properties of the AW-DPD-LASSO estimators crucially depends 
	note only on the choice of weights $\boldsymbol{w}$ but also on 
	the structure of the oracle through different assumptions on $\boldsymbol{w}_0$ and $\boldsymbol{w}_1$.
	These are not practically feasible to satisfy with the nonadaptive fixed weights 
	and we need to estimates the weights adaptively from the data.
	However, the results developed in this section with fixed weights 
	will indeed be useful in understanding the behaviors of the AW-DPD-LASSO estimators 
	with stochastic (adaptive) weights as discussed in the next subsection.

	\subsection{General AW-DPD-LASSO estimator with Adaptive Weights}
	\label{SEC:Oracle_AdaptiveW}

	We now consider the AW-DPD-LASSO estimators with adaptive (stochastic) weights as in the objective function in (\ref{EQ:DPD_AL_Gen}), but with $\sigma=1$.
	Given an initial estimator of $\widetilde{\boldsymbol{\beta}}=(\widetilde{\beta}_1, \ldots, \widetilde{\beta}_p)^T$,
	we can now re-express the objective function in $\boldsymbol{\beta}$ as 
	\begin{eqnarray}
		\widehat{Q}_{n,\gamma,\lambda}(\boldsymbol{\beta})=L_{n,\gamma }(\boldsymbol{\beta }) + \lambda_n\sum\limits_{j=1}^{p} \widehat{w}_{j}
		\left\vert \beta_{j}\right\vert,
		\label{EQ:ObjFUnc_AdWeight}
	\end{eqnarray}
	where the loss $L_{n, \gamma}(\boldsymbol{\beta})$ are exactly as in Section \ref{SEC:Oracle_fixedW}
	but the weights are now estimated adaptively as $\widehat{w}_j = w\left(\left\vert \widetilde{\beta}_{,j}\right\vert\right)$ 
	for some suitable function $w(\cdot)$. We first study the property of the resulting estimator for general weight function 
	and initial estimator $\widetilde{\boldsymbol{\beta}}$ 
	and then simplify them for an important special case.

	In this section as well, we will continue to use the notation of Section \ref{SEC:Oracle_fixedW},
	and additionally denote $\widehat{\boldsymbol{w}}=(\widehat{w}_1, \ldots, \widehat{w}_p)^T$ and
	$\boldsymbol{w}^\ast = (w_1^\ast, \ldots, w_p^\ast)^T$ with $w_j^\ast = w(|\beta_{0j}|)$
	and their partitions $\widehat{\boldsymbol{w}}_0=\widehat{\boldsymbol{w}}_{S_0}$, $\widehat{\boldsymbol{w}}_1=\widehat{\boldsymbol{w}}_{S_0^c}$,
	$\boldsymbol{w}_0^\ast=\boldsymbol{w}_{S_0}^\ast$ and $\boldsymbol{w}_1^\ast=\boldsymbol{w}_{S_0^c}^\ast$.
	Further, let us define ${\delta}_n^\ast = \big[\sqrt{{s(\log n)}/{n}} 
	+ {\lambda_n}\left(||\boldsymbol{w}_0^\ast||_2 + C_2c_5\sqrt{s(\log p)/n}\right)\big]$,
	where $C_2$ and $c_5$ are defined in the following assumptions.
	
	\begin{itemize}
		\item[(A7)]  The initial estimator $\widetilde{\boldsymbol{\beta}}$ satisfies
		$||\widetilde{\boldsymbol{\beta}} - \boldsymbol{\beta}_0||_2 \leq C_2 \sqrt{s(\log p)/n}$ 
		for some constant $C_2>0$, with probability tending to one.
		\item[(A8)] The weight function $w(\cdot)$ is non-increasing over $(0, \infty)$ and is Lipschitz continuous 
		with Lipschitz constant $c_5>0$. Further, 
		$w(C_2\sqrt{s(\log p)/n}) > \frac{1}{2}w(0+)$ for large enough $n$, 
		where $C_2$ is as in Assumption (A7).
		\item[(A9)] With $C_2$ as in Assumption (A7), $\min\limits_{1\leq j \leq s}|\beta_{0j}| > 2C_2 \sqrt{s(\log p)/n}$.
		Further, the derivative of the wight satisfies $w'(|b|) =o(s^{-1}\lambda_n^{-1}(n\log p)^{-1/2})$ for any 
		$|b| > \frac{1}{2} \min\limits_{1\leq j \leq s}|\beta_{0j}|$.
	\end{itemize}

	These assumptions are exactly the same as Conditions 4--6 of  Fan et al. \cite{Fan/etc:2014}.
	The first one put the weaker restriction on the initial estimator to be only consistent 
	in order to achieve the variable selection consistency of the second stage AW-DPD-LASSO estimators
	as shown in the following theorem;
	this minor requirement is satisfied by most simple LASSO estimators, including the DPD-LASSO, 
	so that either of them can be used as the initial estimator here. 
	On the other hand, Assumptions (A8) and (A9) put restrictions on the weight functions used
	where the first one is needed for model selection consistency and 
	both are needed for the asymptotic normality of the resulting estimators.

	\begin{theorem}\label{THM:AdW_Main}
		Suppose that the assumptions of Theorem \ref{THM:FixW_Main} hold with $\boldsymbol{w}=\boldsymbol{w}^\ast$
		and $\delta_n =\delta_n^\ast$. Additionally, if Assumptions (A7)-(A8) hold
		with $\lambda_n s \kappa_n \sqrt{(\log p)/n} \rightarrow 0$
		then, with probability tending to one, 
		there exists a global minimizer $\widehat{\boldsymbol{\beta}} = (\widehat{\boldsymbol{\beta}}_1^T, \widehat{\boldsymbol{\beta}}_2^T)^T$ 
		of the AW-DPD-LASSO objective function $Q_{n,\gamma,\lambda }(\boldsymbol{\beta})$ in (\ref{EQ:ObjFUnc_AdWeight}),
		with adaptive weights, such that
		$$
		\left|\left|\widehat{\boldsymbol{\beta}}_1 - \boldsymbol{\beta}_{10}\right|\right|_2 \leq C_1\delta_n,
		~~~~\mbox{and }~~~~~
		\widehat{\boldsymbol{\beta}}_2 = \boldsymbol{0}_{p-s}.
		$$ 
	\end{theorem}

	\begin{theorem}\label{THM:AdW_AsymNormal}
		Suppose that the assumptions of Theorem \ref{THM:FixW_AsymNormal} hold with $\boldsymbol{w}=\boldsymbol{w}^\ast$
		and $\delta_n =\delta_n^\ast$. Additionally, suppose that Assumptions (A7)-(A9) hold.
		Then, with probability tending to one, 
		there exists a global minimizer $\widehat{\boldsymbol{\beta}} 
		= (\widehat{\boldsymbol{\beta}}_1^T, \widehat{\boldsymbol{\beta}}_2^T)^T$ 
		of the AW-DPD-LASSO objective function $Q_{n,\gamma,\lambda }(\boldsymbol{\beta})$ in (\ref{EQ:ObjFUnc_AdWeight}),
		with adaptive weights, having the same asymptotic properties as those described  in Theorem \ref{THM:FixW_AsymNormal}.
	\end{theorem}

	It is important to note that these results can be further simplified for any give weight function.
	For example, if we use the SCAD weight function as defined in (\ref{EQ:SCAD}), 
	it has been verified in  Fan et al.  \cite{Fan/etc:2014} that Assumption (A8) is satisfied 
	if $\lambda_n > \frac{2}{(a+1)}C_2\sqrt{\frac{s(\log p)}{n}}$, 
	whereas Assumption (A9) is satisfied if $\min\limits_{1\leq j \leq s}|\beta_{0j}| \geq 2a\lambda_n$.
	Therefore, under appropriately simplified conditions as given in Corollary 1 of  Fan et al. \cite{Fan/etc:2014},
	the AW-DPD-LASSO estimator with SCAD weight function can be seen to satisfy 
	exactly the same model section oracle property and asymptotic normality as 
	the DPD-ncv estimator with SCAD penalty. 
	However, the Ad-DPD-LASSO estimators may not enjoy both of these properties
	since the corresponding weight function is not Lipschitz around zero
	but they are so locally on intervals of the form $(c, \infty)$ for any $c>0$. 
	
	\vspace{-.3cm}
	\section{Computational algorithm}
	\label{SEC:computation}
	
	Several computationally efficient algorithms have been developed in the literature 
	to solve the least squares regression problem with different types (e.g., adaptive, grouped, etc.) of LASSO penalties. 
	These techniques often use the local convexity of the objective function. 
	In this section, we develop an appropriate  estimating algorithm for the DPD-LASSO, the Ad-DPD-LASSO and the general AW-DPD-LASSO estimators under the assumption of normal error distribution. 
	To minimize the corresponding objective functions, we propose an iterative optimization algorithm. 
	Firstly the objective function is minimized in $\boldsymbol{\beta}$ with a fixed $\sigma$, 
	and secondly it is optimized on $\sigma$ for fixed  $\boldsymbol{\beta}$ ; these two steps are repeated consecutively until convergence.
	To perform the first minimization with respect to $\bold\beta$, we use the approach of MM-algorithm; here, the observed data are weighted to bound the DPD loss function by a quadratic loss, 	and hence transforming the minimization problem to a least-squares LASSO penalized problem. 
	In the alternative steps, given $\bold\beta$, $\sigma$ is updated using coordinate descent algorithm. A detailed step-by-step derivation of the computing algorithm can be found in Section 2 of the Appendix (available in the Online Supplementary Material).\\
	
	\noindent\textbf{Algorithm 1} (General AW-DPD-LASSO estimator)
	\begin{enumerate}
		\item Set $m=0$. Choose values of the hyper-parameters $\lambda$ and $\gamma$, 
		and set robust initial values $\tilde{\boldsymbol{\beta}}$ and 
		$(\widehat{\boldsymbol{\beta}}^0, \widehat{\sigma}^0)$ using any suitable robust algorithm. 
		\item  For each $m=0, 1, \ldots$, do the following:
		\begin{enumerate}
			\item Compute $\widehat{\boldsymbol{\beta}}^{(m+1)}$ from $(\widehat{\boldsymbol{\beta}}^{(m)}, \widehat{\sigma}^{(m)})$  as 
			\begin{equation} \label{betaupdate}
				\widehat{\boldsymbol{\beta}}^{(m+1)} = \operatorname{arg min} 
				\left[\sum_{i=1}^n \mu_i^{(m)}  \left( \frac{y_i-\boldsymbol{x}_i^T \boldsymbol{\beta}}{\sigma^{(m)}}\right)^2 + \lambda \sum_{j=1}^{p} \omega(| \tilde{\beta}_j|)\cdot|\beta_j| \right]
			\end{equation} 
			with $$\mu_i^{(m)} = \exp\left(-\frac{\gamma}{2} \left(\frac{y_i-\boldsymbol{x}_i^T\widehat{\boldsymbol{\beta}}^{(m)}}{\widehat{\sigma}^{(m)}}\right)^2\right) \left[\sum_{l=1}^n \exp\left(-\frac{\gamma}{2} \left(\frac{y_l-\boldsymbol{x}_l^T\widehat{\boldsymbol{\beta}}^{(m)}}{\widehat{\sigma}^{(m)}}\right)^2\right)\right]^{-1}.$$
			\item Compute $\widehat{\sigma}^{(m+1)}$ from $(\widehat{\boldsymbol{\beta}}^{(m+1)}, \widehat{\sigma}^{(m)})$  as 
			\begin{equation}\label{sigmaupdate}
				(\widehat{\sigma}^{(m+1)})^2 = \left[\frac{1}{n} \sum_{i=1}^n w_i^{(m)} - \frac{\gamma}{(\gamma+1)^{3/2}} \right]^{-1}  \frac{1}{n} \sum_{i=1}^n w_i^{(m)} \left(y_i-\boldsymbol{x}_i^T\widehat{\boldsymbol{\beta}}^{(m+1)}\right)^2, 
			\end{equation}
			where 
			$ w_i^{(m)} 
			=  \exp\left(-\frac{\gamma}{2}\left(\frac{y_i-\boldsymbol{x}_i^T\widehat{\boldsymbol{\beta}}^{(m+1)}}{\widehat{\sigma}^{(m)}}\right)^2 \right).
			$
		\end{enumerate}
		\item If 
		$|Q_{n,\gamma,\lambda}(\widehat{\boldsymbol{\beta}}^{(m+1)}, \widehat{\sigma}^{(m+1)}) 
		- Q_{n,\gamma,\lambda}(\widehat{\boldsymbol{\beta}}^{(m)}, \widehat{\sigma}^{(m)}) | \leq \varepsilon$ (pre-specified): 
		Go to Step 4. \\
		Else: Set $\tilde{\boldsymbol{\beta}} = \widehat{\boldsymbol{\beta}}^{(m)}$, $m=m+1$  and return to Step 2.
		
		\item Output: $\widehat{\boldsymbol{\beta}} = \widehat{\boldsymbol{\beta}}^{(m+1)}$ and $\widehat{\sigma}=\widehat{\sigma}^{(m+1)}$
		as the  AW-DPD-LASSO estimator. 
	\end{enumerate}

	The performance of Algorithm 1 depends on the choice of initial values  $\tilde{\boldsymbol{\beta}}$ and 
	$(\widehat{\boldsymbol{\beta}}^0, \widehat{\sigma}^0)$, as well as on the regularization parameter $\lambda$. 
	We can apply any robust standard regression method, such as  RLARS, or DPD-LASSO to get the required initial values of the parameters. 
	Note that $\tilde{\boldsymbol{\beta}}$ and  $\widehat{\boldsymbol{\beta}}^0$ can be chosen differently or the same  but they should both be conservative as not to loose any important covariate.
	In our analysis, we employed the DPD-LASSO estimator for both $\tilde{\boldsymbol{\beta}}$ and  $(\widehat{\boldsymbol{\beta}}^0, \widehat{\sigma}^0),$ which is robust and tends to be quite conservative in variable selection. Just like the usual LASSO estimator serves as a good initial estimator for the computation of the classical adaptive LASSO estimator, we have seen that the  our choice of DPD-LASSO as the initial estimator in computation of AW-DPD-LASSO estimator gives desirable and sufficiently good results; so we stick with this recommendation for the choice of initial estimators. However, other initial estimators may be explored in future although the proposed algorithm appears not to be very sensitive to these choices as long as the initial estimators of the regression coefficients are robust enough and do not miss any important variables. 
	
	For the selection of the (best) $\lambda$, we propose to use the high-dimensional Bayesian Information Criterion (HBIC), 
	which has demonstrably better performance compared to the standard BIC under non-polynomial dimensionality 
	(Kim et al. \cite{kim}; Wang et al. \cite{wang}; Fan and Tang \cite{FanTang}). The HBIC is defined as
	\begin{equation} \label{HBIC}
		\text{HBIC}(\lambda) = \log(\widehat{\sigma}_{\lambda}^2) + \frac{\log \log(n)\log p}{n}\|\widehat{\boldsymbol{\beta}}_{\lambda} \|_0.
	\end{equation}
	and we select the optimal $\lambda$ minimizing the HBIC values over a pre-determined set (via grid search).

	\section{Simulation study \label{sec:simulation}}
	
	\subsection{Simulation Settings \label{sec:simulationsettings}}
	
	The data are generated from the LRM in (\ref{1.1}) following similar set-up as Zou and Li \cite{Zou3}. 
	We set the sample size $n=100$, the true error standard deviation $\sigma_0 = 0.5$, and $p=500, 1000.$ Errors are generated independently from $\mathcal{N}(0, \sigma_0^2$) distributions.
	 Two different scenarios are considered for the true sparse regression coefficients $\boldsymbol{\beta}_0$ as follows:
	\begin{itemize}
		\item Setting A: Only three true coefficients are not null. $\boldsymbol{\beta}_0 = (3, 1.5, 0, 0, 2, \boldsymbol{0}_{p-5}).$
		\item Setting B : A more challenging scenario obtained by modifying the precedent Setting A, where 
		we divide the first 60 components into continuous blocks of size 20, and assign the coefficient values $(3, 1.5, 0, 0, 2, \boldsymbol{0}_{15})$ to each block. Thus, the true model here has nine important covariates.
	\end{itemize}
	Rows of the design matrix $\boldsymbol{X}$ are drawn from the normal distribution $\mathcal{N}\left(\boldsymbol{0, \boldsymbol{\Sigma}}\right)$, 
	where $\boldsymbol{\Sigma}$ is a positive definite matrix of Toeplitz structure, with the $(i,j)$-th element being $0.5^{|i-j|}$. 
	To study the robustness of the method, we additionally modify these (pure) data by four type of contamination as follows: 
	\begin{itemize}
		\item $Y$-outliers : We add random observations, drawn independently from a normal distribution $\mathcal{N}\left(20,1\right)$, 
		to the response variable for a random $10\%$ of each sample.
		\item $\boldsymbol{X}$-outliers : We add random observations, drawn independently from a normal distribution, $\mathcal{N}\left(20,1\right),$ 
		to the covariate values in $10$ columns of $\boldsymbol{X}$ for a random $10\%$ of each sample. 
		\item Heavy-tailed $Y$-outliers : We add random observations centred at $20$, drawn independently from a t-student distributions with 3 degrees of freedom,   
		to the response variable for a random $10\%$ of each sample.
		\item Heavy-tailed $\boldsymbol{X}$-outliers : We add random observations centred at $20$, drawn independently from a t-student distributions with 3 degrees of freedom,
		to the covariate values in $10$ columns of $\boldsymbol{X}$ for a random $10\%$ of each sample. 
	\end{itemize}
	We repeat the simulations over $R=100$ replications to compute different performance measures. 
	We evaluate the variable selection performances through Model Size (MS), True Positive proportion (TP) and True Negative proportion, 
	defined as 
	\begin{align*}
		\operatorname{MS}(\widehat{\boldsymbol{\beta}}) = |\text{supp}(\widehat{\boldsymbol{\beta}})|,
		~	\operatorname{TP}(\widehat{\boldsymbol{\beta}}) = \frac{|\text{supp}(\widehat{\boldsymbol{\beta}})\cap\text{supp}(\boldsymbol{\beta}_0)|}{|\text{supp}(\boldsymbol{\beta}_0)|},
		~	\operatorname{TN}(\widehat{\boldsymbol{\beta}}) = \frac{|\text{supp}^c(\widehat{\boldsymbol{\beta}})\cap\text{supp}^c(\boldsymbol{\beta}_0)|}{|\text{supp}^c(\boldsymbol{\beta}_0)|}.
	\end{align*}
	Additionally, in order to asses the estimation accuracy, we compute the mean square error for the true non-zero coefficients (MSES) 
	and zero coefficients (MSEN) of $\widehat{\boldsymbol{\beta}}$ separately, as well as the absolute Estimation Error (EE) of $\widehat{\sigma}$:  
	\begin{align*}
		\operatorname{MSES}(\widehat{\boldsymbol{\beta}}) = \frac{1}{s} \parallel \widehat{\boldsymbol{\beta}}_{\mathcal{S}} - \boldsymbol{\beta}_{0\mathcal{S}} \parallel^2, 
		~~~~	\operatorname{MSEN}(\widehat{\boldsymbol{\beta}}) = \frac{1}{p-s} \parallel \widehat{\boldsymbol{\beta}}_{\mathcal{N}} ||^2,
		~~~~	\operatorname{EE}(\widehat{\sigma}) = |\widehat{\sigma}-\sigma_0|.
	\end{align*}
	Finally, we also examine the prediction accuracy on an unused test sample of size $n=100$, 
	generated from the same model distributions as that of the trainning sample.
	For this purpose, we use the Absolute Prediction Bias (APrB) defined as
	\begin{align*}
		\operatorname{APrB}(\widehat{\boldsymbol{\beta}}) &= \parallel \boldsymbol{y}_{\text{test}}-\mathbb{X}_{\text{test}}\widehat{\boldsymbol{\beta}} \parallel_1.
	\end{align*}

	\subsection{Competing methods}
	
	We compare our adaptive robust methods with the robust least angle regression (RLARS; Khan et al. \cite{khan}),  sparse least trimmed squares (sLTS; Alfons  et al. \cite{alfons}), random sample consensus (RANSAC),  and the LAD-LASSO (Wang et al. \cite{wang}) estimators.  
	Additionally, our Ad-DPD-LASSO and AW-DPD-LASSO methods (with SCAD penalty) are also compared with the related DPD based methods, namely the DPD-LASSO and the DPD-ncv with SCAD penalty, for the same values of tuning parameters ($\gamma= 0.1, 0.3, 0.5, 0.7, 1$). 
	Moreover, for comparing our methods in terms of efficiency loss, we use three standard non-robust methods, 
	namely the ones based on the least-squares loss along with the LASSO, SCAD and MCP penalties, 
	which we will refer to as the LS-LASSO, LS-SCAD and LS-MCP, respectively. 
	The standard adaptive LASSO (Ad-LS-LASSO) is also considered, with the initial parameters being obtained via LS-LASSO.

	For the DPD-LASSO and DPD-ncv, the starting points are taken as the RLARS estimates because of computational efficiency.
	We use 5-fold cross-validation for the selection of the regularized parameter $\lambda$ in all the above competing methods 
	except LAD-LASSO, DPD-LASSO and DPD-ncv.  We use BIC criterion to choose the optimum $\lambda$ in LAD-LASSO, 
	whereas the HBIC is used for DPD-LASSO and DPD-ncv.
	
	\vspace{-0.2cm}
	\subsection{Simulation Results}

	For brevity, we present the simulation results (performance measures) for the most challenging case of Setting B with $p=1000$
	in Tables \ref{p1000e0signal3}-\ref{p1000hex1signal3}, for pure data and four types of contaminated data. 
	The results for the remaining cases are provided in Section C of the Appendix available as Supplementary Material.
	
	All these results evidence the significant  advantage entailed by DPD-based methods in terms of robustness, accentuated in the high-dimensional context where classical inferential methods are specially sensitive to outliers.	Moreover, all other robust methods considered, LAD-LASSO, RLARS, sLTS  and RANSAC, perform worse than the proposed DPD-based methods in variable selection and parameter estimation, under all scenarios of contamination as well as under pure data. Then, DPD-based estimation methods represent an appealing robust alternative to classical likelihood-based methods, with the lowest efficiency loss and greatest robustness gain.	It is interesting to note that all methods underperform on the second scenario (Setting B), where the true model size is greater. 
	Additionally, as expected, adaptive methods including DPD-based estimators improve the estimation accuracy with respect to their corresponding standard LASSO. Indeed, adaptive methods select more parsimonious models containing all true significant variables. 
	
	\begin{table}[!h]
		\centering
		\caption{Performance measures obtained by different methods for $p=1000$, Setting B and no outliers}
		\resizebox{\textwidth}{!}{\begin{tabular}{l|rrr|rrr|r}
			\hline
			Method	& MS($\widehat{\boldsymbol{\beta}}$) & TP($\widehat{\boldsymbol{\beta}}$) & TN($\widehat{\boldsymbol{\beta}}$) & MSES($\widehat{\boldsymbol{\beta}}$) &  MSEN($\widehat{\boldsymbol{\beta}}$)  &  EE($\widehat{\sigma}$) & APrB($\widehat{\boldsymbol{\beta}}$)  \\ 
			& & & &  $(10^{-2})$ & $(10^{-5})$ &  $(10^{-2})$ &  $(10^{-2})$ \\
			\hline
			LS-LASSO & 21.27 & 1.00 & 0.99 & 2.52 & 1.52 & 34.91 & 6.31 \\ 
			Ad-LS-LASSO & 9.00 & 1.00 & 1.00 & 1.49 & 0.00 & 31.30 & 5.16 \\ 
			LS-SCAD & 9.03 & 1.00 & 1.00 & 0.36 & 0.00 & 19.62 & 4.23 \\ 
			LS-MCP & 9.04 & 1.00 & 1.00 & 0.36 & 0.00 & 19.60 & 4.22 \\
			\hline 
			LAD-LASSO & 16.70 & 1.00 & 0.99 & 7.55 & 2.63 & 53.97 & 9.17 \\ 
			RLARS & 12.70 & 1.00 & 1.00 & 0.48 & 2.82 & 7.65 & 4.49 \\ 
			sLTS & 61.80 & 0.79 & 0.94 & 219.62 & 300.30 & 6.79 & 46.08 \\ 
			RANSAC & 15.10 & 0.52 & 0.99 & 238.20 & 394.16 & 143.68 & 49.35 \\ 
			\hline
			DPD-LASSO $\gamma = $ 0.1 & 10.37 & 1.00 & 1.00 & 4.59 & 114.54 & 58.75 & 11.10 \\ 
			DPD-LASSO $\gamma = $ 0.3 & 8.73 & 0.90 & 1.00 & 27.95 & 701.22 & 211.92 & 25.86 \\ 
			DPD-LASSO $\gamma = $ 0.5 & 13.13 & 0.99 & 1.00 & 3.13 & 95.99 & 35.29 & 10.61 \\ 
			DPD-LASSO $\gamma = $ 0.7 & 22.12 & 0.73 & 0.98 & 70.05 & 1201.08 & 25.81 & 29.53 \\ 
			DPD-LASSO $\gamma = $ 1 & 14.40 & 0.69 & 0.99 & 73.96 & 1500.41 & 163.41 & 37.38 \\ 
			\hline
			DPD-ncv $\gamma = $ 0.1 & 9.00 & 1.00 & 1.00 & 0.14 & 2.64 & 5.58 & 4.20  \\ 
			DPD-ncv  $\gamma = $  0.3 & 9.00 & 1.00 & 1.00 & 0.17 & 3.08 & 10.28 & 4.21 \\ 
			DPD-ncv  $\gamma = $  0.5 & 9.00 & 1.00 & 1.00 & 0.19 & 3.36 & 13.86 & 4.23 \\ 
			DPD-ncv  $\gamma = $ 0.7 & 9.00 & 1.00 & 1.00 & 0.23 & 3.82 & 16.63 & 4.25  \\ 
			DPD-ncv $\gamma = $ 1 & 9.01 & 1.00 & 1.00 & 0.26 & 4.12 & 19.82 & 4.31  \\ 
			\hline
			Ad-DPD-LASSO  $\gamma = $ 0.1 & 11.34 & 1.00 & 1.00 & 0.49 & 2.79 & 6.18 & 4.24  \\ 
			Ad-DPD-LASSO  $\gamma = $ 0.3 & 9.10 & 1.00 & 1.00 & 0.64 & 0.04 & 3.04 & 4.33 \\ 
			Ad-DPD-LASSO  $\gamma = $ 0.5 & 9.58 & 1.00 & 1.00 & 0.66 & 0.45 & 4.65 & 4.39  \\ 
			Ad-DPD-LASSO  $\gamma = $ 0.7 & 9.16 & 1.00 & 1.00 & 0.79 & 0.11 & 5.29 & 4.31  \\ 
			Ad-DPD-LASSO  $\gamma = $ 1 & 9.18 & 1.00 & 1.00 & 1.06 & 0.38 & 8.68 & 4.36  \\ 
			\hline
			AW-DPD-LASSO  $\gamma = $ 0.1 & 10.68 & 1.00 & 1.00 & 0.34 & 0.28 & 3.92 & 3.78  \\ 
			AW-DPD-LASSO  $\gamma = $ 0.3 & 9.62 & 1.00 & 1.00 & 0.36 & 0.00 & 3.93 & 3.92  \\ 
			AW-DPD-LASSO  $\gamma = $ 0.5 & 10.70 & 1.00 & 1.00 & 0.43 & 0.30 & 5.46 & 4.00  \\ 
			AW-DPD-LASSO  $\gamma = $ 0.7 & 9.82 & 1.00 & 1.00 & 0.50 & 0.20 & 6.68 & 4.04  \\ 
			AW-DPD-LASSO  $\gamma = $ 1 & 9.60 & 1.00 & 1.00 & 0.66 & 0.09 & 9.56 & 3.99  \\
			\hline
		\end{tabular}}
		\label{p1000e0signal3}
	\end{table}

	\begin{table}[h]
		\centering
		\caption{Performance measures obtained by different methods for $p=1000$, Setting B and $Y$-outliers}
		\resizebox{\textwidth}{!}{\begin{tabular}{l|rrr|rrr|r}
			\hline
			Method	& MS($\widehat{\boldsymbol{\beta}}$) & TP($\widehat{\boldsymbol{\beta}}$) & TN($\widehat{\boldsymbol{\beta}}$) & MSES($\widehat{\boldsymbol{\beta}}$) &  MSEN($\widehat{\boldsymbol{\beta}}$)  &  EE($\widehat{\sigma}$) & APrB($\widehat{\boldsymbol{\beta}}$)  \\ 
			& & & &  $(10^{-2})$ & $(10^{-5})$ &  $(10^{-2})$ &  $(10^{-2})$ \\
			\hline
			LS-LASSO & 7.85 & 0.46 & 1.00 & 332.57 & 56.63 & 923.77 & 53.67 \\
			Ad-LS-LASSO & 5.24 & 0.42 & 1.00 & 282.04 & 196.13 & 728.11 & 46.04 \\  
			LS-SCAD & 26.32 & 0.63 & 0.98 & 253.30 & 396.77 & 526.46 & 44.85 \\ 
			LS-MCP & 11.40 & 0.52 & 0.99 & 256.43 & 316.72 & 584.97 & 44.96 \\ 
			\hline
			LAD-LASSO & 25.93 & 0.85 & 0.98 & 144.66 & 226.31 & 315.34 & 36.80 \\ 
			RLARS & 20.25 & 0.79 & 0.99 & 90.67 & 320.57 & 144.45 & 26.38 \\ 
			sLTS & 51.28 & 0.93 & 0.96 & 83.66 & 131.04 & 7.52 & 22.84 \\ 
			RANSAC & 13.00 & 0.37 & 0.99 & 319.90 & 646.28 & 296.99 & 56.92 \\ 
			\hline
			DPD-LASSO $\gamma = $ 0.1 & 8.50 & 0.74 & 1.00 & 65.91 & 1475.56 & 348.96 & 38.15 \\ 
			DPD-LASSO $\gamma = $ 0.3 & 6.94 & 0.62 & 1.00 & 80.41 & 1681.17 & 381.90 & 39.62 \\ 
			DPD-LASSO $\gamma = $ 0.5 & 12.22 & 0.88 & 1.00 & 27.28 & 618.61 & 130.50 & 20.97 \\ 
			DPD-LASSO $\gamma = $ 0.7 & 23.58 & 0.61 & 0.98 & 92.40 & 1643.77 & 30.59 & 36.54 \\ 
			DPD-LASSO $\gamma = $ 1 & 14.03 & 0.53 & 0.99 & 105.81 & 2131.45 & 215.98 & 47.96 \\ 
			\hline
			DPD-ncv $\gamma = $ 0.1 & 9.10 & 0.91 & 1.00 & 23.00 & 425.28 & 127.54 & 16.34 \\ 
			DPD-ncv  $\gamma = $ 0.3 & 9.34 & 0.98 & 1.00 & 4.78 & 118.05 & 19.60 & 6.88 \\ 
			DPD-ncv $\gamma = $ 0.5 & 9.12 & 0.98 & 1.00 & 4.65 & 110.96 & 13.06 & 6.67 \\ 
			DPD-ncv $\gamma = $ 0.7 & 9.21 & 0.98 & 1.00 & 4.47 & 101.86 & 10.88 & 6.64 \\ 
			DPD-ncv $\gamma = $ 1 & 9.22 & 0.99 & 1.00 & 4.49 & 100.00 & 10.20 & 6.62 \\ 
			\hline
			Ad-DPD-LASSO  $\gamma = $ 0.1 & 18.71 & 0.81 & 0.99 & 81.49 & 835.39 & 130.08 & 22.43 \\ 
			Ad-DPD-LASSO $\gamma = $ 0.3 & 10.92 & 0.99 & 1.00 & 6.21 & 27.67 & 10.02 & 6.14 \\ 
			Ad-DPD-LASSO  $\gamma = $ 0.5 & 9.98 & 0.99 & 1.00 & 5.81 & 24.30 & 5.94 & 5.50 \\ 
			Ad-DPD-LASSO  $\gamma = $ 0.7 & 9.48 & 0.99 & 1.00 & 3.87 & 11.31 & 7.30 & 5.11 \\ 
			Ad-DPD-LASSO  $\gamma = $ 1 & 9.96 & 0.97 & 1.00 & 15.38 & 48.27 & 8.59 & 6.66 \\ 
			\hline
			AW-DPD-LASSO  $\gamma = $ 0.1 & 12.98 & 0.80 & 0.99 & 92.47 & 567.77 & 169.98 & 23.06 \\ 
			AW-DPD-LASSO  $\gamma = $ 0.3 & 10.84 & 0.99 & 1.00 & 6.36 & 31.77 & 9.71 & 5.91 \\ 
			AW-DPD-LASSO  $\gamma = $ 0.5 & 12.72 & 0.99 & 1.00 & 3.84 & 15.54 & 7.87 & 4.91 \\ 
			AW-DPD-LASSO  $\gamma = $ 0.7 & 10.80 & 0.99 & 1.00 & 3.69 & 13.75 & 8.37 & 5.14 \\ 
			AW-DPD-LASSO  $\gamma = $ 1 & 11.51 & 0.96 & 1.00 & 17.81 & 57.74 & 8.38 & 8.38 \\ 
			\hline
		\end{tabular}}
		\label{p1000e1signal3}
	\end{table}

	\begin{table}[h]
		\centering
		\caption{Performance measures obtained by different methods for $p=1000$, Setting B and $\boldsymbol{X}$-outliers}
		\resizebox{\textwidth}{!}{\begin{tabular}{l|rrr|rrr|r}
			\hline
			Method	& MS($\widehat{\boldsymbol{\beta}}$) & TP($\widehat{\boldsymbol{\beta}}$) & TN($\widehat{\boldsymbol{\beta}}$) & MSES($\widehat{\boldsymbol{\beta}}$) &  MSEN($\widehat{\boldsymbol{\beta}}$)  &  EE($\widehat{\sigma}$) & APrB($\widehat{\boldsymbol{\beta}}$)  \\ 
			& & & &  $(10^{-2})$ & $(10^{-5})$ &  $(10^{-2})$ &  $(10^{-2})$ \\
			\hline
			LS-LASSO & 7.85 & 0.46 & 1.00 & 332.57 & 56.63 & 923.77 & 53.67 \\ 
			Ad-LS-LASSO & 9.00 & 1.00 & 1.00 & 1.52 & 0.00 & 31.56 & 4.85 \\ 
			LS-SCAD & 26.32 & 0.63 & 0.98 & 253.30 & 396.77 & 526.46 & 44.85 \\ 
			LS-MCP & 11.40 & 0.52 & 0.99 & 256.43 & 316.72 & 584.97 & 44.96 \\ 
			\hline
			LAD-LASSO & 25.93 & 0.85 & 0.98 & 144.66 & 226.31 & 315.34 & 36.80 \\ 
			RLARS & 20.25 & 0.79 & 0.99 & 90.67 & 320.57 & 144.45 & 26.38 \\ 
			sLTS & 51.28 & 0.93 & 0.96 & 83.66 & 131.04 & 7.52 & 22.84 \\ 
			RANSAC & 12.23 & 0.36 & 0.99 & 335.92 & 696.48 & 322.69 & 51.78 \\ 
			\hline
			DPD-LASSO $\gamma = $ 0.1 & 10.33 & 1.00 & 1.00 & 4.58 & 114.30 & 58.74 & 10.80 \\ 
			DPD-LASSO $\gamma = $ 0.3 & 8.71 & 0.90 & 1.00 & 27.94 & 700.75 & 211.92 & 26.32 \\ 
			DPD-LASSO $\gamma = $ 0.5 & 13.08 & 0.99 & 1.00 & 3.14 & 96.10 & 35.47 & 9.83 \\ 
			DPD-LASSO $\gamma = $ 0.7 & 21.04 & 0.79 & 0.99 & 56.38 & 945.35 & 31.38 & 24.17 \\ 
			DPD-LASSO $\gamma = $ 1 & 14.30 & 0.69 & 0.99 & 73.45 & 1515.09 & 164.59 & 37.40 \\ 
			\hline
			DPD-ncv $\gamma = $ 0.1 & 9.00 & 1.00 & 1.00 & 0.14 & 2.64 & 5.58 & 3.88 \\ 
			DPD-ncv $\gamma = $ 0.3 & 9.00 & 1.00 & 1.00 & 0.17 & 3.09 & 10.29 & 3.89 \\ 
			DPD-ncv $\gamma = $ 0.5 & 9.00 & 1.00 & 1.00 & 0.19 & 3.38 & 13.85 & 3.94 \\ 
			DPD-ncv $\gamma = $ 0.7 & 9.00 & 1.00 & 1.00 & 0.23 & 3.81 & 16.61 & 3.99 \\ 
			DPD-ncv $\gamma = $ 1 & 9.01 & 1.00 & 1.00 & 0.28 & 4.14 & 19.80 & 4.10 \\ 
			\hline
			Ad-DPD-LASSO  $\gamma = $ 0.1 & 11.20 & 1.00 & 1.00 & 0.50 & 2.62 & 5.91 & 4.20 \\ 
			Ad-DPD-LASSO  $\gamma = $ 0.3 & 9.10 & 1.00 & 1.00 & 0.64 & 0.04 & 3.04 & 4.15 \\ 
			Ad-DPD-LASSO  $\gamma = $ 0.5 & 9.58 & 1.00 & 1.00 & 0.66 & 0.45 & 4.65 & 4.13 \\ 
			Ad-DPD-LASSO  $\gamma = $ 0.7 & 9.16 & 1.00 & 1.00 & 0.79 & 0.11 & 5.30 & 4.17 \\ 
			Ad-DPD-LASSO  $\gamma = $ 1 & 9.18 & 1.00 & 1.00 & 1.06 & 0.39 & 8.68 & 4.02 \\ 
			\hline
			AW-DPD-LASSO  $\gamma = $ 0.1 & 10.68 & 1.00 & 1.00 & 0.34 & 0.28 & 3.92 & 3.58 \\ 
			AW-DPD-LASSO  $\gamma = $ 0.3 & 9.42 & 1.00 & 1.00 & 0.36 & 0.00 & 3.88 & 3.76 \\ 
			AW-DPD-LASSO  $\gamma = $ 0.5 & 10.70 & 1.00 & 1.00 & 0.43 & 0.30 & 5.46 & 3.80 \\ 
			AW-DPD-LASSO  $\gamma = $ 0.7 & 9.80 & 1.00 & 1.00 & 0.50 & 0.20 & 6.67 & 3.94 \\ 
			AW-DPD-LASSO  $\gamma = $ 1 & 9.58 & 1.00 & 1.00 & 0.67 & 0.09 & 9.56 & 4.02 \\ 
			\hline
		\end{tabular}}
		\label{p1000ex1signal3}
	\end{table}

	Furthermore, there exist significant differences between the four DPD-based methods considered. 
	The DPD-ncv (with SCAD penalty) estimator performs generally better
	than Ad-DPD-LASSO and AW-DPD-LASSO estimators in terms of variable selection 
	and it registers lower mean square error for the true non-zero coefficient (MSES).
	However, DPD-ncv estimator presents  higher mean squared error for the zero coefficients (MSEN) as well as much greater errors for the error variance. 
	On the other hand, DPD-LASSO performs worse in both MSES and MSEN measures compared to any other DPD-based estimator, as it selects a larger number of non-significant variables.
	Note that DPD loss with the parameter value $\gamma = 0.1$ is the comparatively less robust, 
	and hence produces greater MSES and MSEN in the presence of data contamination (close to the LS based results). 
	For larger values of $\gamma\geq 0.3$,  DPD-based estimators provide extremely robust inference 	against contamination in both the response  variable and the covariates, often achieving the true model (TP of 100\%),
	 contrarily to non-robust methods. 
	From our empirical results, optimal $\gamma$  value hover around $\gamma \in[0.3, 0.5].$
	Additionally, it is also evident from simulations that the AW-DPD-LASSO estimator (with SCAD based weight) performs very competitively with the DPD-ncv method (with SCAD penalty) in terms of variables selection
	and even better than the DPD-ncv in terms of estimating $\sigma$ and prediction performance.  
	Therefore, it can serve as a fast yet excellent alternative  to  DPD-ncv  in ultra-high dimensional set-ups. 
	
	Further comparisons of the proposed Ad-DPD-LASSO and AW-DPD-LASSO (with SCAD penalty) estimators are provided in Section C of the Appendix available as Supplementary Material.
	For both of them, the  prediction errors are seen to decrease as the number of covariates or the values of tuning parameter $\gamma$ increase. The AW-DPD-LASSO performs slightly better than the Ad-DPD-LASSO
	at larger values of $\gamma$, whereas the opposite is observed at smaller values of $\gamma$.
	
	\begin{table}[ht]
		\centering
		\caption{Performance measures obtained by different methods for $p=1000$, Setting B and heavy-tailed $Y$-outliers}
		\label{p1000he1signal3}
		\resizebox{\textwidth}{!}{	{\color{black}
		\begin{tabular}{l|rrr|rrr|r}
			\hline
			Method	& MS($\widehat{\boldsymbol{\beta}}$) & TP($\widehat{\boldsymbol{\beta}}$) & TN($\widehat{\boldsymbol{\beta}}$) & MSES($\widehat{\boldsymbol{\beta}}$) &  MSEN($\widehat{\boldsymbol{\beta}}$)  &  EE($\widehat{\sigma}$) & APrB($\widehat{\boldsymbol{\beta}}$)  \\ 
			& & & &  $(10^{-2})$ & $(10^{-5})$ &  $(10^{-2})$ &  $(10^{-2})$ \\
			\hline
			LS-LASSO & 6.71 & 0.48 & 1.00 & 331.65 & 39.19 & 933.86 & 49.59 \\ 
			LS-SCAD & 26.50 & 0.62 & 0.98 & 259.29 & 423.82 & 531.59 & 36.96 \\ 
			LS-MCP & 11.23 & 0.51 & 0.99 & 265.74 & 339.76 & 595.32 & 36.20 \\ 
			\hline
			LAD-LASSO & 24.15 & 0.84 & 0.98 & 145.58 & 174.40 & 321.14 & 33.35 \\ 
			RLARS & 18.49 & 0.77 & 0.99 & 93.50 & 311.20 & 152.40 & 23.97 \\ 
			sLTS & 50.92 & 0.93 & 0.96 & 85.14 & 137.32 & 8.24 & 22.48 \\ 
			RANSAC & 9.10 & 0.37 & 0.99 & 321.35 & 371.72 & 521.68 & 47.27 \\ 
			\hline
			DPD-LASSO $\gamma = $ 0.1 & 11.44 & 0.81 & 1.00 & 103.20 & 5.27 & 213.00 & 17.11 \\ 
			DPD-LASSO $\gamma = $ 0.3 & 13.03 & 0.96 & 1.00 & 27.20 & 1.65 & 60.73 & 9.03 \\ 
			DPD-LASSO $\gamma = $ 0.5 & 10.79 & 0.90 & 1.00 & 69.31 & 4.74 & 135.92 & 17.10 \\ 
			DPD-LASSO $\gamma = $ 0.7 & 16.09 & 0.91 & 0.99 & 73.84 & 38.48 & 80.23 & 16.43 \\ 
			DPD-LASSO $\gamma = $ 1 & 13.21 & 0.68 & 0.99 & 207.95 & 95.24 & 219.09 & 35.05 \\
			\hline
			DPD-ncv $\gamma = $ 0.1 & 8.39 & 0.90 & 1.00 & 23.67 & 405.46 & 110.48 & 11.75 \\ 
			DPD-ncv $\gamma = $ 0.3 & 9.04 & 0.99 & 1.00 & 3.41 & 76.12 & 21.48 & 6.66 \\ 
			DPD-ncv $\gamma = $ 0.5 & 9.19 & 0.99 & 1.00 & 2.86 & 87.24 & 17.18 & 5.64 \\ 
			DPD-ncv $\gamma = $ 0.7 & 8.79 & 0.92 & 1.00 & 12.38 & 251.15 & 13.92 & 7.95 \\ 
			DPD-ncv $\gamma = $ 1 & 6.77 & 0.64 & 1.00 & 53.31 & 1131.87 & 14.37 & 21.74 \\ 
			\hline
			Ad-DPD-LASSO $\gamma = $ 0.1 & 7.86 & 0.85 & 1.00 & 58.77 & 61.17 & 129.14 & 11.96 \\ 
			Ad-DPD-LASSO $\gamma = $ 0.3 & 8.64 & 0.96 & 1.00 & 18.49 & 15.55 & 30.00 & 6.00 \\ 
			Ad-DPD-LASSO $\gamma = $ 0.5 & 8.64 & 0.96 & 1.00 & 15.42 & 2.39 & 26.15 & 6.18 \\ 
			Ad-DPD-LASSO $\gamma = $ 0.7 & 8.72 & 0.97 & 1.00 & 13.27 & 3.34 & 20.85 & 6.12 \\ 
			Ad-DPD-LASSO $\gamma = $ 1 & 8.64 & 0.96 & 1.00 & 18.30 & 3.34 & 24.72 & 6.63 \\ 
			\hline
			AW-DPD-LASSO $\gamma = $ 0.1 & 12.18 & 0.82 & 1.00 & 83.09 & 491.03 & 144.38 & 21.07 \\ 
			AW-DPD-LASSO $\gamma = $ 0.3 & 10.80 & 0.98 & 1.00 & 12.56 & 90.51 & 8.87 & 6.14 \\ 
			AW-DPD-LASSO $\gamma = $ 0.5 & 9.84 & 0.97 & 1.00 & 13.18 & 78.97 & 12.90 & 6.47 \\ 
			AW-DPD-LASSO $\gamma = $ 0.7 & 10.23 & 0.97 & 1.00 & 17.79 & 73.99 & 16.81 & 7.30 \\ 
			AW-DPD-LASSO $\gamma = $ 1 & 8.84 & 0.84 & 1.00 & 82.33 & 114.36 & 85.87 & 12.59 \\ 
			\hline
		\end{tabular}
	}	}
	\end{table}

	\begin{table}[ht]
		\centering
		\caption{Performance measures obtained by different methods for $p=1000$, Setting B and heavy-tailed $\boldsymbol{X}$-outliers}
		\label{p1000hex1signal3}
		\resizebox{\textwidth}{!}{ {\color{black}
		\begin{tabular}{l|rrr|rrr|r}
			\hline
			Method	& MS($\widehat{\boldsymbol{\beta}}$) & TP($\widehat{\boldsymbol{\beta}}$) & TN($\widehat{\boldsymbol{\beta}}$) & MSES($\widehat{\boldsymbol{\beta}}$) &  MSEN($\widehat{\boldsymbol{\beta}}$)  &  EE($\widehat{\sigma}$) & APrB($\widehat{\boldsymbol{\beta}}$)  \\ 
			& & & &  $(10^{-2})$ & $(10^{-5})$ &  $(10^{-2})$ &  $(10^{-2})$ \\
			\hline
			LS-LASSO & 21.23 & 1.00 & 0.99 & 2.58 & 1.47 & 34.90 & 5.24 \\ 
			LS-SCAD & 9.02 & 1.00 & 1.00 & 0.35 & 0.00 & 19.97 & 3.55 \\ 
			LS-MCP & 9.01 & 1.00 & 1.00 & 0.35 & 0.00 & 19.99 & 3.54 \\ 
			\hline
			LAD-LASSO & 16.59 & 1.00 & 0.99 & 7.61 & 3.03 & 55.06 & 7.78 \\ 
			RLARS & 13.79 & 1.00 & 1.00 & 0.76 & 3.92 & 8.60 & 4.01 \\ 
			sLTS & 61.45 & 0.79 & 0.95 & 217.87 & 304.40 & 6.04 & 42.31 \\ 
			RANSAC & 10.44 & 0.54 & 0.99 & 235.33 & 182.41 & 359.73 & 39.38 \\ 
			\hline
			DPD-LASSO $\gamma = $ 0.1 & 12.46 & 1.00 & 1.00 & 4.09 & 0.71 & 21.24 & 6.03 \\ 
			DPD-LASSO $\gamma = $ 0.3 & 12.46 & 1.00 & 1.00 & 5.54 & 1.01 & 25.40 & 6.91 \\ 
			DPD-LASSO $\gamma = $ 0.5 & 11.07 & 0.97 & 1.00 & 30.79 & 2.30 & 75.72 & 12.66 \\ 
			DPD-LASSO $\gamma = $ 0.7 & 14.11 & 0.96 & 0.99 & 44.89 & 12.47 & 75.02 & 14.84 \\ 
			DPD-LASSO $\gamma = $ 1 & 9.56 & 0.75 & 1.00 & 178.82 & 30.81 & 258.91 & 33.98 \\ 
			\hline
			DPD-ncv $\gamma = $ 0.1 & 9.00 & 1.00 & 1.00 & 0.13 & 2.60 & 5.17 & 3.56 \\ 
			DPD-ncv $\gamma = $ 0.3 & 9.00 & 1.00 & 1.00 & 0.16 & 3.03 & 10.65 & 3.61 \\ 
			DPD-ncv $\gamma = $ 0.5 & 9.00 & 1.00 & 1.00 & 0.19 & 3.71 & 15.54 & 3.60 \\ 
			DPD-ncv $\gamma = $ 0.7 & 9.01 & 1.00 & 1.00 & 0.24 & 4.46 & 19.84 & 3.72 \\ 
			DPD-ncv $\gamma = $ 1 & 8.98 & 0.99 & 1.00 & 0.94 & 25.21 & 24.32 & 4.19 \\ 
			\hline 
			Ad-DPD-LASSO $\gamma = $ 0.1 & 9.00 & 1.00 & 1.00 & 0.64 & 0.00 & 2.84 & 3.66 \\ 
			Ad-DPD-LASSO $\gamma = $ 0.3 & 9.00 & 1.00 & 1.00 & 0.82 & 0.00 & 3.05 & 3.72 \\ 
			Ad-DPD-LASSO $\gamma = $ 0.5 & 9.00 & 1.00 & 1.00 & 1.10 & 0.00 & 3.42 & 3.87 \\ 
			Ad-DPD-LASSO $\gamma = $ 0.7 & 9.00 & 1.00 & 1.00 & 0.88 & 0.00 & 5.19 & 3.74 \\ 
			\hline
			Ad-DPD-LASSO $\gamma = $ 1 & 9.00 & 1.00 & 1.00 & 1.56 & 0.00 & 6.79 & 4.25 \\ 
			AW-DPD-LASSO $\gamma = $ 0.1 & 9.58 & 1.00 & 1.00 & 0.39 & 1.77 & 4.99 & 3.53 \\ 
			AW-DPD-LASSO $\gamma = $ 0.3 & 9.47 & 1.00 & 1.00 & 0.40 & 1.34 & 5.38 & 3.59 \\ 
			AW-DPD-LASSO $\gamma = $ 0.5 & 9.15 & 1.00 & 1.00 & 0.45 & 0.60 & 5.47 & 3.66 \\ 
			AW-DPD-LASSO $\gamma = $ 0.7 & 9.43 & 1.00 & 1.00 & 0.56 & 1.39 & 8.76 & 3.85 \\ 
			AW-DPD-LASSO $\gamma = $ 1 & 9.43 & 0.98 & 1.00 & 11.60 & 8.13 & 17.07 & 5.49 \\ 
			\hline
		\end{tabular}
	} }
	\end{table}

	\subsection{Runtime Comparisons \label{sec:runtime}}
	
	To illustrate the computational advantages of the proposed robust  adaptive methods, we compare the computation times of the proposed and competitive methods in the simulation  Settings A and B presented in the paper  (refer to Section \ref{sec:simulationsettings}).
	In particular, we compare the mean running time of the non-concave penalty SCAD and its corresponding weighted adaptive penalty, both combined with the robust DPD loss. To fairly compare both methods, at each step of the iterative algorithm the auxiliary weighted least squares problem is solved using  the same function, namely \textit{ncvreg}, implemented in the R package \textit{robustHD}. Further, both methods are initialized with same initial estimates obtained from robust RLARS and the optimal tuning parameter $\lambda$ is chosen applying HBIC criterion over a grid.  The $\lambda$ grid is not pre-specified, but the maximum value of the grid is defined based on the data and the method;	maximum $\lambda$ is first computed using bivariate winsorization to estimate the smallest value of the LASSO tuning parameter for sLTS regression so that sets all coefficients to zero, and then this bound is optimized to the smallest value shrinking all coefficients toward zero for the corresponding method. The length of the grid is fixed to be $50$ for all methods. Table C16 in Appendix on the supplementary material reports computational times taken in all the process of fitness, including the optimal grid choice and the fitness of the methods with the optimal $\lambda.$ 
	
	Further, we compare computation times of the different methods for increasing $p$ under the  contamination setting A  as specified in Section \ref{sec:simulationsettings}.	We compare four competing methods based on the DPD loss, namely DPD-LASSO, Ad-DPD-LASSO, AW-DPD-LASSO and DPD-ncv. 	We fix the value of penalty parameter $\lambda$ (based on our simulation experiences) so that  all methods perform 	suitably for model selection and parameter estimation. For such purpose, we have taken  $\lambda = 0.025$ for LASSO-based methods, DPD-LASSO and Ad-DPD-LASSO, and $\lambda = 0.078$ for the AW-DPD-LASSO and DPD-ncv methods with the nonconcave penalty SCAD.  All methods are initialized using RLARS, and the sample size is fixed to $n=100.$

	Figure \ref{figure:timecurves} shows computation time curves (in seconds) of the four proposed methods under three scenarios of contamination, namely pure data, $10\%$ of $Y$-outliers and $10\%$ of $\boldsymbol{X}$-outliers with $R=100$ replications.
	For a better comparison between competing methods, computational time curves for LASSO-based methods are plotted on the left, 
	and computational time curves for AW-DPD-LASSO and DPD-ncv (based on SCAD) are plotted on the right part of Figure \ref{figure:timecurves}.	These plots show smaller computational  times for the proposed robust  adaptive methods, in comparison to their non-adaptive competitors, for any given value of the tuning parameter $\gamma$.  	Furthermore, such computational advantage becomes more significant with large values of $p.$
	In particular Ad-DPD-LASSO is generally faster than DPD-LASSO for all values of $\gamma,$ but in presence of $Y$-outliers its average computation time increases slightly 	approaching to the runtimes took by classical DPD-LASSO.	On the other hand, the AW-DPD-LASSO method, which weights are chosen to linearly approximate the nonconcave SCAD penalty, is faster than the DPD-ncv in all scenarios. 
	Therefore, in the cases with large values of the ratio $p/n$, of the number of explanatory variables to the sample size (ultra-high dimensional settings), 	the proposed AW-DPD-LASSO method is an appealing alternative to the DPD-ncv estimator for a faster computation without efficiency and robustness loss.
	\begin{figure}
		\begin{subfigure}[]{\textwidth}
			\centering
			\includegraphics[height=6.5cm, width=7cm]{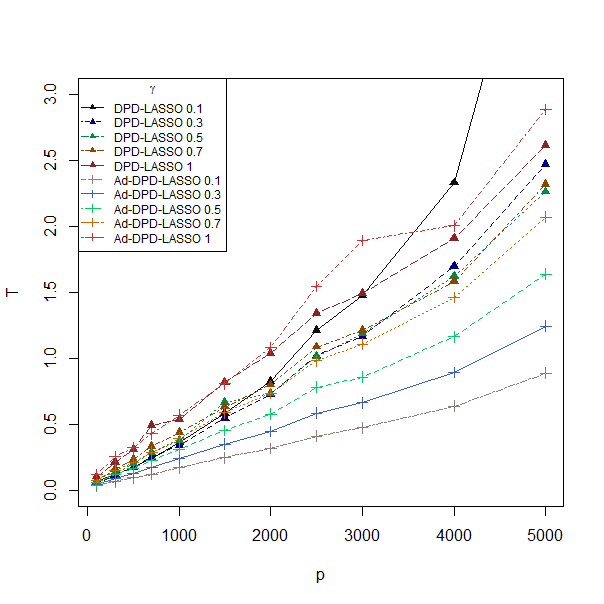} 
			\includegraphics[height=6.5cm, width=7cm]{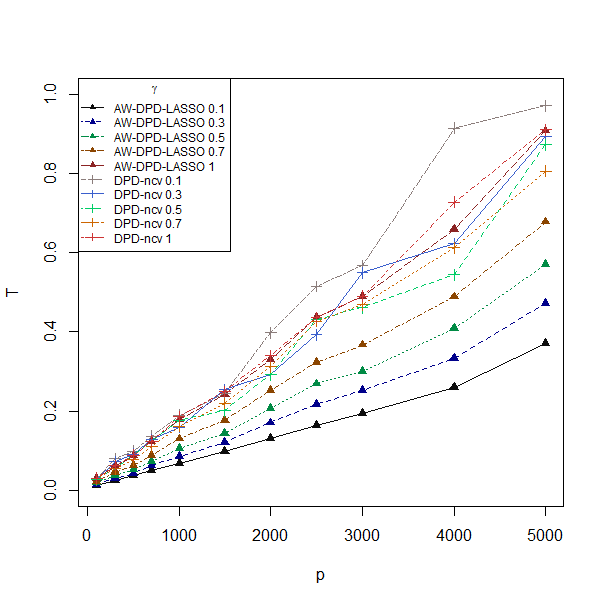}
			\caption{Pure data}
		\end{subfigure}
		\begin{subfigure}[]{\textwidth}
			\centering
			\includegraphics[height=6.5cm, width=7cm]{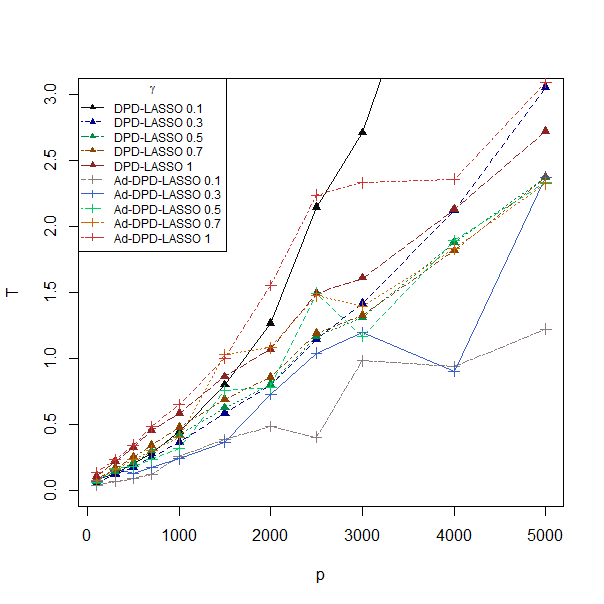} 
			\includegraphics[height=6.5cm, width=7cm]{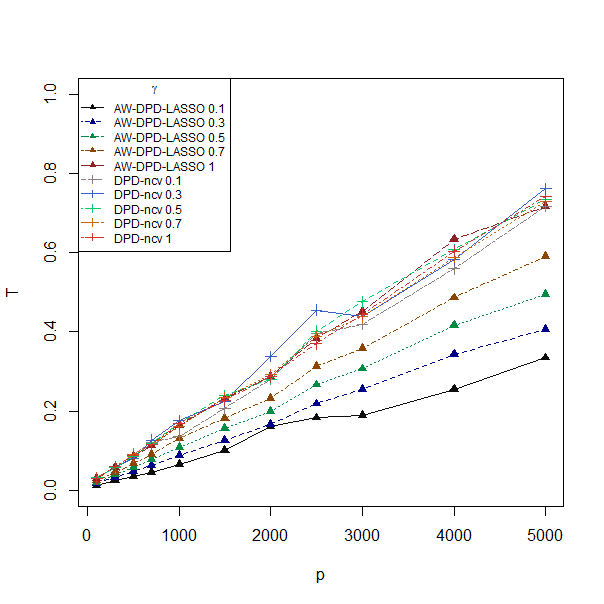}
			\caption{$Y$-outliers}
		\end{subfigure}
		\begin{subfigure}[]{\textwidth}
			\centering
			\includegraphics[height=6.5cm, width=7cm]{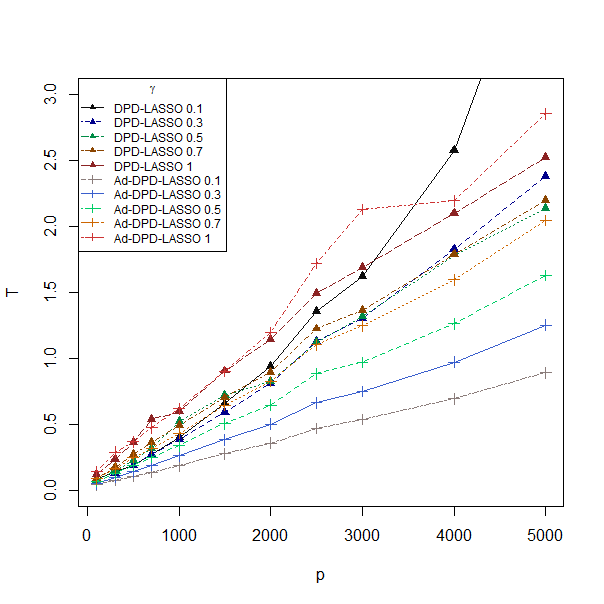} 
			\includegraphics[height=6.5cm, width=7cm]{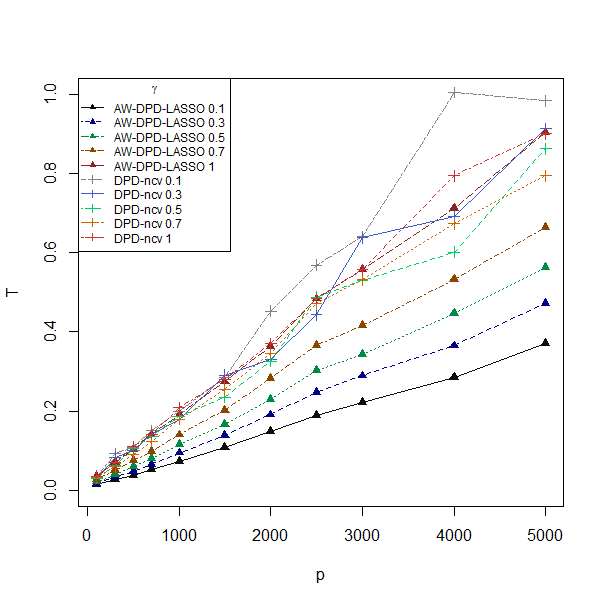}
			\caption{$\boldsymbol{X}$-outliers}
		\end{subfigure}
		\caption{Computational time curves (in seconds) for different penalized procedures under different simulation settings.}
		\label{figure:timecurves}
	\end{figure}

	\subsection{On the Choice of robustness tuning parameter $\gamma$ \label{sec:choice}}
	
		We would like to point out that the choice of robustness tuning parameters ($\gamma$ in the present case with DPD) 
		is an important practical issue in the classical (low-dimensional) robustness literature, 
		since larger values of the tuning parameter produce more robust but less efficient estimators. 
		A moderately large value $\gamma$ around 0.3 to 0.5  generally offers a decent trade-off between efficiency and robustness  
		in most practical use of the DPD-based robust inferential procedures.
		However, an optimal value of  $\gamma$ would depend on the contamination level in data, which is quite difficult to measure in real applications. 
		Several criteria have been proposed in the literature for choosing optimal values of the DPD tuning parameter for classical (low-dimensional) statistical models. 
		Most popularly,  Warwick and Jones \cite{warwick2005} introduced a useful data-based procedure for IID data based on 
		the minimization of an estimate of the  asymptotic MSE of the minimum DPD estimators based on the given sample data, 
		which is later extended and extensively studied for non-homogeneous models, including general class of parametric regressions, 
		by Ghosh and Basu \cite{ghosh1, ghosh2016}. 
		This method depends on the choice of a pilot estimator which can have a significant impact on the optimal tuning parameter choice, as the pilot invariably draws the final estimator towards itself. 
		Basak et al. \cite{Basak2021} improved the method by alleviating the dependency on the initial estimator, 
		by proposing an iterative algorithm which updates the pilot estimator at each step with the optimal estimate obtained till then and the process is repeated until there is no further change in the optimal estimate.

		For the present cases with the ultra-high dimensional setting, the computation of optimal $\gamma$ at each step would inevitably involve a high-time cost. 	This cost would not benefit us much as we have seen through our extensive simulation exercises in this paper. 	The model section and estimation performances of the proposed Ad-DPD-LASSO and AW-DPD-LASSO does not depend too crucially 
		on the choice of $\gamma$ within the range [0.3, 0.7] (unlike the traditional robustness literature on parameter estimation), 
		and hence, it does not justify the time-cost for selection of $\gamma$ during model selection through these procedures; 
		any moderate value of $\gamma$ around 0.5 should provide reasonably good model selection performances in any practical applications.
		However, if one wants to bear the high computational cost of choosing a data-driven value of $\gamma$ for a high-dimensional dataset,
		the procedure of Warwick and Jones \cite{warwick2005} or Basak et al. \cite{Basak2021}  may be extended suitably for the high-dimensional settings
		or a similar procedure as in \cite{Fan2017} may be developed for the present case of DPD based loss function. 
		Considering the length of the current manuscript, we have left the detailed exploration of  
		such data-driven tuning parameter selection procedures for our future research.

	\section{Real data analysis \label{SEC:realdata}}
	
	Robust high-dimensional regression is highly important  in the field of chemometrics, 
	where hundreds or even thousands of spectra need to be analyzed. 
	We apply our proposed methods to a real dataset regarding electron-probe X-ray microanalysis (EPXMA)	of archaeological glass vessels from the 16th and 17th centuries, where each of $n = 180$ glass vessels is represented by a spectrum on 1920 frequencies. 
	For each vessel the contents of thirteen chemical compounds are registered. 
	These data were first introduced in  Janssens et al. \cite{Janssens}, where the archaeological glass vessels were investigated through chemical analysis. 
	However, it was realized that some spectra in the dataset had been measured with a different detector efficiency, 
	which in the statistical sense, may lead to bad leverage points (outliers in the covariates space). 
	Besides leverage points, there are also four different material compositions of the glass vessels, 
	further increasing the inhomogeneity of the spectral data.
	These data have been used to identify multivariate outliers by Filzmoser et al. \cite{Filzmoser}
	and subsequently in Serneels et al. \cite{serneels}, Maronna \cite{Ma} and Smucler and Yohai  \cite{smu}
	to illustrate high-dimensional robust regression methods using the content of PbO (lead monoxide)  as response variable. 
	However, because the data on PbO (and any of its Box-Cox transformations) are not normally distributed, we have instead opted to use the seventh variable in the data, which is the content of chlorine (Cl),  as response variable. For this response variable (Cl content), the Shapiro-Wilk normality test yielded a p-value of 0.6552, indicating that the data can be assumed to follow a normal distribution and hence our proposed algorithm (in Section \ref{SEC:computation}) for the computation of the AW-DPD-LASSO estimator can be used here. 
	Accordingly, a linear model with normal error distribution  is fitted with this response data (Cl content) and considering  the frequencies as covariates. 
	
	Since the frequencies below 15 and above 500 have mostly the values of zero, we keep frequencies from 15 to 500 in our modelling, 
	so that we have $p = 486$ covariates.  We estimate the coefficients of the regression model fitted using the adaptive lasso penalties 
	Ad-LS-LASSO, LS-SCAD,	and the four DPD based methods penalized with adaptive lasso penalties, namely the Ad-DPD-LASSO, AW-DPD-LASSO. For comparison purposes with the AW-DPD-LASSO we also all DPD-ncv.
	To study the performance of the different methods, Maronna  \cite{Ma} used 10$\%$ trimmed root mean square error, RMSE(0.9), 
	which is a more robust criterion than the usual RMSE. Using this measure prevents the outliers from inflating the RMSE.
	To compare the precedent estimating methods, we report the model size (MS), RMSE(0.9), and the minimum and maximum error (MAX and MIN)
	in Table C17 of the Appendix. The robust DPD-based methods present greater maximum error (MAX) than non-robust ones, as well as higher MSE and MSPE(0.9), as outliers lead to larger residuals in a robust fit. 	Indeed, it can be noticed that  the difference between the RMSEs(0.9), are more pronounced than for MSEs, showing again the great gain in robustness.
	
	\begin{figure}[!h]
		\centering
		\includegraphics[height=7cm, width=10.7cm]{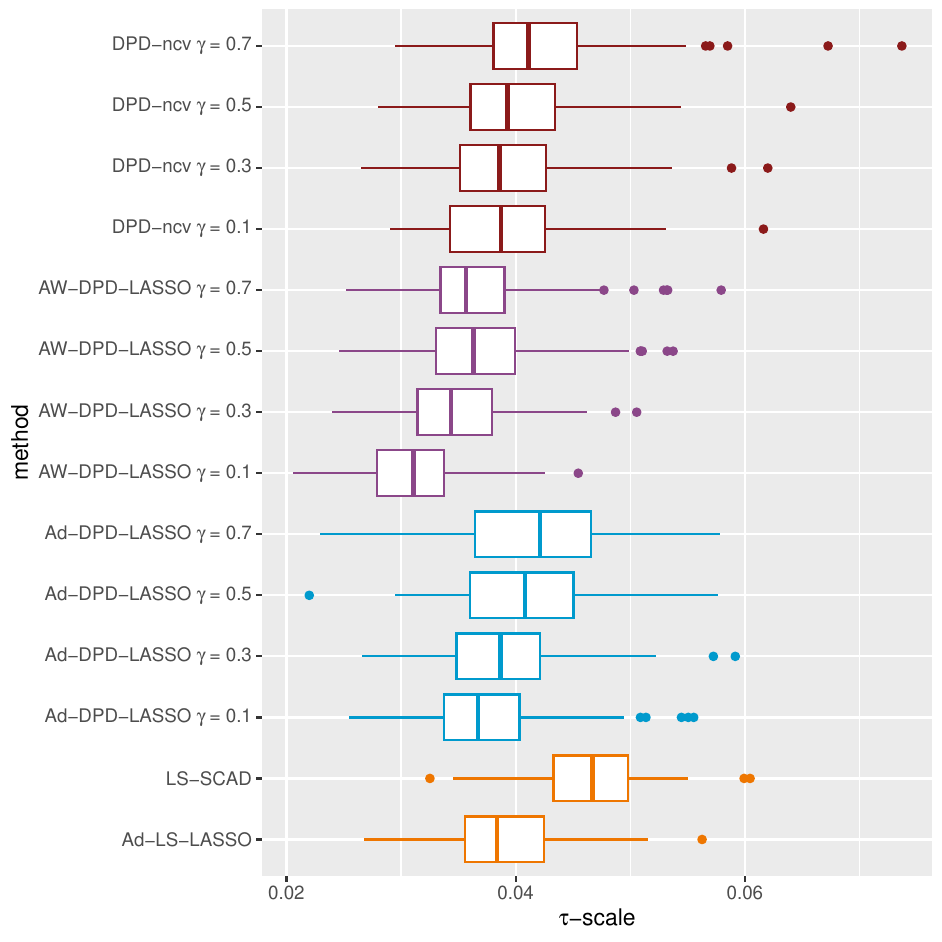}
		\caption{Box-plots of the $\tau-$scale of the prediction residuals, for the EPXMA data}
		\label{boxplot}
	\end{figure}

	Further, to asses the prediction performance of different methods on these data, 
	Smucler and Yohai \cite{smu} proposed using the $\tau$-scale of the residuals, 	calculated as in Maronna and Zamar  \cite{MaronnaZamar}. 
	To define this $\tau$-scale for a univariate sample $\boldsymbol{x} = (x_1,..x_n)$, we consider the function 
	$$
	W_c(x) = \left(1-\left(\frac{x}{c}\right)^2\right)^2 \operatorname{I}(|x|\leq c), \hspace{0.3cm} \text{and put } 
	\hspace{0.3cm} 
	w_i = W_{c_1}\left(\frac{x_i-\operatorname{med}(\boldsymbol{x})}{\sigma_0}\right), ~~ i=1, \ldots, n,
	$$
	where  $\sigma_0$ is  the median absolute deviation of the sample $\boldsymbol{x}$.
	Then, the $\tau-$scale statistic for $\boldsymbol{x}$ is defined as 
	$\sigma(\boldsymbol{x}) = \frac{\sigma_0^2}{n} \sum_i^n\rho_{c_2}\left(\frac{x_i-\mu(x)}{\sigma_0}\right)$, 
	where $\rho_c(x) = \operatorname{min}(x^2,c^2)$ and $\mu(\boldsymbol{x}) = \frac{\sum x_iw_i}{\sum w_i}$ is a weighted mean.
	To combine robustness and efficiency, Maronna and Zamar  \cite{MaronnaZamar} suggested to take $c_1 = 4.5$ and $c_2=3$, 
	which yield approximately $80\%$ efficient univariate location and scale estimation for normal data. 
	In our case, we randomly split the data into a train set ($n=120$) used to fit the model and 
	a test set ($n=60$) used to calculate prediction residuals and their $\tau$-scale statistics.
	We apply the same schema for all the methods adaptive considered, namely the 
	Ad-LS-LASSO, LS-SCAD, 	Ad-DPD-LASSO, AW-DPD-LASSO and DPD-ncv, and repeat the process $R = 100$ times.
	The box-plots of the $\tau$-scale statistics of the prediction residuals, obtained by different methods, across the 100 replications 
	are presented in Figure \ref{boxplot}; the median value of the $\tau$-scales are reported in Table C18 of the Appendix available as Supplementary Material. It is again evident that DPD-based methods produce lower error than the least-squares-based method in the test set for both adaptive penalties,-- adaptive lasso and weighted adaptive lasso.
	Note that the $\tau-$scale measure is generally greater for larger values of  $\gamma$; moderate and low values of $\gamma$ offer a great trade-off between efficiency and robustness.
	Moreover, the performance of AW-DPD-LASSO (with SCAD penalty) is sufficiently close (or even better) to that of the corresponding DPD-ncv estimator,	which again justifies the usefulness of our proposed general adaptive DPD-LASSO estimators as a fast alternative to the  DPD-ncv approach.
	Additionally, it is interesting to note that the weighted adaptive lasso penalty performs better than the standard adaptive lasso with our robust DPD loss, but does not do the same with the least-squares loss.

	\section{Conclusions}
	 
	In this paper we have presented a new robust adaptive LASSO method based on the DPD-loss function for the LRM under the  ultra-high dimensional set-up. 	This election of the loss grants high robustness against outliers in the data, while the use of adaptive LASSO penalty ensures the oracle property,	and hence it performs as well as if the true underlying model was given in advance. 
	Further, the computation of the proposed AW-DPD-LASSO estimator can be solved by using the same efficient algorithms as the DPD-LASSO method, 	just properly transforming the data. Through an extensive simulation study, 	it has been shown that the use of AW-DPD-LASSO method improves the accuracy of the parameters estimation 
	(both regression coefficients and error variance) compared to several other existing robust and non-robust methods.	
	This advantage is accentuated for the estimation of the error standard deviation.
	
	We should note that, although the theoretical results are derived for any general error distribution from the location-scale family,
	the computational algorithm for our proposed estimators is developed only for the common case of normal error distribution. 
	It would be an important future research to appropriately modify the proposed algorithm for the computation of our AW-DPD-LASSO  estimator 	with any other specified (non-normal) error density.	
	Additionally, the simplicity and usefulness of the adaptive DPD framework encourages for its extension to other parametric regression models.
	In particular, in future works,  our interest will be to consider the general adaptively weighted  DPD-LASSO approach 
	for the binary and multiple logistic regression models, as well as for Poisson regression model, under the ultra-high dimensional set-up.
	Their extension to the ultra-high dimensional generalized linear models will also be an important future work.
	
	\section*{Acknowledgements}
	Authors are grateful to the anonymous referees and the  Editor for their insightful comments and suggestions that have improved the manuscript significantly.
	This research is supported by the Spanish Grants PGC2018-095 194-B-100 and  FPU 19/01824. Research of AG is also partially supported by an INSPIRE Faculty Research Grant and a grant (No. SRG/2020/000072) from SERB, both under the  Department of Science and Technology, Government of India, India. \\
	M.Jaenada and L.Pardo are members of the Interdisciplinary Mathematics Institute. The authors have no conflict of interest, financial or otherwise.

\bibliographystyle{plain}
\bibliography{bibliography.bib}

\appendix
\section{Proof of the theorems}
\label{APP:proof}

\subsection{Proof of Theorem 4.1}
\label{APP:proof.1}

Let us consider the set-up and notation of Section 4 and define, for given $M>0$, 
the set 
$$\mathcal{B}_0(M) =\left\{ \boldsymbol{\beta}\in\mathbb{R}^p : ||\boldsymbol{\beta} - \boldsymbol{\beta}_0 ||_2\leq M, 
Supp(\boldsymbol{\beta}) \subseteq Supp(\boldsymbol{\beta}_0) = S_0 \right\},$$ 
and the function
\begin{eqnarray}
	Z_n(M) = \sup\limits_{\boldsymbol{\beta}\in \mathcal{B}_0(M)} \bigg| 
	\left(L_{n,\gamma }(\boldsymbol{\beta }) - L_{n,\gamma }(\boldsymbol{\beta}_0)\right)
	-E\left(L_{n,\gamma }(\boldsymbol{\beta }) - L_{n,\gamma }(\boldsymbol{\beta}_0)\right)\bigg|.
	\label{EQ:Z_nM}
\end{eqnarray}
Put $L=(1+\gamma)L_{\gamma}/\gamma$. Then, we start by proving the following lemma
which provides us the convergence rate for $Z_n(M)$ to be used in the proof of the main theorem.

\begin{lemma}\label{LEM:L1}
	Under Assumptions (A1)--(A2), we have the following result for any $t>0$.
	\begin{eqnarray}
		P\left(Z_n(M) \geq 4ML\sqrt{s/n} + t \right) \leq e^{-\frac{nc_0t^2}{8M^2L^2}}.  
	\end{eqnarray}
\end{lemma}

\begin{proof}
	Let $W_1, \ldots, W_n$ be a Rademacher sequence, independent of our model components.
	Then, an application of the symmetrization theorem  (Theorem 14.3, B\"{u}hlmann  and van de Geer  \cite{buhlmann}) yield
	\begin{eqnarray}
		E[Z_n(M)] &\leq& 2 E \left[\sup\limits_{\boldsymbol{\beta}\in \mathcal{B}_0(M)} \left| \frac{1}{n} \sum_{i=1}^n W_i
		\left(\rho_{\gamma }(\boldsymbol{x}_i^T\boldsymbol{\beta}, y_i) - \rho_{\gamma }(\boldsymbol{x}_i^T\boldsymbol{\beta}_0, y_i)\right)\right|\right]
		\nonumber\\
		&\leq& \frac{2}{n} E \left[\sup\limits_{\boldsymbol{\beta}\in \mathcal{B}_0(M)} \sum_{i=1}^n |W_i|
		\left|\rho_{\gamma }(\boldsymbol{x}_i^T\boldsymbol{\beta}, y_i) - \rho_{\gamma }(\boldsymbol{x}_i^T\boldsymbol{\beta}_0, y_i)\right|\right].
	\end{eqnarray}
	But, by Assumption (A1), the function $\rho_{\gamma}(\cdot, y_i)$ is Lipschitz 
	for each $i=1, \ldots, n$, with the common Lipschitz constant $L$ independent of $i$. 
	Hence, we have
	\begin{eqnarray}
		\left|\rho_{\gamma }(\boldsymbol{x}_i^T\boldsymbol{\beta}, y_i) - \rho_{\gamma }(\boldsymbol{x}_i^T\boldsymbol{\beta}_0, y_i)\right| \leq L \left|\boldsymbol{x}_i^T\boldsymbol{\beta} - \boldsymbol{x}_i^T\boldsymbol{\beta}_0\right|,
		\label{EQ:L1}
	\end{eqnarray}
	and we further get
	\begin{eqnarray}
		E[Z_n(M)] 
		&\leq&\frac{2L}{n} E \left[\sup\limits_{\boldsymbol{\beta}\in \mathcal{B}_0(M)} \sum_{i=1}^n 
		\left|\boldsymbol{x}_i^T\boldsymbol{\beta} - \boldsymbol{x}_i^T\boldsymbol{\beta}_0\right|\right]
		\nonumber\\
		&=& \frac{2L}{n} E \left[\sup\limits_{\boldsymbol{\beta}\in \mathcal{B}_0(M)} \sum_{j=1}^s  \sum_{i=1}^n |x_{ij}|
		\left|\beta_j - \beta_{j0}\right|\right]
		\nonumber\\
		&\leq& \frac{2L}{n}  \sup\limits_{\boldsymbol{\beta}\in \mathcal{B}_0(M)} 
		\left||\boldsymbol{\beta} - \boldsymbol{\beta}_0|\right|
		E\left[  \sum_{j=1}^s \left|\sum_{i=1}^n |x_{ij}|\right|^2\right]^{1/2}
		\nonumber\\
		&& ~~~~~~~~~~~~~~~~~~~~~~~~~~~~~~~~~~~~~~~~~~\mbox{[By Cauchy-Schwarz inequality]}
		\nonumber\\
		&\leq& \frac{2ML}{n}
		\left[  \sum_{j=1}^s E\left|\sum_{i=1}^n |x_{ij}|\right|^2\right]^{1/2}
		~~~~~\mbox{[By the Jensen's inequality]}
		\nonumber\\
		&=& \frac{2ML}{n} \sqrt{sn} = 2ML\sqrt{s/n}.
		\label{EQ:L1_1}
	\end{eqnarray}
	Now, let us define $Z_{i,\gamma}=\frac{\sqrt{c_0}}{LM}\left[\rho_{\gamma }(\boldsymbol{x}_i^T\boldsymbol{\beta}, y_i) - \rho_{\gamma }(\boldsymbol{x}_i^T\boldsymbol{\beta}_0, y_i)\right]$. 
	Then, using (\ref{EQ:L1}) and Assumption (A2), we have, 
	for any $\boldsymbol{\beta}=(\boldsymbol{\beta}_1^T, \boldsymbol{0}^T)^T \in \mathcal{B}_0(M)$,
	\begin{eqnarray}
		\frac{1}{n}\sum_{i=1}^n Z_{i,\gamma}^2
		&\leq& \frac{1}{n}\frac{c_0}{L^2M^2} L^2\sum_{i=1}^n 
		\left[\boldsymbol{x}_i^T(\boldsymbol{\beta} - \boldsymbol{\beta}_0)\right]^2
		\nonumber\\
		&=& \frac{c_0}{M^2}\frac{1}{n}(\boldsymbol{\beta}_1 - \boldsymbol{\beta}_{10})^T[\boldsymbol{X}_1^T\boldsymbol{X}_1](\boldsymbol{\beta}_1 - \boldsymbol{\beta}_{10})
		\nonumber\\ 
		&\leq&  \frac{c_0}{M^2} \frac{1}{c_0} ||\boldsymbol{\beta}_1 - \boldsymbol{\beta}_{10}||_2^2 
		\leq 1.
	\end{eqnarray}
	But, $Z_n(M) = \frac{LM}{\sqrt{c_0}}\sup\limits_{\boldsymbol{\beta}\in \mathcal{B}_0(M)} \bigg| 
	\frac{1}{n}\sum_{i=1}^n\left(Z_{i,\gamma} - E(Z_{i,\gamma})\right)\bigg|$. 	
	Therefore, using the above result, 	
	along the Lipschitz property of $\rho_{\gamma}$ to apply Massart's concentration theorem (Massart, 2000),	
	we get 
	$$
	P\left(\frac{\sqrt{c_0}}{LM} Z_n(M) \geq \frac{\sqrt{c_0}}{LM} E[Z_n(M)] + \frac{\sqrt{c_0}}{LM}t \right)
	\leq e^{-\frac{nc_0t^2}{8M^2L^2}}.
	$$
	Combining it with (\ref{EQ:L1_1}), we finally get the desired lemma as follows.
	\begin{eqnarray}
		P\left(Z_n(M) \geq 2ML\sqrt{s/n} + t \right) 
		\leq P\left(Z_n(M) \geq E[Z_n(M)] + t \right)
		\leq e^{-\frac{nc_0t^2}{8M^2L^2}}.  
	\end{eqnarray}
\end{proof}

	\begin{proof}[Proof of the theorem]
		Consider the set $\mathcal{B}_0(M)$ as defined above with $M=o(\kappa_n^{-1}s^{-1/2})$.
		Then, a Taylor series expansion of $L_{n,\gamma}(\boldsymbol{\beta})$ at any 
		$\boldsymbol{\beta}=(\boldsymbol{\beta}_1^T, \boldsymbol{0}^T)^T\in\mathcal{B}_0(M)$ around $\boldsymbol{\beta}_0$, 
		along with Assumption (A4), yields 
		\begin{eqnarray}
			E\left[L_{n,\gamma}(\boldsymbol{\beta}) - L_{n,\gamma}(\boldsymbol{\beta}_0)\right]
			&=& \frac{1}{2n}(\boldsymbol{\beta}_1 - \boldsymbol{\beta}_{10})^TE[\boldsymbol{X}_1^T\boldsymbol{H}_{\gamma}^{(2)}(\boldsymbol{\beta}_0)\boldsymbol{X}_1] 
			(\boldsymbol{\beta}_1 - \boldsymbol{\beta}_{10})
			+ o(1), \hspace{1cm}
			\label{EQ:A.1}
		\end{eqnarray}
		since the first order derivative of $E[L_{n,\gamma}(\boldsymbol{\beta})]$ is zero at $\boldsymbol{\beta}=\boldsymbol{\beta}_0$.
		Now, by Assumptions (A2) and (A3),
		we have 
		$$
		\frac{1}{n}(\boldsymbol{\beta}_1 - \boldsymbol{\beta}_{10})^TE[\boldsymbol{X}_1^T\boldsymbol{H}(\boldsymbol{\beta}_0)\boldsymbol{X}_1] 
		(\boldsymbol{\beta}_1 - \boldsymbol{\beta}_{10}) \geq c_0c_1||\boldsymbol{\beta}_1 - \boldsymbol{\beta}_{10}||_2^2,
		$$
		for sufficiently large $n$, which implies
		\begin{eqnarray}
			E\left[L_{n,\gamma}(\boldsymbol{\beta}) - L_{n,\gamma}(\boldsymbol{\beta}_0)\right]
			&\geq& \frac{1}{2} c_0c_1||\boldsymbol{\beta}_1 - \boldsymbol{\beta}_{10}||_2^2. 
			\label{EQ:ThP.1}
		\end{eqnarray}
		
		Next, we will proceed as in the proof of Theorem 1 of  Fan et al.
		\cite{Fan/etc:2014} to prove the first part of our theorem;
		the second part follows from there in a straightforward manner and hence omitted for brevity. 
		Note that, the oracle parameter value  $\widehat{\boldsymbol{\beta}}^o$ may not belongs to $\boldsymbol{\beta}_o(M)$
		and hence may not satisfy (\ref{EQ:ThP.1}). So, we define 
		$\widetilde{\boldsymbol{\beta}}_1 = u \widehat{\boldsymbol{\beta}}_1^o + (1-u) \boldsymbol{\beta}_{10}$
		with $u = M/(M+||\widehat{\boldsymbol{\beta}}_1^o - \boldsymbol{\beta}_{10}||_2)$ so that 
		$\widetilde{\boldsymbol{\beta}} = (\widetilde{\boldsymbol{\beta}}_1^T, \boldsymbol{0}^T)^T \in \mathcal{B}_0(M)$ and satisfies (\ref{EQ:ThP.1}).
		But, by convexity of our objective function $Q_{n,\gamma, \lambda}(\boldsymbol{\beta})$, it can be easily shown that 
		$Q_{n,\gamma, \lambda}(\widetilde{\boldsymbol{\beta}}) \leq Q_{n,\gamma, \lambda}(\boldsymbol{\beta}_0)$
		and hence 
		\begin{eqnarray}
			E\left[L_{n,\gamma}(\widetilde{\boldsymbol{\beta}}) - L_{n,\gamma}(\boldsymbol{\beta}_0)\right]
			&\leq& Z_n(M) + \lambda_n \sum\limits_{j=1}^{s} w_{j}\left\vert \widetilde{\beta}_{j} - \beta_{j0}\right\vert
			\nonumber\\
			&\leq& Z_n(M) + \lambda_n ||\boldsymbol{w}_0||_2 ||\widetilde{\boldsymbol{\beta}}_1 - \boldsymbol{\beta}_{01}||_2
			\leq Z_n(M) + \lambda_n ||\boldsymbol{w}_0||_2 M,\nonumber
		\end{eqnarray}
		where the second last inequality holds by Cauchy-Schwarz inequality.
		Therefore, on the event $\mathcal{E}_n = \left\{Z_n(M) \leq 2MLn^{-1/2}\sqrt{s\log n} \right\}$, we get 
		\begin{eqnarray}
			E\left[L_{n,\gamma}(\widetilde{\boldsymbol{\beta}}) - L_{n,\gamma}(\boldsymbol{\beta}_0)\right]
			&\leq& 2MLn^{-1/2}\sqrt{s\log n} + \lambda_n ||\boldsymbol{w}_0||_2 M,
			\label{EQ:ThP.2}
		\end{eqnarray}
		
		Now, take $L=(1+\gamma)L_{\gamma}/\gamma$ so that, by applying Lemma \ref{LEM:L1} with $t=ML\sqrt{s(\log n)/n}$, 
		we get for large enough $n$ (satisfying $\log n \geq 4$)
		\begin{eqnarray}
			P(\mathcal{E}_n^c) &=& P\left(Z_n(M) > 2MLn^{-1/2}\sqrt{s\log n}\right)
			\nonumber\\
			&\leq& P\left(Z_n(M) > 2ML \sqrt{s/n} + t\right) \leq \exp(-{c_0s(\log n)}/{8}).
			\nonumber
		\end{eqnarray}
		The first inequality follows because $2ML\sqrt{s/n} + t < 2MLn^{-1/2}\sqrt{s\log n}$.
		Hence, 
		$$
		P(\mathcal{E}_n) \geq 1 - \exp(-{c_0s(\log n)}/{8}) = 1 - n^{-c_0s/8}.  
		$$
		Note that we can assume $L>1$ without any loss of generality. \\
		Also, take $M=L^{-1}\left[2\sqrt{s/n} + \lambda_n ||\boldsymbol{w}_0||_2\right]$ so that it satisfies $M=o(\kappa_n^{-1}s^{-1/2})$
		by Assumption (A2) and the fact that $\lambda_n||\boldsymbol{w}_0||_2\sqrt{s}\kappa_n \rightarrow 0$.
		Therefore, combining (\ref{EQ:ThP.1}) for $\boldsymbol{\beta}_1 = \widetilde{\boldsymbol{\beta}}_1$ 
		and (\ref{EQ:ThP.2}), on the event $\mathcal{E}_n$, we get that 
		\begin{eqnarray}
			\frac{1}{2} c_0c_1||\widetilde{\boldsymbol{\beta}}_1 - \boldsymbol{\beta}_{10}||_2^2
			&\leq& \left(2n^{-1/2}\sqrt{s\log n} + \lambda_n ||\boldsymbol{w}_0||_2\right) ML
			\nonumber\\
			&\leq& \left(2n^{-1/2}\sqrt{s\log n} + \lambda_n ||\boldsymbol{w}_0||_2\right) \left(2\sqrt{s/n} + \lambda_n ||\boldsymbol{w}_0||_2\right),
			\nonumber
		\end{eqnarray}
		and hence 
		\begin{eqnarray}
			||\widetilde{\boldsymbol{\beta}}_1 - \boldsymbol{\beta}_{10}||_2
			&\leq& O\left(\sqrt{s(\log n)/n} + \lambda_n ||\boldsymbol{w}_0||_2\right),
		\end{eqnarray}
		But $||\widetilde{\boldsymbol{\beta}}_1 - \boldsymbol{\beta}_{10}||_2\leq M/2$, and hence
		$$
		\frac{M ||\widehat{\boldsymbol{\beta}}^o_1 - \boldsymbol{\beta}_{10}||_2}{
			M + ||\widehat{\boldsymbol{\beta}}^o_1 - \boldsymbol{\beta}_{10}||_2} \leq \frac{M}{2},
		~~~~\mbox{ or }~~
		||\widehat{\boldsymbol{\beta}}^o_1 - \boldsymbol{\beta}_{10}||_2 \leq M 
		= O\left(2\sqrt{s/n} + \lambda_n ||\boldsymbol{w}_0||_2\right).
		$$
		Then, it follows that
		$$
		||\widehat{\boldsymbol{\beta}}^o_1 - \boldsymbol{\beta}_{10}||_2
		\leq O\left(\sqrt{s(\log n)/n} + \lambda_n ||\boldsymbol{w}_0||_2\right),
		$$
		on the event $\mathcal{E}_n$ having probability at least $1-n^{-cs}$ with $c=c_0/8$.
		This completes the proof of the first part of the theorem.
	\end{proof}

	\subsection{Proof of Theorem 4.2}
	\label{APP:proof.2}
	
	Consider the set-up and notation of Section 4 and Subsection \ref{APP:proof.1} above.
	Note that, by Assumption (A1), the function $|\psi_{1, \gamma}|$ is bounded and let us denote its upper bound by $R_{\gamma}$.
	For any given $r>0$, define a ball in $\mathbb{R}^s$ around $\boldsymbol{\beta}_0$ as
	$$
	\mathcal{N}(r) = \left\{ \boldsymbol{\beta}=(\boldsymbol{\beta}_1^T, \boldsymbol{\beta}_2^T)^T\in\mathbb{R}^p : \boldsymbol{\beta}_2 = \boldsymbol{0}_{p-s}, 
	||\boldsymbol{\beta}_1 - \boldsymbol{\beta}_{10} ||_2\leq r \right\},
	$$ 
	and the functional space $\Gamma_j(r) = \left\{ h_{j,\boldsymbol{\beta}} : \boldsymbol{\beta} \in \mathcal{N}(r) \right\}$, 
	where
	\begin{align*}
		h_{j,\boldsymbol{\beta}}(\boldsymbol{x}_i, y_i)  = x_{ij} (1+\gamma)\big[&\psi_{1, \gamma}(y_i - \boldsymbol{x}_i^T\boldsymbol{\beta})
		-\psi_{1, \gamma}(y_i - \boldsymbol{x}_i^T\boldsymbol{\beta}_0)\\
		&- E\left(\psi_{1, \gamma}(y_i - \boldsymbol{x}_i^T\boldsymbol{\beta}) - \psi_{1, \gamma}(y_i - \boldsymbol{x}_i^T\boldsymbol{\beta}_0)\right)\big],
	\end{align*}
	
	for $i=1, \ldots, n$ and $j=1, \ldots, p$.
	We again start with an useful Lemma and subsequently prove the main theorem using it.
	
	\begin{lemma}\label{LEM:L2}
		Consider the ball $\mathcal{N}=\mathcal{N}(C_1\delta_n)$ for some sequence $\delta_n \rightarrow 0$ satisfying $\kappa_n\delta_n^2 = o(\lambda_n)$ 
		and assume that, for $j=s+1, \ldots, p$, the functional space $\Gamma_j=\Gamma_j(\delta_n)$ has the covering number 
		$N(\cdot, \Gamma_j, ||\cdot||_2)$ satisfying
		$$
		\log(1+N(2^{-k}R_{\gamma}, \Gamma_j, ||\cdot||_2)) \leq A_{j,n} 2^{2k}, ~~~~~~~~~~~0\leq k \leq (\log_2 n)/2, 
		$$
		for some constant $A_{j,n}>0$ satisfying  $A_{j,n}\log_2n = o(\sqrt{n}\lambda_n)$.
		Further, if Assumptions (A1)--(A4) hold with $\lambda_n > 2\sqrt{(c+1)\log p / n}$ 
		and $\min(|\boldsymbol{w}_1|) >c_3$ for some constant $c, c_3>0$,
		and $\lambda_n\sqrt{n/\log p} \rightarrow \infty$, then there exists some constant $c>0$ such that 
		\begin{eqnarray}
			P\left( \sup\limits_{\boldsymbol{\beta}\in \mathcal{N}}\left|\left| n^{-1}
			\boldsymbol{X}_2^T\boldsymbol{H}_{\gamma}^{(1)}(\boldsymbol{\beta}_1)\right|\right|_\infty
			\geq \lambda_n \min(|\boldsymbol{w}_1|)\right)\leq o(p^{-c}).
		\end{eqnarray}
	\end{lemma}
	\begin{proof}
		We follow the line of arguments used in the proof of Lemma 2 of Fan et al. \cite{Fan/etc:2014} and consider the decomposition
		$$
		\sup\limits_{\boldsymbol{\beta}\in \mathcal{N}}\left|\left|  n^{-1}\boldsymbol{X}_2^T\boldsymbol{H}_{\gamma}^{(1)}(\boldsymbol{\beta}_1)\right|\right|_\infty
		\leq I_1 + I_2 + I_3,
		$$
		where
		\begin{eqnarray}
			I_1 &=& \sup\limits_{\boldsymbol{\beta}\in \mathcal{N}}\left|\left|  \frac{1}{n}\boldsymbol{X}_2^T
			E\left(\boldsymbol{H}_{\gamma}^{(1)}(\boldsymbol{\beta}_{1}) - \boldsymbol{H}_{\gamma}^{(1)}(\boldsymbol{\beta}_{10})\right)
			\right|\right|_\infty,
			\nonumber\\
			I_2 &=& \frac{1}{n} \left|\left| \boldsymbol{X}_2^T\boldsymbol{H}_{\gamma}^{(1)}(\boldsymbol{\beta}_{10})\right|\right|_\infty,
			\nonumber\\
			I_3 &=& \max\limits_{j>s} \sup\limits_{\boldsymbol{\beta}\in \mathcal{N}}\left|  n^{-1}\sum_{i=1}^n h_{j,\boldsymbol{\beta}}(\boldsymbol{x}_i, y_i) \right|.
			\nonumber
		\end{eqnarray}

		We first consider $I_1$. Then, a Taylor series expansion of $\boldsymbol{H}_{\gamma}^{(1)}(\boldsymbol{\beta}_1)$ at any 
		$\boldsymbol{\beta}=(\boldsymbol{\beta}_1^T, \boldsymbol{0}^T)^T\in\mathcal{N}$ around $\boldsymbol{\beta}_{10}$, 
		along with Assumption (A3), yields 
		\begin{eqnarray}
			E\left[\boldsymbol{H}_{\gamma}^{(1)}(\boldsymbol{\beta}_1) - \boldsymbol{H}_{\gamma}^{(1)}(\boldsymbol{\beta}_{10})\right]
			&=& E[\boldsymbol{H}_{\gamma}^{(2)}(\boldsymbol{\beta}_{10})]\boldsymbol{X}_1(\boldsymbol{\beta}_1 - \boldsymbol{\beta}_{10})
			+ O(\delta_n^2), \nonumber
		\end{eqnarray}
		for sufficiently large $n$. Therefore, 
		\begin{eqnarray}
			I_1 &\leq& \sup\limits_{\boldsymbol{\beta}\in \mathcal{N}}\left|\left|  \frac{1}{n}\boldsymbol{X}_2^T
			E[\boldsymbol{H}_{\gamma}^{(2)}(\boldsymbol{\beta}_{10})]\boldsymbol{X}_1(\boldsymbol{\beta}_1 - \boldsymbol{\beta}_{10})\right|\right|_\infty
			+ O(\kappa_n\delta_n^2)
			\nonumber\\
			&\leq& \sup\limits_{\boldsymbol{\beta}\in \mathcal{N}}\left|\left|  \frac{1}{n}\boldsymbol{X}_2^T
			E[\boldsymbol{H}_{\gamma}^{(2)}(\boldsymbol{\beta}_{10})]\boldsymbol{X}_1\right|\right|_{2,\infty} ||\boldsymbol{\beta}_1 - \boldsymbol{\beta}_{10}||_2
			+ O(\kappa_n\delta_n^2)
			\nonumber\\
			&\leq& \frac{\lambda_n\min(|\boldsymbol{w}_1|)}{2C_1\delta_n} C_1\delta_n + O(\kappa_n\delta_n^2),
			\nonumber
		\end{eqnarray}
		where the last inequality follows from Assumption (A5) and the definition of $\mathcal{N}$. Hence, for sufficiently large $n$, we have
		\begin{eqnarray}
			I_1 &\leq& \frac{\lambda_n\min(|\boldsymbol{w}_1|)}{2} + o(\lambda_n).
			\label{EQ:Pf2.1}
		\end{eqnarray}
		
		Next, we consider $I_2$ and use the condition $\lambda_n > 2\sqrt{(c+1)\log p / n}$. An application of Hoeffding's inequality then implies that
		\begin{eqnarray}
			P\left(\left|\left| \boldsymbol{X}_2^T\boldsymbol{H}_{\gamma}^{(1)}(\boldsymbol{\beta}_{10})\right|\right|_\infty \geq n \lambda_n\right)
			&\leq &  \sum_{j=s+1}2\exp\left(-\frac{n^2\lambda_n^2}{R_{\gamma}^2 \sum_{i=1}^nx_{ij}^2}\right)
			\nonumber\\
			&=& 2 \exp\left(\log(p-s)-\frac{n\lambda_n^2}{R_{\gamma}^2}\right) \leq O(p^{-c}).
			\nonumber
		\end{eqnarray}
		Therefore, with probability at least $1-O(p^{-c})$, we have
		\begin{eqnarray}
			I_2 = n^{-1} \left|\left| \boldsymbol{X}_2^T\boldsymbol{H}_{\gamma}^{(1)}(\boldsymbol{\beta}_{10})\right|\right|_\infty \leq \lambda_n.
			\label{EQ:Pf2.2}
		\end{eqnarray}
		
		Finally, considering $I_3$, we recall that $|\psi_{1, \gamma}|$ is bounded and the columns of $\boldsymbol{X}$ 
		are standardized to have $\ell_1$ norm $\sqrt{n}$ so that their $\ell_2$ norm is bounded by $\sqrt{n}$.
		Thus, we get $\sum_{i=1}^n h_{j, \boldsymbol{\beta}}^2(\boldsymbol{x}_i, y_i) \leq n R_{\gamma}^2$ and hence
		$|| h_{j, \boldsymbol{\beta}}||_2 \leq \sqrt{n} R_{\gamma}$. 
		Additionally, in view of the assumption on the covering number of the spaces $\Gamma_j$,
		we can apply Corollary 14.4 of B\"{u}hlmann and van de Geer \cite{buhlmann} to get, for each $j>s$ and any $t>0$,  
		$$
		P\left(\sup\limits_{\boldsymbol{\beta}\in \mathcal{N}}\left|  n^{-1}\sum_{i=1}^n h_{j,\boldsymbol{\beta}}(\boldsymbol{x}_i, y_i) \right|
		\geq \frac{4R_{\gamma}}{\sqrt{n}}(3\sqrt{A_{j,n}}\log_2 n + 4+ 4t)\right)
		\leq 4 \exp\left(-nt^2/8\right).
		$$
		Taking $t=\sqrt{C(\log p )/n}$ with some large enough $C>0$, as $p\rightarrow\infty$, we deduce
		\begin{eqnarray}
			P\left(\max\limits_{j>s}\sup\limits_{\boldsymbol{\beta}\in \mathcal{N}}\left|  n^{-1}\sum_{i=1}^n h_{j,\boldsymbol{\beta}}(\boldsymbol{x}_i, y_i) \right|
			\geq \frac{12R_{\gamma}}{\sqrt{n}}\sqrt{A_{j,n}}\log_2 n\right) 
			&\leq& 4(p-s)\exp\left(-C(\log p)/8\right) 
			\nonumber\\
			&&~
			=4(p-s)p^{-C/8} \rightarrow 0. \nonumber
		\end{eqnarray}
		Since $\sqrt{A_{j,n}}\log_2n = o(\sqrt{n}\lambda_n)$, we finally get that, with probability at least $1- o(p^{-c})$, 
		\begin{eqnarray}
			I_3 = \max\limits_{j>s}\sup\limits_{\boldsymbol{\beta}\in \mathcal{N}}\left|  n^{-1}\sum_{i=1}^n h_{j,\boldsymbol{\beta}}(\boldsymbol{x}_i, y_i) \right|
			= o_P(\lambda_n).
			\label{EQ:Pf2.3}
		\end{eqnarray}
		Now the lemma follows by combining (\ref{EQ:Pf2.1}), (\ref{EQ:Pf2.2}) and (\ref{EQ:Pf2.3}).
	\end{proof}
	
	\begin{proof}[Proof of the theorem]
		The proof follows exactly in the same manner as in the proof of Theorem 2 of Fan et al. \cite{Fan/etc:2014}
		using Theorem 4.1 and Lemma \ref{LEM:L2}.
	\end{proof}
	
	\begin{remark}
		Our proof of Theorem 4.2 is based on Lemma \ref{LEM:L2} which, in turn, depends on an additional assumption 
		about the covering number of spaces $\Gamma_j$. This assumption is a bit restrictive; Fan et al. \cite{Fan/etc:2014}
		proved that this holds for quantile regression loss by  proper choices of the convergence rates of the sequences involved 
		($\delta_n, \lambda_n, \kappa_n$). Although we do not have a proof available right now, 
		we conjecture that it would be possible to prove Lemma \ref{LEM:L2} and hence Theorem 4.2
		for our DPD loss function avoiding this assumption appropriately. 
	\end{remark}
	
	\subsection{Proof of Theorem 4.3}
	\label{APP:proof.3}
	
	Consider the set-up and notation of Section 4 and Subsection \ref{APP:proof.2} above.
	Let us set $\eta_{i0}=\boldsymbol{x}_{1i}^T\boldsymbol{\beta}_{10}$ for $i=1, \ldots, n$ and  
	$\boldsymbol{\theta} = \boldsymbol{V}_n^{-1}(\boldsymbol{\beta}_1 - \boldsymbol{\beta}_{10})$ so that 
	$\boldsymbol{\beta}_1 = \boldsymbol{\beta}_{10} + \boldsymbol{V}_n\boldsymbol{\theta}$. 
	Also define 
	$$R_{n,i}(\boldsymbol{\theta}) = \rho_{\gamma}(\eta_{i0} + \boldsymbol{Z}_{n,i}^T\boldsymbol{\theta}, y_i) - \rho_{\gamma}(\eta_{i0}, y_i)
	+ (1+\gamma)\psi_{1,\gamma}(y_i-\eta_{i0}) \boldsymbol{Z}_{n,i}^T\boldsymbol{\theta},$$ 
	for each $i$ and
	$R_n(\boldsymbol{\theta}) = \sum_{i=1}^n R_{n,i}(\boldsymbol{\theta})$.
	Consider $\boldsymbol{\theta}$ over the convex open set 
	$B_0(n) = \left\{\boldsymbol{\theta}\in \mathbb{R}^s : ||\boldsymbol{\theta}||_2 <c^* \sqrt{s} \right\}$ 
	for some constant $c^*>0$ independent of $s$. 
	The following lemma gives the convergence rate of $r_n(\boldsymbol{\theta}) = R_n(\boldsymbol{\theta}) - E[R_n(\boldsymbol{\theta})]$,
	which will subsequently be used to prove the main theorem.

	\begin{lemma}\label{LEM:L3}
		Under the assumptions of Theorem 4.3, we have, for any $\epsilon>0$ and $\boldsymbol{\theta}\in B_0(n)$, 
		$$
		P\left(|r_n(\boldsymbol{\theta})|\geq \epsilon)\right) \leq \exp\left(-C\epsilon b_ns^2(\log s)\right), 
		$$
		where $b_n$ is some diverging sequence such that $b_ns^{7/2}(\log s) \max_i ||\boldsymbol{Z}_{n,i}||_2 \rightarrow 0$ and
		$C>0$ is some constant.
	\end{lemma}
	\begin{proof}
		Note that the $R_{n,i}(\boldsymbol{\theta})$'s are independent to each other. So, for any $\epsilon, t >0$, we get by Markov inequality that
		\begin{eqnarray}
			P\left(r_n(\boldsymbol{\theta}) \geq \epsilon \right)
			= P\left(e^{t r_n(\boldsymbol{\theta})} \geq e^{t\epsilon} \right)
			&\leq& e^{-t\epsilon}E\left[e^{t\sum_{i=1}^n \left(R_{n,i}(\boldsymbol{\theta}) - E[R_{n,i}(\boldsymbol{\theta})]\right)}\right]
			\nonumber\\
			&=& e^{-t\epsilon - t \sum_{i=1}^n E[R_{n,i}(\boldsymbol{\theta})]}\prod_{i=1}^{n} E\left[e^{t R_{n,i}(\boldsymbol{\theta})}\right].
			\label{EQ:C.1}
		\end{eqnarray}
		Let us first consider the term $E[R_{n,i}(\boldsymbol{\theta})]$ and use a Taylor series to obtain
		$$
		E[R_{n,i}(\boldsymbol{\theta})] = \frac{1}{2}(1+\gamma)E[J_{11,\gamma}(y_i - \eta_{i0})](\boldsymbol{Z}_{n,i}^T\boldsymbol{\theta})^2 
		+ O((\boldsymbol{Z}_{n,i}^T\boldsymbol{\theta})^3),
		$$
		for $i=1, \ldots, n$, with $O(\cdot)$ being uniform over $i$. Then, we get
		\begin{eqnarray}
			t\sum_{i=1}^n  E[R_{n,i}(\boldsymbol{\theta})] &=&
			\frac{t}{2}\boldsymbol{\theta}^T(\boldsymbol{Z}_n^TE[\boldsymbol{H}_{\gamma}^{(2)}(\boldsymbol{\beta}_0)]\boldsymbol{Z}_n)\boldsymbol{\theta} 
			+ O(t\sum_{i=1}^n (\boldsymbol{Z}_{n,i}^T\boldsymbol{\theta})^3)
			\nonumber\\
			&=& \frac{t}{2}||\boldsymbol{\theta}||_2^2 + O\left(t\sum_{i=1}^n (\boldsymbol{Z}_{n,i}^T\boldsymbol{\theta})^3\right),
			\label{EQ:C.2}
		\end{eqnarray}
		since  $\boldsymbol{Z}_n^TE[\boldsymbol{H}_{\gamma}^{(2)}(\boldsymbol{\beta}_0)]\boldsymbol{Z}_n$ is the identity matrix by definition of $\boldsymbol{Z}_n$.
		\\
		Next consider the term  $E\left[e^{t R_{n,i}(\boldsymbol{\theta})}\right]$ and use Taylor series expansions along with Assumption (A4) to obtain
		$$
		E\left[e^{t R_{n,i}(\boldsymbol{\theta})}\right]= 1 + \frac{t}{2}(1+\gamma)E[J_{11,\gamma}(y_i - \eta_{i0})](\boldsymbol{Z}_{n,i}^T\boldsymbol{\theta})^2 
		+ O(t^2(\boldsymbol{Z}_{n,i}^T\boldsymbol{\theta})^3).
		$$
		Next, using the fact that $\prod_{i=1}^n (1+x_i) \leq e^{\sum_{i=1}^n x_i}$ for large enough $x_i>0$, and the argument similar to (\ref{EQ:C.2}), 
		we get 
		\begin{eqnarray}
			\prod_{i=1}^n E\left[e^{t R_{n,i}(\boldsymbol{\theta})}\right] 
			&\leq& e^{\sum_{i=1}^n E[e^{t R_{n,i}(\boldsymbol{\theta})}  - 1]}
			\nonumber\\
			&\leq& \exp\left[\frac{t}{2}||\boldsymbol{\theta}||_2^2 + O\left(t^2\sum_{i=1}^n(\boldsymbol{Z}_{n,i}^T\boldsymbol{\theta})^3\right)\right].
			\label{EQ:C.3}
		\end{eqnarray}
		Combining (\ref{EQ:C.2}) and (\ref{EQ:C.3}) in (\ref{EQ:C.1}), we finally have
		\begin{eqnarray}
			P\left(r_n(\boldsymbol{\theta}) \geq \epsilon \right)
			&\leq& \exp\left[{-t\epsilon - O\left(t\sum_{i=1}^n(\boldsymbol{Z}_{n,i}^T\boldsymbol{\theta})^3\right) + O\left(t^2\sum_{i=1}^n(\boldsymbol{Z}_{n,i}^T\boldsymbol{\theta})^3\right)}\right].
			\label{EQ:C.4}
		\end{eqnarray}
		Now, taking $t=2s^2(log s)b_n$ with $b_n$ satisfying the conditions of the lemma, for $\boldsymbol{\theta}\in B_0(n)$, we get 
		\begin{eqnarray}
			t\sum_{i=1}^n(\boldsymbol{Z}_{n,i}^T\boldsymbol{\theta})^3 &\leq & 
			t \max_{i}|\boldsymbol{Z}_{n,i}^T\boldsymbol{\theta}| \left(\boldsymbol{Z}_n^TE[\boldsymbol{H}_{\gamma}^{(2)}(\boldsymbol{\beta}_0)]\boldsymbol{Z}_n\right)
			||\boldsymbol{\theta}||_2^2
			\nonumber\\
			&\leq& 2 s^2(\log s)b_n \max_i ||\boldsymbol{Z}_{n,i}||_2 ||\boldsymbol{\theta}||_2 ||\boldsymbol{\theta}||_2^2 
			\nonumber\\
			&\leq& 2(C^*)^3 s^{7/2} (\log s)b_n\max_i ||\boldsymbol{Z}_{n,i}||_2 \rightarrow 0.
			\nonumber
		\end{eqnarray}
		Here, in the last step we have used that $||\boldsymbol{\theta}||_2<C^\ast\sqrt{s}$ and the given conditions in the lemma.
		Then, it follows from (\ref{EQ:C.4}) that 
		\begin{eqnarray}
			P\left(r_n(\boldsymbol{\theta}) \geq \epsilon \right)
			&\leq& \exp\left[-\widetilde{C}t\epsilon \right]
			= \exp\left[-C\epsilon b_n s^2(\log s)\right],
			\nonumber
		\end{eqnarray}
		for appropriate constants $\widetilde{C}$ and $C$.
		The same can also be obtained for $P\left(r_n(\boldsymbol{\theta}) \leq -\epsilon \right)$ in a similar manner,
		which then completes the proof of the lemma.
	\end{proof}
	
	\noindent
	To complete the proof of Theorem 4.3, we further need the following lemma from Fan et al. \cite{Fan/etc:2014},
	which we state here for the sake of completeness. 
	
	\begin{lemma}[Lemma 4, Fan et al. \cite{Fan/etc:2014}, Supplementary material]\label{LEM:L4}
		Let $h(\boldsymbol{\theta})$ be a positive function defined on the convex open set $B_0(n)$,
		and $\left\{ h_n(\boldsymbol{\theta}): \boldsymbol{\theta}\in B_0(n) \right\}$ be a sequence of random convex functions.
		Suppose that there exists a diverging sequence $b_n$ such that, for every $\boldsymbol{\theta}\in B_0(n)$ and for all $\epsilon>0$, 
		$$
		P\left(|h_n(\boldsymbol{\theta}) - h(\boldsymbol{\theta})|\geq \epsilon \right)
		\leq c_4 \exp\left[ - c_5 \epsilon b_n s^2(\log s)\right],
		$$ 
		for some constants $c_4, c_5 >0$. Further, assume that there exists a constant $c_6>0$ such that 
		$
		|h(\boldsymbol{\theta}_1) - h(\boldsymbol{\theta}_2)| \leq c_6 s ||\boldsymbol{\theta}_1 - \boldsymbol{\theta}_2||_{\infty},
		$
		for any $\boldsymbol{\theta}_1, \boldsymbol{\theta}_2 \in B_0(n)$.
		Then, we have 
		$$
		\sup_{K_s}\left| h_n(\boldsymbol{\theta}) - h(\boldsymbol{\theta})\right| = o_P(1),
		$$
		where $K_s$ is any compact set in $\mathbb{R}^s$ defined as 
		$K_s =\left\{ \boldsymbol{\theta}\in \mathbb{R}^s : ||\boldsymbol{\theta}||_2 \leq c_7\sqrt{s} \right\} \subset B_0(n)$ 
		for some $c_7 \in (0, c^*)$.
	\end{lemma}
	
	\begin{proof}[Proof of the theorem]
		Let us now prove our main theorem, 
		extending the line of arguments in the proof of Theorem 3 of Fan et al. \cite{Fan/etc:2014}. 
		Define $Q_{n,\gamma}^\ast(\boldsymbol{\theta}) 
		= n \left[Q_{n,\gamma,\lambda_n}(\boldsymbol{\beta}_1, \boldsymbol{0}) - Q_{n,\gamma, \lambda_n}(\boldsymbol{\beta}_{01}, \boldsymbol{0})\right]$,
		which is then minimized at $\widehat{\boldsymbol{\theta}}_n 
		= \boldsymbol{V}_n^{-1}\left(\widehat{\boldsymbol{\beta}}_1^o - \boldsymbol{\beta}_{01}\right)$,
		because $Q_{n,\gamma,\lambda_n}(\boldsymbol{\beta}_1, \boldsymbol{0})$ is minimized at 
		$\boldsymbol{\beta}_1=\widehat{\boldsymbol{\beta}}_1^o$.
		Now, we study $Q_{n,\gamma}^\ast(\boldsymbol{\theta})$ by decomposing it into 
		its mean and centralized stochastic component as 
		$
		Q_{n,\gamma}^\ast(\boldsymbol{\theta}) = M_n(\boldsymbol{\theta}) + T_n(\boldsymbol{\theta}),\nonumber
		$
		where 
		\begin{equation}
			\begin{aligned}
				M_n(\boldsymbol{\theta}) &=E[Q_{n,\gamma}^\ast(\boldsymbol{\theta})] 
				= nE\left[L_{n,\gamma}(\boldsymbol{\beta}_1, \boldsymbol{0}) - L_{n,\gamma}(\boldsymbol{\beta}_{10}, \boldsymbol{0})\right] \\
				&~~~~~~~~~~~~~~~~~~~~~~~~~~+ n\lambda_n\sum\limits_{j=1}^{p} w_{j}\left(\left\vert \beta_{0j} + (\boldsymbol{V}_n\boldsymbol{\theta})_j\right\vert 
				- \left\vert \beta_{0j}\right\vert\right),
				\nonumber\\
				T_n (\boldsymbol{\theta}) &=  Q_{n,\gamma}^\ast(\boldsymbol{\theta}) - E[Q_{n,\gamma}^\ast(\boldsymbol{\theta})]
				= r_n(\boldsymbol{\theta}) - \boldsymbol{H}_{\gamma}^{(1)}(\boldsymbol{\theta}_0)\boldsymbol{Z}_n\boldsymbol{\theta}.
			\end{aligned}
		\end{equation}
		Here, we have used $E[\boldsymbol{H}_{\gamma}^{(1)}(\boldsymbol{\theta}_0)\boldsymbol{Z}_n\boldsymbol{\theta}]=\boldsymbol{0}$
		and $r_n(\boldsymbol{\theta})$ as in Lemma \ref{LEM:L3}.
		
		Let us first consider the first term in $M_n(\boldsymbol{\theta})$ for $\boldsymbol{\theta}\in B_0(n)$ 
		and proceed as in the proof of (\ref{EQ:A.1}) to get
		\begin{eqnarray}
			nE\left[L_{n,\gamma}(\boldsymbol{\beta}_1, \boldsymbol{0}) - L_{n,\gamma}(\boldsymbol{\beta}_{10}, \boldsymbol{0})\right]  
			= \frac{1}{2}  \boldsymbol{\theta}^T(\boldsymbol{Z}_n^TE[\boldsymbol{H}_{\gamma}^{(2)}(\boldsymbol{\beta}_0)]\boldsymbol{Z}_n)\boldsymbol{\theta}  +o(1)
			= \frac{1}{2}||\boldsymbol{\theta}||_2^2 + o(1).
			\nonumber
		\end{eqnarray}
		Next, considering the second term in $M_n(\boldsymbol{\theta})$ for $\boldsymbol{\theta}\in B_0(n)$, 
		we note that Assumption (A2), (A3) and (A6) implies
		$$
		||\boldsymbol{V}_n\boldsymbol{\theta}||_{\infty} \leq ||\boldsymbol{V}_n\boldsymbol{\theta}||_2
		\leq C n^{-1/2} ||\boldsymbol{\theta}||_2 = o\left(\min\limits_{1\leq j \leq s}|\beta_{j0}|\right), 
		$$
		and hence $sign(\boldsymbol{\beta}_{01}+ \boldsymbol{V}_n\boldsymbol{\theta}) = sign(\boldsymbol{\beta}_{01})$ and 
		$$
		\sum\limits_{j=1}^{p} w_{j}\left(\left\vert \beta_{0j} + (\boldsymbol{V}_n\boldsymbol{\theta})_j\right\vert 
		- \left\vert \beta_{0j}\right\vert\right) = \widetilde{\boldsymbol{w}_0}^T\boldsymbol{V}_n\boldsymbol{\theta}.
		$$
		Therefore, for any $\boldsymbol{\theta}\in B_0(n)$, we have
		\begin{eqnarray}
			M_n(\boldsymbol{\theta}) = \frac{1}{2}||\boldsymbol{\theta}||_2^2 + n \lambda_n\widetilde{\boldsymbol{w}_0}^T\boldsymbol{V}_n\boldsymbol{\theta} + o(1).
			\label{EQ:Ct.1}
		\end{eqnarray}
		This leads to the expression
		\begin{eqnarray}
			Q_{n,\gamma}^\ast(\boldsymbol{\theta}) 
			&=& \frac{1}{2}||\boldsymbol{\theta}||_2^2 + n \lambda_n\widetilde{\boldsymbol{w}_0}^T\boldsymbol{V}_n\boldsymbol{\theta} 
			- \boldsymbol{H}_{\gamma}^{(1)}(\boldsymbol{\theta}_0)\boldsymbol{Z}_n\boldsymbol{\theta}  + r_n(\boldsymbol{\theta}) + o(1)
			\nonumber\\
			&=& \frac{1}{2} ||\boldsymbol{\theta} - \boldsymbol{\eta}_n||_2^2 - \frac{1}{2}||\boldsymbol{\eta}_n||_2^2 + r_n(\boldsymbol{\theta}) + o(1),
			\label{EQ:Ct.2}
		\end{eqnarray}
		with $\boldsymbol{\eta}_n = 
		\left[\boldsymbol{H}_{\gamma}^{(1)}(\boldsymbol{\theta}_0)\boldsymbol{Z}_n - n \lambda_n\boldsymbol{V}_n\widetilde{\boldsymbol{w}_0}\right]$,
		which has an asymptotic normal distribution. 
		This is because, by the central limit theorem we have 
		$\boldsymbol{u}^T[\boldsymbol{Z}_n^T\boldsymbol{\Omega}_n\boldsymbol{Z}_n]^{-1/2} \boldsymbol{H}_{\gamma}^{(1)}(\boldsymbol{\theta}_0)\boldsymbol{Z}_n$
		asymptotically follows a standard normal distribution for any $\boldsymbol{u}\in \mathbb{R}^s$ with $\boldsymbol{u}^T\boldsymbol{u}=1$, 
		and hence 
		\begin{eqnarray}
			\boldsymbol{u}^T[\boldsymbol{Z}_n^T\boldsymbol{\Omega}_n\boldsymbol{Z}_n]^{-1/2} \left[ \boldsymbol{\eta}_n 
			+ n\lambda_n\boldsymbol{V}_n\widetilde{\boldsymbol{w}_0} \right]
			\mathop{\rightarrow}^{\mathcal{L}} N(0,1).
			\label{EQ:Ct.3}
		\end{eqnarray}
		This also ensures that $\boldsymbol{\eta}_n$ is bounded in $\ell_2$-norm as 
		\begin{eqnarray}
			||\boldsymbol{\eta}_n||_2 &\leq&  
			||\boldsymbol{H}_{\gamma}^{(1)}(\boldsymbol{\theta}_0)\boldsymbol{Z}_n||_2 + n \lambda_n||\boldsymbol{V}_n\widetilde{\boldsymbol{w}_0}||_2
			\nonumber\\
			&\leq& O_P(\sqrt{s}) + C\lambda_n\sqrt{n}||\widetilde{\boldsymbol{w}_0}||_2
			~~~~~~\mbox{[Using Assumptions (A2)-(A3)]}
			\nonumber\\
			&=& C\sqrt{s}\left[1 + O_P(1)\right],
			~~~~~~~~~~~~\mbox{[Using Assumption (A6)].}
			\label{EQ:Ct.3a}
		\end{eqnarray}
		
		Finally, we will now show that this quantity $\boldsymbol{\eta}_n$ is close to the minimizer $\widehat{\boldsymbol{\theta}}_n$ 
		of $Q_{n,\gamma}^\ast(\boldsymbol{\theta})$, with probability tending to one. 
		To this end, we need to investigate the quantity $r_n(\boldsymbol{\theta})$ 
		in the stochastic term $T_n(\boldsymbol{\theta})$ for $\boldsymbol{\theta}\in B_0(n)$ with a large enough $c^* \gg C$. 
		For this purpose, using Lemma \ref{LEM:L3}, 
		we get a sequence $b_n \rightarrow \infty$ such that, for any $\epsilon>0$, 
		$$
		P\left(|r_n(\boldsymbol{\theta})| \geq \epsilon \right) \leq \exp\left[-C\epsilon b_n s (\log s)\right].
		$$
		Then, to apply Lemma \ref{LEM:L4}, we note that $r_n(\boldsymbol{\theta})$ can be written as 
		$r_n(\boldsymbol{\theta}) = h_n(\boldsymbol{\theta}) - h(\boldsymbol{\theta}) +o(1)$ for 
		$$
		h_n(\boldsymbol{\theta}) = Q_{n,\gamma}^\ast(\boldsymbol{\theta}) - n \lambda_n\widetilde{\boldsymbol{w}_0}^T\boldsymbol{V}_n\boldsymbol{\theta} 
		+ \boldsymbol{H}_{\gamma}^{(1)}(\boldsymbol{\theta}_0)\boldsymbol{Z}_n\boldsymbol{\theta},
		~~~~\mbox{ and }~~
		h(\boldsymbol{\theta}) = ||\boldsymbol{\theta}||_2^2.
		$$
		By definition, these functions $h_n(\boldsymbol{\theta})$ and $h(\boldsymbol{\theta})$ are convex on $B_0(n)$ and,
		for any $\boldsymbol{\theta}_1, \boldsymbol{\theta}_2 \in B_0(n)$, we have (using Assumption (A3))
		\begin{eqnarray}
			\left|h(\boldsymbol{\theta}_1) - h(\boldsymbol{\theta}_2)\right| &=& 
			\left|(\boldsymbol{\theta}_1+ \boldsymbol{\theta}_2)^T(\boldsymbol{\theta}_1 -  \boldsymbol{\theta}_2)\right|
			\nonumber\\
			&\leq& ||\boldsymbol{\theta}_1+ \boldsymbol{\theta}_2||_2 ||\boldsymbol{\theta}_1 - \boldsymbol{\theta}_2||_2
			\leq C s ||\boldsymbol{\theta}_1 -  \boldsymbol{\theta}_2||_{\infty}.\nonumber
		\end{eqnarray}
		Thus, all the conditions of Lemma \ref{LEM:L4} are satisfied and we have 
		$$
		\sup_{K_s}\left| r_n(\boldsymbol{\theta})\right| = o_P(1),
		$$
		for any compact set  $K_s =\left\{ \boldsymbol{\theta}\in \mathbb{R}^s : ||\boldsymbol{\theta}||_2 \leq c_7\sqrt{s} \right\} \subset B_0(n)$ 
		with $c_7 \in (0, c^*)$.
		We choose $c_7$ large enough such that, for each $s$, the corresponding set $K_s$ cover the ball 
		$B_1(n)$ centered at $\boldsymbol{\eta}_n$ and radius $\epsilon$, an arbitrary fixed positive constant, with probability tending to one.
		Then, we get 
		\begin{eqnarray}
			\Delta_n = \sup\limits_{\boldsymbol{\theta}\in B_1(n)} |r_n(\boldsymbol{\theta})| 
			\leq \sup_{K_s}\left| r_n(\boldsymbol{\theta})\right| = o_P(1).
			\label{EQ:Ct.4}
		\end{eqnarray}
		Next, to study the behavior of $Q_{n,\gamma}^\ast(\boldsymbol{\theta})$ outside $B_1(n)$, 
		we write a vector outside $B_1(n)$ as $\boldsymbol{\theta} = \boldsymbol{\eta}_n + a \boldsymbol{e} \in \mathbb{R}^s$ 
		for $\boldsymbol{e}\in \mathbb{R}^s$ is a unit vector and $a$ is a constant satisfying $a>\epsilon$, the radius of $B_1(n)$.
		Let $\boldsymbol{\theta}^\ast$ be the boundary point of $B_1(n)$ that lies on the line segment joining $\boldsymbol{\eta}_n$ and $\boldsymbol{\theta}$
		so that $\boldsymbol{\theta}^\ast = \boldsymbol{\eta}_n + \epsilon \boldsymbol{e} = (1 - \epsilon/a) \boldsymbol{\eta}_n + (\epsilon/a)\boldsymbol{\theta}$.
		Now, by convexity of $Q_{n,\gamma}^\ast(\boldsymbol{\theta})$, along with (\ref{EQ:Ct.2}), we get \\
		\begin{eqnarray}
			\frac{\epsilon}{a} Q_{n,\gamma}^\ast(\boldsymbol{\theta}) + \left(1 - \frac{\epsilon}{a}\right) Q_{n,\gamma}^\ast(\boldsymbol{\eta}_n)
			\geq Q_{n,\gamma}^\ast(\boldsymbol{\theta}^\ast) 
			&\geq& \frac{1}{2}\epsilon^2 - \frac{1}{2}||\boldsymbol{\eta}_n||_2^2 - \Delta_n 
			\nonumber\\
			&\geq & \frac{1}{2}\epsilon^2 + Q_{n,\gamma}^\ast(\boldsymbol{\eta}_n) - 2\Delta_n.
			\nonumber
		\end{eqnarray}
		Since $\epsilon<a$, using  (\ref{EQ:Ct.4}) we get, for large enough $n$, 
		$$
		\inf\limits_{||\boldsymbol{\theta} - \boldsymbol{\eta}_n||_2>\epsilon} Q_{n,\gamma}^\ast(\boldsymbol{\theta})
		\geq Q_{n,\gamma}^\ast(\boldsymbol{\eta}_n) + \frac{a}{\epsilon}\left[\frac{\epsilon^2}{2} - o_P(1)\right] 
		> Q_{n,\gamma}^\ast(\boldsymbol{\eta}_n).
		$$
		Therefore, the minimizer $\widehat{\boldsymbol{\theta}}_n$ of $Q_{n,\gamma}^\ast(\boldsymbol{\theta})$ must
		lie within the ball $B_1(n)$ with probability tending to one, i.e., for any arbitrarily chosen $\epsilon>0$,
		$$
		P\left(||\widehat{\boldsymbol{\theta}}_n - \boldsymbol{\eta}_n||_2 > \epsilon\right) \rightarrow 0.
		$$
		This result, combined with (\ref{EQ:Ct.3}) via Slutsky's theorem, yields
		\begin{eqnarray}
			\boldsymbol{u}^T[\boldsymbol{Z}_n^T\boldsymbol{\Omega}_n\boldsymbol{Z}_n]^{-1/2} \left[ \widehat{\boldsymbol{\theta}}_n 
			+ n\lambda_n\boldsymbol{V}_n\widetilde{\boldsymbol{w}_0} \right]
			\mathop{\rightarrow}^{\mathcal{L}} N(0,1),
			\nonumber
		\end{eqnarray}
		from which the theorem follows by the fact that $\widehat{\theta}_n = \boldsymbol{V}_n^{-1}\left(\widehat{\beta}_1^o - \beta_{01}\right)$.
	\end{proof}

	\subsection{Proof of Theorem 4.5}
	\label{APP:proof.4}
	
	Following the idea of the proof of Theorems 4.1 and 4.2,
	let us first consider the minimization of $\widehat{Q}_{n,\gamma, \lambda}(\boldsymbol{\beta})$ 
	over the oracle subspace $\left\{ \boldsymbol{\beta} = (\boldsymbol{\beta}_1^T, \boldsymbol{\beta}_2^T)^T \in \mathbb{R}^p
	:  \boldsymbol{\beta}_2 = \boldsymbol{0}_{p-s} \right\}$.
	
	Let us start with a $=\boldsymbol{\beta} = (\boldsymbol{\beta}_1^T, \boldsymbol{0}_{p-s})^T$
	with $\boldsymbol{\beta}_1 = \boldsymbol{\beta}_{01} + \delta_n^* \boldsymbol{v}_1 \in \mathbb{R}^s$
	where $||\boldsymbol{v}_1||_2 = C$ for some large enough constant $C>0$.
	For this particular $\boldsymbol{\beta}$, we have 
	\begin{eqnarray}
		n\left[\widehat{Q}_{n,\gamma, \lambda}(\boldsymbol{\beta}) - \widehat{Q}_{n,\gamma, \lambda}(\boldsymbol{\beta}_0)\right] 
		= I_1(\boldsymbol{v}_1) + I_2(\boldsymbol{v}_1),
		\label{EQ:ThP4.1}
	\end{eqnarray}
	where 
	\begin{eqnarray}
		I_1(\boldsymbol{v}_1) &=& n \left[L_{n,\gamma}(\boldsymbol{\beta}_{10}+\delta_n^\ast\boldsymbol{v}_1 , \boldsymbol{0}) 
		- L_{n,\gamma}(\boldsymbol{\beta}_{10}, \boldsymbol{0})\right], 
		\nonumber\\
		I_2 (\boldsymbol{v}_1) &=&  n\lambda_n\sum\limits_{j=1}^{p} \widehat{w}_{j}\left(\left\vert \beta_{0j} + \delta_n^*(\boldsymbol{v}_1)_j\right\vert 
		- \left\vert \beta_{0j}\right\vert\right).
		\nonumber
	\end{eqnarray}
	But, as in the proof of Theorem 4.1 in Subsection \ref{APP:proof.1}, we have 
	\begin{eqnarray}
		E\left[I_1(\boldsymbol{v}_1)\right]
		&\geq& \frac{1}{2} c_0c_1n||\delta_n^* \boldsymbol{v}_1||_2^2, 
		\label{EQ:ThP4.2}
	\end{eqnarray}
	and also, with probability at least $1 - n ^{-cs}$,  we have
	\begin{eqnarray}
		\left|I_1(\boldsymbol{v}_1) - E\left[I_1(\boldsymbol{v}_1)\right]\right|
		\leq nZ_n(C\delta_n^*) \leq 2L\delta_n^*\sqrt{sn(\log n)} ||\boldsymbol{v}_1||_2, 
		\nonumber
	\end{eqnarray}
	since $\delta_n^* = o(k_n^{-1}s^{-1/2})$ by our assumptions.
	Then, by triangle inequality, we get the lower bound of $I_1$ as
	\begin{eqnarray}
		I_1(\boldsymbol{v}_1) \geq \frac{1}{2} c_0c_1n(\delta_n^*)^2 ||\boldsymbol{v}_1||_2^2
		- 2L\delta_n^*\sqrt{sn(\log n)} ||\boldsymbol{v}_1||_2. 
		\label{EQ:ThP4.3}
	\end{eqnarray}
	Next, we can bound the term $I_2$ using the triangle and Cauchy-Swartz inequalities as  
	\begin{eqnarray}
		|I_2(\boldsymbol{v}_1)| \leq n\lambda_n\sum\limits_{j=1}^{p} \left\vert\widehat{w}_{j} \delta_n^*(\boldsymbol{v}_1)_j\right\vert 
		\leq n \lambda_n\delta_n^*||\widehat{\boldsymbol{w}}_0||_2||\boldsymbol{v}_1||_2. 
		\label{EQ:ThP4.4}
	\end{eqnarray}
	However, by Assumptions (A7)-(A8), we have
	\begin{eqnarray}
		||\widehat{\boldsymbol{w}}_0||_2 &\leq& ||\widehat{\boldsymbol{w}}_0 - \boldsymbol{w}_0^*||_2 + ||\boldsymbol{w}_0^\ast||_2
		\nonumber\\
		&\leq& c_5 ||\widetilde{\boldsymbol{\beta}}_1 - \boldsymbol{\beta}_1||_2 + ||\boldsymbol{w}_0^*||_2 
		\leq C_2c_5\sqrt{s (\log p)/n} + ||\boldsymbol{w}_0^*||_2.
	\end{eqnarray}
	Combining all the above results, we finally get 
	\begin{eqnarray}
		n\left[\widehat{Q}_{n,\gamma, \lambda}(\boldsymbol{\beta}) - \widehat{Q}_{n,\gamma, \lambda}(\boldsymbol{\beta}_0)\right] 
		&\geq& \frac{1}{2} c_0c_1n(\delta_n^*)^2 ||\boldsymbol{v}_1||_2^2
		- 2L\delta_n^*\sqrt{sn(\log n)} ||\boldsymbol{v}_1||_2 
		\nonumber\\
		&&~~ -  n \lambda_n\delta_n^*\left(C_2c_5\sqrt{s (\log p)/n} + + ||\boldsymbol{w}_0^*||_2\right)||\boldsymbol{v}_1||_2.
		\nonumber
	\end{eqnarray}
	Taking $||\boldsymbol{v}_1||_2 = C$ large enough, with probability tending to one, we can obtain
	$$
	n\left[\widehat{Q}_{n,\gamma, \lambda}(\boldsymbol{\beta}) - \widehat{Q}_{n,\gamma, \lambda}(\boldsymbol{\beta}_0)\right] >0,
	~~\mbox{ i.e., }~
	\widehat{Q}_{n,\gamma, \lambda}(\boldsymbol{\beta}_{10}+\delta_n^\ast\boldsymbol{v}_1, \boldsymbol{0}) 
	> \widehat{Q}_{n,\gamma, \lambda}(\boldsymbol{\beta}_{01}, \boldsymbol{0}).
	$$
	Therefore, with probability tending to one, there exists a minimizer $\widehat{\boldsymbol{\beta}}_1$ 
	of $\widehat{Q}_{n,\gamma, \lambda}(\boldsymbol{\beta}_{1}, \boldsymbol{0})$ such that
	$||\widehat{\boldsymbol{\beta}}_1 - \boldsymbol{\beta}_1||_2 \leq C_3\delta_n^*$ for some constant $C_3>0$.
	Then, to complete the proof of the theorem, it is enough to show that 
	$(\widehat{\boldsymbol{\beta}}_1^T, \boldsymbol{0})^T$ is indeed also a global minimizer of 
	$\widehat{Q}_{n,\gamma,\lambda}(\boldsymbol{\beta})$ over the whole $\mathbb{R}^p$.
	
	To this end, as in Theorem 4.2, it is enough to show via the KKT condition that, 
	with probability tending to one,  
	\begin{eqnarray}
		\left|\left| n^{-1} \boldsymbol{X}_2^T\boldsymbol{H}_{\gamma}^{(1)}(\widehat{\boldsymbol{\beta}}_1)\right|\right|_\infty
		< \lambda_n \min(|\widehat{\boldsymbol{w}}_1|).
		\label{EQ:ThP4.5}
	\end{eqnarray}
	To prove (\ref{EQ:ThP4.5}), we note that $\beta_{0j} = 0$ for all $j=s+1, \ldots, p$, 
	and hence ${w}_j^\ast =w(0+)$. But, by Assumption (A7), we have $|\widetilde{\beta}_{j}|\leq C_2 \sqrt{s(\log p)/n}$.
	Therefore, using Assumption (A8), we get
	$$
	\min(|\widehat{\boldsymbol{w}}_1|) =\min_{j>s} w(|\widetilde{\beta}_{j}|)
	\geq w(C_2 \sqrt{s(\log p)/n})  > \frac{1}{2} w(0+) = \frac{1}{2} \min(|{\boldsymbol{w}}_1^\ast|).
	$$  
	Next, we apply Lemma \ref{LEM:L2} under the assumptions of Theorem 4.2 with $\delta_n=\delta_n^*$ 
	to get, with probability at least $1- o(p^{-c})$,
	\begin{eqnarray}
		\sup_{||\boldsymbol{\beta}_1 - \boldsymbol{\beta}_{01}||_2 \leq C_3\delta_n^*}
		\left|\left| n^{-1} \boldsymbol{X}_2^T\boldsymbol{H}_{\gamma}^{(1)}({\boldsymbol{\beta}}_1)\right|\right|_\infty
		< \lambda_n \min(|{\boldsymbol{w}}_1^\ast|)
		< \lambda_n \min(|\widehat{\boldsymbol{w}}_1|).\nonumber
	\end{eqnarray}
	But, since $||\widehat{\boldsymbol{\beta}}_1 - \boldsymbol{\beta}_{01}||_2 \leq C_3\delta_n^*$
	with probability tending to one, the above results also implies (\ref{EQ:ThP4.5})
	completing the proof of the theorem.
	
	\subsection{Proof of Theorem 4.6}
	\label{APP:proof.5}
	
	First, we note that, by Theorem 4.5, with probability tending to one, 
	there exists a global minimizer $(\widehat{\boldsymbol{\beta}}_1^T, \boldsymbol{0})^T$  of 
	$\widehat{Q}_{n,\gamma,\lambda}(\boldsymbol{\beta})$ such that 
	$||\widehat{\boldsymbol{\beta}}_1 - \boldsymbol{\beta}_{01}||_2 \leq C_3\delta_n^*$.
	
	Then, to show the asymptotic normality of $\widehat{\boldsymbol{\beta}}_1$,
	we reconsider the notation used in Theorem 4.3 and its proof and define 
	$$
	\widehat{Q}_{n,\gamma}^\ast(\boldsymbol{\theta}) 
	= n \left[\widehat{Q}_{n,\gamma,\lambda_n}(\boldsymbol{V}_n\boldsymbol{\theta} + \boldsymbol{\beta}_{10}, \boldsymbol{0}) 
	- \widehat{Q}_{n,\gamma, \lambda_n}(\boldsymbol{\beta}_{01}, \boldsymbol{0})\right].
	$$
	Then, clearly a global minimizer of $\widehat{Q}_{n,\gamma}^\ast(\boldsymbol{\theta}) $ 
	is $\widehat{\boldsymbol{\theta}}_n = \boldsymbol{V}_n^{-1}\left(\widehat{\boldsymbol{\beta}}_1 - \boldsymbol{\beta}_{01}\right)$.
	
	Then, under Assumption (A9), one can follow the arguments of the proof of Theorem 5 of Fan et al. \cite{Fan/etc:2014}
	to show that 
	$\sup\limits_{\boldsymbol{\theta}\in B_0(n)} 
	\left|\widehat{Q}_{n,\gamma}^\ast(\boldsymbol{\theta}) - {Q}_{n,\gamma}^\ast(\boldsymbol{\theta}) \right| = o_P(1),$
	and then $\widehat{\boldsymbol{\theta}} - \boldsymbol{\eta}_n = o_P(1)$, 
	where $\boldsymbol{\eta}_n$ is as defined in the proof of Theorem 4.3 in Subsection \ref{APP:proof.3}.
	Then, the present theorem follows from the asymptotic properties of $\boldsymbol{\eta}_n$ 
	as studied in this supplementary material.
	
	\section{Computational Algorithm}
	We develop in this section the computational algorithm used to fit the general AW-DPD-LASSO estimator. We propose an iterative optimization algorithm in two steps. For any given solution $(\boldsymbol{\beta}^{(m)},\sigma^{(m)})$, we first update the  value of $\boldsymbol{\beta}$ and secondly we update the value of $\sigma.$ The first step is computed using MM-algorithm and the second one is carried out by by approximating the solution of the estimating equations of $\sigma$.
	\subsection{MM-algorithm for miization of the DPD loss in $\bold\beta$} \label{subsection5.1}
	MM-algorithm (Hunter and Lange \cite{Hunter}) is one of the most common strategies for  constructing optimization algorithms. 
	It operates by creating a surrogate function that minorizes (or majorizes) the objective function. 
	When the surrogate function is optimized, the objective function is driven uphill or downhill as needed. 
	Let be $h(\nu)$ a real-valued  objective function.  The function $h_{MM}(\nu|\nu^{(m)})$  is said to majorize $h(\nu)$ 
	at the point $\nu^{(m)}$ if 
	\begin{equation} \label{MMeq}
		h_{MM}(\nu^{(m)}|\nu^{(m)}) = h(\nu^{(m)}), 
		~~~~\mbox{ and }~~~~ h_{MM}(\nu|\nu^{(m)}) \geq h(\nu).
	\end{equation}
	%
	If $\nu^{(m)}$ is the current solution at step $m$, for $m=0,1...$, then MM-algorithm proposes to iteratively update the solution as 
	$$ 
	\nu^{(m+1)} = \operatorname{arg}\operatorname{min}_{\nu} h_{MM}(\nu | \nu^{(m)}). 
	$$
	From conditions (\ref{MMeq}), the descending property can be drawn
	\begin{equation} \label{descentprop}
		h (\nu^{(m+1)}) \leq h_{MM}(\nu^{(m+1)} | \nu^{(m)}) \leq  h_{MM}(\nu^{(m)} | \nu^{(m)}) = h(\nu^{(m)}),
	\end{equation}  
	and hence, it can be shown that MM-algorithm converges to the required minimizer of $h(\nu)$ (Proposition 11, Castilla et al. \cite{castilla}).
	The descent property (\ref{descentprop}) lends an MM-algorithm remarkable numerical stability. 
	Furthermore, it is not necessary to obtain a strict local minimum of the majorization function to ensure the descent property (\ref{descentprop}), 
	but it suffices to downward $ h_{MM}(\nu | \nu^{(m)}) $ and so any optimization algorithm could be used to minimize it. 
	To apply this technique to our DPD loss given in Equation (7) of the main paper, as a function of $\bold\beta$ with a fixed $\sigma$, 
	we consider an alternative equivalent objective function given by
	\begin{equation}\label{dpdloss2}
		\tilde{L}_{n,\gamma,\lambda}(\boldsymbol{\beta}) = -\log\left[\frac{1}{n}\sum_{i=1}^{n}\exp\left(
		-\frac{\gamma}{2}\left(  \frac{y_{i}-\boldsymbol{x}_{i}^{T}\boldsymbol{\beta}%
		}{\sigma}\right)  ^{2}\right) \right].
	\end{equation}
	Note that, minimizing $\tilde{L}_{n,\gamma,\lambda}(\boldsymbol{\beta})$ with respect to $\bold\beta$ is equivalent 
	to minimizing the DPD loss $L_{n,\gamma,\lambda}(\boldsymbol{\beta},\sigma)$ given in Equation (7) of the main paper as a function of $\bold\beta$ with fixed $\sigma$. 
	Then, we can equivalently apply the MM-algorithm to $\tilde{L}_{n,\gamma,\lambda}(\boldsymbol{\beta})$ in (\ref{dpdloss2}),
	which can be further bounded by a quadratic function using the following Jensen inequality.
	$$
	\kappa(\boldsymbol{z}^T \boldsymbol{\nu}) = \kappa\left( \frac{z_i\nu_i^{(m)}}{\boldsymbol{z}^T \boldsymbol{\nu}^{(m)}} \nu_i \frac{\boldsymbol{z}^T \boldsymbol{\nu}^{(m)}}{\nu_i^{(m)}}\right) \leq \sum_{i} \frac{z_i \nu_i^{(m)}}{\boldsymbol{z}^T \boldsymbol{\nu}^{(m)}} \kappa\left(\nu_i \frac{\boldsymbol{z}^T \boldsymbol{\nu}^{(m)}}{\nu_i^{(m)}}\right),
	$$
	where $\kappa(\nu)$ is a convex function and  $\boldsymbol{z}$, $\boldsymbol{\nu}$ and $\boldsymbol{\nu^{(m)}}$ are positive vectors.
	For our purpose, let us take $\kappa(u)= - \log (u)$, $\boldsymbol{z} = \left(\frac{1}{n},..,\frac{1}{n}\right)$,
	\begin{align*}
		\nu_i & = \exp\left(\frac{-\gamma}{2} \left(\frac{y_i-\boldsymbol{x}_i^T\boldsymbol{\beta}}{\sigma^{(m)}}\right)^2\right),
		~~~~~\mbox{ and }~~
		\nu^{(m)}_i = \exp\left(\frac{-\gamma}{2} \left(\frac{y_i-\boldsymbol{x}_i^T\widehat{\boldsymbol{\beta}}^{(m)}}{\widehat{\sigma}^{(m)}}\right)^2\right),
	\end{align*}
	so that the above-mentioned Jensen inequality yields
	\begin{equation} \label{MMalg}
		\begin{aligned}
			\tilde{L}_{n,\gamma,\lambda}(\boldsymbol{\beta}) 
			& = -\log\left[\frac{1}{n}\sum_{i=1}^{n}\exp\left(
			-\frac{\gamma}{2}\left(  \frac{y_{i}-\boldsymbol{x}_{i}^{T}\boldsymbol{\beta}%
			}{\sigma^{(m)}}\right)  ^{2}\right) \right] \\
			&\leq  \sum_{i=1}^n \mu_i^{(m)}  \frac{1}{2} \left( \frac{y_i-\boldsymbol{x}_i^T \boldsymbol{\beta}}{\sigma^{(m)}}\right)^2 + \sum_{i=1}^n \mu_i^{(m)}  \log \left(\frac{1}{\mu_i^{(m)}}\right), 
		\end{aligned}
	\end{equation}
	with $$\mu_i^{(m)} = \exp\left(-\frac{\gamma}{2} \left(\frac{y_i-\boldsymbol{x}_i^T\widehat{\boldsymbol{\beta}}^{(m)}}{\widehat{\sigma}^{(m)}}\right)^2\right) \left[\sum_{l=1}^n \exp\left(-\frac{\gamma}{2} \left(\frac{y_l-\boldsymbol{x}_l^T\widehat{\boldsymbol{\beta}}^{(m)}}{\widehat{\sigma}^{(m)}}\right)^2\right)\right]^{-1}.$$
	Note that only the first term in (\ref{MMalg}) depends on $\boldsymbol{\beta}$, and so it suffices to minimize the quadratic loss
	\begin{equation}\label{surrogate1}
		U_{MM}(\boldsymbol{\beta}| \boldsymbol{\beta}^{(m)}, \sigma^{(m)}) =\sum_{i=1}^n \mu_i^{(m)}  \left( \frac{y_i-\boldsymbol{x}_i^T \boldsymbol{\beta}}{\sigma^{(m)}}\right)^2.
	\end{equation}
	with $$\mu_i^{(m)} = \exp\left(-\frac{\gamma}{2} \left(\frac{y_i-\boldsymbol{x}_i^T\widehat{\boldsymbol{\beta}}^{(m)}}{\widehat{\sigma}^{(m)}}\right)^2\right) \left[\sum_{l=1}^n \exp\left(-\frac{\gamma}{2} \left(\frac{y_l-\boldsymbol{x}_l^T\widehat{\boldsymbol{\beta}}^{(m)}}{\widehat{\sigma}^{(m)}}\right)^2\right)\right]^{-1}.$$
	Furthermore, if we transfer the data $(y_i, \boldsymbol{x}_i)$ to their weighted versions   
	$
	y_i^* = \sqrt{\frac{\mu_i^{(m)}}{\widehat{\sigma}^{(m)}}} y_i, 
	$ and 
	$\boldsymbol{x}_i^* = \sqrt{\frac{\mu_i^{(m)}}{\widehat{\sigma}^{(m)}}} \boldsymbol{x}_i,
	$
	then the function $U_{MM}(\boldsymbol{\beta}| \boldsymbol{\beta}^{(m)}, \sigma^{(m)})$ becomes just  the least-square loss for $y_i^*$ and $\boldsymbol{x}_i^*$, i.e., 
	$
	U_{MM}(\boldsymbol{\beta}| \boldsymbol{\beta}^{(m)}, \sigma^{(m)}) =\sum_{i=1}^n \left({y_i^* - \boldsymbol{x}_i^{*T}\boldsymbol{\beta}}\right)^2. 
	$
	This final function, even after considering the additional penalty term,  can easily be minimized by existing algorithms.
	
	\subsection{Computation of General Adaptive  LASSO with Least Squares Loss}
	\label{subsection5.2}
	We now discuss an algorithm for general adaptive LASSO with the least squares loss. 
	So, we consider the objective function
	\begin{equation}\label{problemU}
		U_{n,\gamma,\lambda}(\boldsymbol{\beta}) = \sum_{i=1}^n\left(y_i - \boldsymbol{x}_i^T\boldsymbol{\beta}\right)^2 + \lambda \sum_{j=1}^{p} \omega(| \tilde{\beta}_j|)\cdot|\beta_j|,
	\end{equation}
	where $\omega(\cdot)$ is a general weight function and $\tilde{\boldsymbol{\beta}}$ is a robust initial estimator. Note that the function $U_{n,\gamma,\lambda}$  depends only on $\beta$ (there is no $\sigma$).
	Zou  \cite{zou} showed the minimization problem, 
	$\operatorname{min}_{\boldsymbol{\beta}} U_{n,\gamma,\lambda}(\boldsymbol{\beta})$, 
	can be solved  by the standard LASSO algorithm applied to a weighted design matrix, 
	which is explicitly exhibited in the following Algorithm S1.\\
	
	\noindent \textbf{Algorithm S1} (Computing general adaptive LASSO estimator)
	\begin{enumerate}
		\item Define $\widehat{\omega}_j = \omega(| \tilde{\beta}_j|) $ and $\boldsymbol{x}_j^{*} =  \frac{\boldsymbol{x}_j}{\widehat{\omega}_j}$, $j = 1,2,...,p$
		\item Solve the LASSO problem 
		\begin{equation}\label{lassoclassic}
			\widehat{\boldsymbol{\beta}}^{*} = \operatorname{arg min}_{\boldsymbol{\beta}} \sum_{i=1}^n \left(y_i- \boldsymbol{x}_i^{*T}\boldsymbol{\beta} \right)^2 + \lambda \sum_{j=1}^p |\beta_j|.
		\end{equation}
		\item Output $\widehat{\boldsymbol{\beta}} = \widehat{\boldsymbol{\beta}}^{*}/\widehat{\omega}_j$.
	\end{enumerate}
	
	The problem in (\ref{lassoclassic}) is a standard one and can be solve by using the existing methods such as the coordinate descent approach 
	(Friedman et al.  \cite{Friedman}).
	Also, the selection of optimal $\lambda$ can be done through the Least Angle Regression (LARS) algorithm (Efron et al. \cite{efron}), 
	or using cross validation or any information criterion.
	
	\subsection{Computing the proposed AW-DPD-LASSO Estimator}
	
	We now explain the final computational algorithm for our robust  generalized AW-DPD-LASSO estimator  by minimizing  (8),
	Since a direct minimization of $Q_{n,\gamma,\lambda}(\boldsymbol{\beta}, \sigma)$ is computationally inefficient, an alternative optimization technique is proposed, combining the  MM-algorithm and 
	the computational algorithm for  general adaptive LASSO (as described in Algorithm S1 in Section \ref{subsection5.2}).
	
	Our algorithm for the computation of AW-DPD-LASSO estimator proceeds iteratively as follows. 
	Suppose that  a current solution at step $m$ is given by $(\widehat{\boldsymbol{\beta}}^{(m)}, \widehat{\sigma}^{(m)})$. 
	First, we fix $\widehat{\sigma}^{(m)}$ and we consider the minimization problem with respect to $\bold\beta$ only, following the discussions in Section \ref{subsection5.1},  as
	\begin{equation}
		\operatorname{min}_{\boldsymbol{\beta}} \tilde{Q}_{n,\gamma,\lambda}\left(\boldsymbol{\beta}| \boldsymbol{\beta}^{(m)}, \sigma^{(m)}\right) = \operatorname{min}_{\boldsymbol{\beta}} \left( \tilde{L}_{n,\gamma,\lambda}(\widehat{\sigma}^{(m)},\boldsymbol{\beta}) + \lambda \sum_{j=1}^{p} \omega(| \tilde{\beta}_j|)\cdot|\beta_j| \right).
	\end{equation}
	Using Jensen's inequality as in Section \ref{subsection5.1},
	$\tilde{Q}_{n,\gamma,\lambda}\big(\boldsymbol{\beta}| \boldsymbol{\beta}^{(m)}, \sigma^{(m)}\big)$ can be majorized by 
	$$ 
	U_{MM}\left(\boldsymbol{\beta}| \boldsymbol{\beta}^{(m)}, \sigma^{(m)}\right) + \lambda \sum_{j=1}^{p} \omega(| \tilde{\beta}_j|)\cdot|\beta_j|  
	=\sum_{i=1}^n \left({y_i^* - \boldsymbol{x}_i^{*T}\boldsymbol{\beta}}\right)^2 + \lambda \sum_{j=1}^{p} \omega(| \tilde{\beta}_j|)\cdot|\beta_j|, 
	$$
	with $y_i^*$ and $\boldsymbol{x}_i^*$ being as defined in Section \ref{subsection5.1}. 
	This final majorization function can easily be minimized using Algorithm S1 of Section 2 of the Appendix. 
	Thus, MM-algorithm can be applied to optimize $Q_{n,\gamma,\lambda}\left(\boldsymbol{\beta}| \boldsymbol{\beta}^{(m)}, \sigma^{(m)}\right)$
	iteratively until convergence.
	Once this MM algorithm converges, we update the solution 	for $\bold\beta$ as
	\begin{equation} \label{betaupdate}
		\widehat{\boldsymbol{\beta}}^{(m+1)} = \operatorname{arg min} 
		Q_{n,\gamma,\lambda}\left(\boldsymbol{\beta}| \boldsymbol{\beta}^{(m)}, \sigma^{(m)}\right).
	\end{equation} 
	
	It is important to note that the robust initial estimate $\tilde{\boldsymbol{\beta}}$ can also be updated at each iteration 
	by considering $\tilde{\boldsymbol{\beta}}=\widehat{\boldsymbol{\beta}}^{(m)}$, when the solution for $\boldsymbol{\beta}$ is updated as
	$\widehat{\boldsymbol{\beta}}^{(m+1)}$.

	Next, we consider the objective function corresponding to the AW-DPD-LASSO as a function of $\sigma$ only 
	(with $\bold\beta$ being fixed at the solution $\beta^{(m+1)}$) as given by  
	\begin{equation}
		\overline{Q_{n,\gamma,\lambda}}(\sigma) = L_{n,\gamma,\lambda}(\sigma,\widehat{\boldsymbol{\beta}}^{(m+1)}) + \lambda \sum_{j=1}^{p} \omega(| \tilde{\beta}_j|)\cdot|\widehat{\beta}^{(m+1)}_j| .
	\end{equation}
	The minimizer  $\widehat{\sigma}$ of $Q_{n,\gamma,\lambda}(\sigma)$ should make its first derivative zero. But,
	\begin{align*}
		\frac{\partial \overline{Q_{n,\gamma,\lambda}}(\sigma)}{\partial \sigma} 
		=  \frac{\sigma^{-\gamma-1}}{(2\pi)^{\gamma/2}} \bigg[&\frac{-\gamma}{\gamma+1} + \frac{\gamma+1}{n} \sum_{i=1}^n \exp\left(-\frac{\gamma}{2}u_i^{(m+1)2} \right)\\
		& - \frac{\gamma+1}{n} \sum_{i=1}^n \exp\left(-\frac{\gamma}{2} u_i^{(m+1)2} \right) u_i^{(m+1)2}  \bigg].
	\end{align*}
	where $ u_i^{(m+1)} = \left(\frac{y_i-\boldsymbol{x}_i^T\widehat{\boldsymbol{\beta}}^{(m+1)}}{\sigma}\right).$
	Thus, $\hat\sigma$ should satisfy the estimating equation 
	\begin{equation} 
		\frac{1}{n} \sum_{i=1}^n \exp\left(-\frac{\gamma}{2}u_i^{(m+1)2} \right) 
		-\frac{1}{n} \sum_{i=1}^n \exp\left(-\frac{\gamma}{2}u_i^{(m+1)2}  \right) 
		u_i^{(m+1)2}  =  \frac{\gamma}{(\gamma+1)^{3/2}}.
		\label{5.41}
	\end{equation}
	This equation (\ref{5.41}) has no explicit solution. 
	However, as in Ghosh and Majumdar \cite{ghosh2}, the solution of (\ref{5.41}) can be approximated  as
	\begin{equation}\label{sigmaupdate}
		(\widehat{\sigma}^{(m+1)})^2 = \left[\frac{1}{n} \sum_{i=1}^n w_i^{(m)} - \frac{\gamma}{(\gamma+1)^{3/2}} \right]^{-1}  \frac{1}{n} \sum_{i=1}^n w_i^{(m)} \left(y_i-\boldsymbol{x}_i^T\widehat{\boldsymbol{\beta}}^{(m+1)}\right)^2, 
	\end{equation}
	where 
	$ w_i^{(m)} 
	=  \exp\left(-\frac{\gamma}{2}\left(\frac{y_i-\boldsymbol{x}_i^T\widehat{\boldsymbol{\beta}}^{(m+1)}}{\widehat{\sigma}^{(m)}}\right)^2 \right).
	$
	We alternatively repeat the above procedures of separately updating the parameters $\bold\beta$ and $\sigma$
	until both the parameter values converges simultaneously in the same iteration, 
	and the final solution (after convergence) is then nothing but their AW-DPD-LASSO estimator. 
	
	\section{Additional tables: Simulations and data application}
	\label{APPB:Additional_results}

	We report our simulation results for Setting A with $p=1000$ and $p=500$, and Setting B with $p=500$ in Tables \ref{p1000e0signal1}-\ref{p500ex1signal3}, respectively. Table \ref{TAB:data1} contains the error measures for X-ray microanalysis (EPXMA) with the different proposed methods and Table \ref{tauscale} contains the median of $\tau$-scale of its residuals.

	\begin{table}[h]
		\centering
		\caption{Performance measures obtained by different methods for $p=1000$, Setting A and no outliers}
		\resizebox{\textwidth}{!}{
\begin{tabular}{l|rrr|rrr|r}
				\hline
				Method	& MS($\widehat{\boldsymbol{\beta}}$) & TP($\widehat{\boldsymbol{\beta}}$) & TN($\widehat{\boldsymbol{\beta}}$) & MSES($\widehat{\boldsymbol{\beta}}$) &  MSEN($\widehat{\boldsymbol{\beta}}$)  &  EE($\widehat{\sigma}$) & APrB($\widehat{\boldsymbol{\beta}}$)  \\ 
				& & & &  $(10^{-2})$ & $(10^{-5})$ &  $(10^{-2})$ &  $(10^{-2})$ \\
				\hline
				LS-LASSO & 4.16 & 1.00 & 1.00 & 2.70 & 0.08 & 33.66 & 5.23 \\
				Ad-LS-LASSO & 3.00 & 1.00 & 1.00 & 1.68 & 0.00 & 28.60 & 4.54 \\ 
				LS-SCAD & 3.11 & 1.00 & 1.00 & 0.29 & 0.00 & 22.70 & 3.98 \\ 
				LS-MCP & 3.06 & 1.00 & 1.00 & 0.29 & 0.00 & 22.70 & 3.98 \\
				\hline 
				LAD-LASSO & 3.60 & 1.00 & 1.00 & 4.43 & 0.13 & 38.51 & 5.72 \\ 
				RLARS & 9.61 & 1.00 & 0.99 & 0.54 & 5.68 & 10.11 & 4.59 \\
\hline
		\end{tabular}}
		\label{p1000e0signal1}
	\end{table}
	
	\begin{table}[h]
		\centering
		\caption{Performance measures obtained by different methods for $p=1000$, Setting A and $Y$-outliers}
		\resizebox{\textwidth}{!}{
			\begin{tabular}{l|rrr|rrr|r}
				\hline
				Method	& MS($\widehat{\boldsymbol{\beta}}$) & TP($\widehat{\boldsymbol{\beta}}$) & TN($\widehat{\boldsymbol{\beta}}$) & MSES($\widehat{\boldsymbol{\beta}}$) &  MSEN($\widehat{\boldsymbol{\beta}}$)  &  EE($\widehat{\sigma}$) & APrB($\widehat{\boldsymbol{\beta}}$)  \\ 
				& & & &  $(10^{-2})$ & $(10^{-5})$ &  $(10^{-2})$ &  $(10^{-2})$ \\
				\hline
				LS-LASSO & 1.06 & 0.29 & 1.00 & 408.22 & 1.23 & 620.23 & 32.80 \\
				Ad-LS-LASSO & 0.99 & 0.30 & 1.00 & 359.73 & 9.94 & 563.67 & 31.31 \\ 
				LS-SCAD & 15.91 & 0.73 & 0.99 & 221.89 & 180.13 & 348.35 & 25.68 \\ 
				LS-MCP & 5.16 & 0.56 & 1.00 & 225.93 & 121.87 & 369.16 & 25.91 \\ 
				\hline
				LAD-LASSO & 5.45 & 0.99 & 1.00 & 95.88 & 11.12 & 260.53 & 17.04 \\ 
				RLARS & 4.71 & 1.00 & 1.00 & 1.32 & 1.58 & 7.23 & 4.08 \\
\hline
		\end{tabular}}
		\label{p1000e1signal1}
	\end{table}
	
	\begin{table}[h]
		\centering
		\caption{Performance measures obtained by different methods for $p=1000$, Setting A and $\boldsymbol{X}$-outliers}
		\resizebox{\textwidth}{!}{
			\begin{tabular}{l|rrr|rrr|r}
				\hline
				Method	& MS($\widehat{\boldsymbol{\beta}}$) & TP($\widehat{\boldsymbol{\beta}}$) & TN($\widehat{\boldsymbol{\beta}}$) & MSES($\widehat{\boldsymbol{\beta}}$) &  MSEN($\widehat{\boldsymbol{\beta}}$)  &  EE($\widehat{\sigma}$) & APrB($\widehat{\boldsymbol{\beta}}$)  \\ 
				& & & &  $(10^{-2})$ & $(10^{-5})$ &  $(10^{-2})$ &  $(10^{-2})$ \\
				\hline
				LS-LASSO & 1.13 & 0.31 & 1.00 & 406.53 & 1.17 & 621.43 & 32.82 \\ 
				Ad-LS-LASSO & 3.00 & 1.00 & 1.00 & 1.72 & 0.00 & 28.96 & 4.45 \\ 
				LS-SCAD & 15.09 & 0.71 & 0.99 & 217.82 & 181.86 & 349.93 & 25.93 \\ 
				LS-MCP & 6.63 & 0.58 & 1.00 & 217.82 & 184.10 & 359.61 & 25.04 \\ 
				\hline
				LAD-LASSO & 5.40 & 0.99 & 1.00 & 91.68 & 10.76 & 255.77 & 16.50 \\ 
				RLARS & 4.12 & 1.00 & 1.00 & 1.41 & 0.80 & 7.40 & 4.22 \\
\hline
		\end{tabular}}
		\label{p1000ex1signal1}
	\end{table}
	
	\begin{table}[ht]
		\centering
		\caption{Performance measures obtained by different methods for $p=1000$, Setting A and heavy-tailed $Y$-outliers}
		\label{p1000he1signal1}
		\resizebox{\textwidth}{!}{ {\color{black}
				\begin{tabular}{l|rrr|rrr|r}
					\hline
					Method	& MS($\widehat{\boldsymbol{\beta}}$) & TP($\widehat{\boldsymbol{\beta}}$) & TN($\widehat{\boldsymbol{\beta}}$) & MSES($\widehat{\boldsymbol{\beta}}$) &  MSEN($\widehat{\boldsymbol{\beta}}$)  &  EE($\widehat{\sigma}$) & APrB($\widehat{\boldsymbol{\beta}}$)  \\ 
					& & & &  $(10^{-2})$ & $(10^{-5})$ &  $(10^{-2})$ &  $(10^{-2})$ \\
					\hline
					LS-LASSO & 1.02 & 0.28 & 1.00 & 406.17 & 1.03 & 617.35 & 35.06 \\ 
					LS-SCAD & 13.13 & 0.72 & 0.99 & 223.72 & 121.97 & 360.45 & 24.41 \\ 
					LS-MCP & 5.32 & 0.55 & 1.00 & 224.32 & 145.71 & 361.29 & 23.15 \\ 
					\hline 
					LAD-LASSO & 5.44 & 0.98 & 1.00 & 95.11 & 10.57 & 259.01 & 16.34 \\
					RLARS & 4.59 & 1.00 & 1.00 & 1.31 & 1.46 & 7.28 & 3.46 \\ 
					sLTS & 12.38 & 1.00 & 0.99 & 3.47 & 1.36 & 15.69 & 4.49 \\
\hline
				\end{tabular}
		} }
	\end{table}
	
	\begin{table}[h]
		\centering
		\caption{Performance measures obtained by different methods for $p=1000$, Setting A and heavy-tailed $\boldsymbol{X}$-outliers}
		\label{p1000hex1signal1}
		\resizebox{\textwidth}{!}{ {\color{black}
				\begin{tabular}{l|rrr|rrr|r}
					\hline
					Method	& MS($\widehat{\boldsymbol{\beta}}$) & TP($\widehat{\boldsymbol{\beta}}$) & TN($\widehat{\boldsymbol{\beta}}$) & MSES($\widehat{\boldsymbol{\beta}}$) &  MSEN($\widehat{\boldsymbol{\beta}}$)  &  EE($\widehat{\sigma}$) & APrB($\widehat{\boldsymbol{\beta}}$)  \\ 
					& & & &  $(10^{-2})$ & $(10^{-5})$ &  $(10^{-2})$ &  $(10^{-2})$ \\
					\hline
					LS-LASSO & 4.49 & 1.00 & 1.00 & 2.73 & 0.11 & 33.77 & 4.28 \\ 
					LS-SCAD & 3.08 & 1.00 & 1.00 & 0.30 & 0.00 & 22.63 & 3.50 \\ 
					LS-MCP & 3.06 & 1.00 & 1.00 & 0.30 & 0.01 & 22.57 & 3.50 \\ 
					\hline
					LAD-LASSO & 3.61 & 1.00 & 1.00 & 4.75 & 0.14 & 39.32 & 4.88 \\ 
					RLARS & 11.60 & 1.00 & 0.99 & 0.69 & 7.51 & 11.64 & 3.83 \\ 
					sLTS & 8.05 & 1.00 & 0.99 & 4.44 & 0.62 & 10.89 & 4.94 \\
\hline
				\end{tabular}
		} }
	\end{table}

	\begin{table}[h]
		\centering
		\caption{Performance measures obtained by different methods for $p=500,$ Setting A and no outliers}
		\resizebox{\textwidth}{!}{
			\begin{tabular}{l|rrr|rrr|r}
				\hline
				Method	& MS($\widehat{\boldsymbol{\beta}}$) & TP($\widehat{\boldsymbol{\beta}}$) & TN($\widehat{\boldsymbol{\beta}}$) & MSES($\widehat{\boldsymbol{\beta}}$) &  MSEN($\widehat{\boldsymbol{\beta}}$)  &  EE($\widehat{\sigma}$) & APrB($\widehat{\boldsymbol{\beta}}$)  \\ 
				& & & &  $(10^{-2})$ & $(10^{-5})$ &  $(10^{-2})$ &  $(10^{-2})$ \\
				\hline
				LS-LASSO & 4.28 & 1.00 & 1.00 & 2.50 & 0.29 & 31.93 & 5.39 \\ 
				Ad-LS-LASSO & 3.00 & 1.00 & 1.00 & 1.97 & 0.00 & 28.02 & 5.08 \\ 
				LS-SCAD & 3.03 & 1.00 & 1.00 & 0.32 & 0.01 & 22.10 & 4.53 \\ 
				LS-MCP & 3.03 & 1.00 & 1.00 & 0.32 & 0.02 & 22.06 & 4.53 \\ 
				\hline
				LAD-LASSO & 3.44 & 1.00 & 1.00 & 3.86 & 0.31 & 36.62 & 5.47 \\ 
				RLARS & 8.80 & 1.00 & 0.99 & 0.53 & 11.53 & 9.45 & 4.91 \\
\hline
		\end{tabular}}
		\label{p100e0signal1}
	\end{table}

	\begin{table}[h]
		\centering
		\caption{Performance measures obtained by different methods for $p=500$, Setting B and no outliers}
		\resizebox{\textwidth}{!}{\begin{tabular}{l|rrr|rrr|r}
				\hline
				Method	& MS($\widehat{\boldsymbol{\beta}}$) & TP($\widehat{\boldsymbol{\beta}}$) & TN($\widehat{\boldsymbol{\beta}}$) & MSES($\widehat{\boldsymbol{\beta}}$) &  MSEN($\widehat{\boldsymbol{\beta}}$)  &  EE($\widehat{\sigma}$) & APrB($\widehat{\boldsymbol{\beta}}$)  \\ 
				& & & &  $(10^{-2})$ & $(10^{-5})$ &  $(10^{-2})$ &  $(10^{-2})$ \\
				\hline
				LS-LASSO & 18.33 & 1.00 & 0.98 & 2.10 & 2.65 & 30.97 & 5.76 \\ 
				Ad-LS-LASSO & 9.00 & 1.00 & 1.00 & 1.38 & 0.00 & 30.60 & 5.77 \\ 
				LS-SCAD & 9.01 & 1.00 & 1.00 & 0.34 & 0.00 & 19.45 & 4.67 \\ 
				LS-MCP & 9.01 & 1.00 & 1.00 & 0.34 & 0.01 & 19.48 & 4.67 \\ 
				\hline
				LAD-LASSO & 14.71 & 1.00 & 0.99 & 5.90 & 4.43 & 47.53 & 7.59 \\ 
				RLARS & 13.27 & 1.00 & 0.99 & 0.48 & 6.34 & 8.12 & 5.03 \\
\hline
		\end{tabular}}
		\label{p100e0signal3}
	\end{table}
	
	\begin{table}[h]
		\centering
		\caption{Performance measures obtained by different methods for $p=500,$ Setting A and $Y$-outliers}
		\resizebox{\textwidth}{!}{\begin{tabular}{l|rrr|rrr|r}
				\hline
				Method	& MS($\widehat{\boldsymbol{\beta}}$) & TP($\widehat{\boldsymbol{\beta}}$) & TN($\widehat{\boldsymbol{\beta}}$) & MSES($\widehat{\boldsymbol{\beta}}$) &  MSEN($\widehat{\boldsymbol{\beta}}$)  &  EE($\widehat{\sigma}$) & APrB($\widehat{\boldsymbol{\beta}}$)  \\ 
				& & & &  $(10^{-2})$ & $(10^{-5})$ &  $(10^{-2})$ &  $(10^{-2})$ \\
				\hline
				LS-LASSO & 1.26 & 0.34 & 1.00 & 385.24 & 3.82 & 620.33 & 31.72 \\ 
				Ad-LS-LASSO & 1.23 & 0.36 & 1.00 & 328.79 & 25.46 & 546.10 & 28.96 \\ 
				LS-SCAD & 12.41 & 0.73 & 0.98 & 191.12 & 334.02 & 326.86 & 20.89 \\ 
				LS-MCP & 5.09 & 0.62 & 0.99 & 201.85 & 227.96 & 336.84 & 21.07 \\
				\hline 
				LAD-LASSO & 4.82 & 0.99 & 1.00 & 75.46 & 20.21 & 234.18 & 14.35 \\ 
				RLARS & 3.99 & 1.00 & 1.00 & 0.37 & 1.28 & 7.03 & 4.47 \\
\hline
		\end{tabular}}
		\label{p500e1signal1}
	\end{table}
	
	\begin{table}[h]
		\centering
		\caption{Performance measures obtained by different methods for $p=500$, Setting B and $Y$-outliers}
		\resizebox{\textwidth}{!}{\begin{tabular}{l|rrr|rrr|r}
				\hline
				\multicolumn{8}{c}{ Setting B} \\
				\hline
				Method	& MS($\widehat{\boldsymbol{\beta}}$) & TP($\widehat{\boldsymbol{\beta}}$) & TN($\widehat{\boldsymbol{\beta}}$) & MSES($\widehat{\boldsymbol{\beta}}$) &  MSEN($\widehat{\boldsymbol{\beta}}$)  &  EE($\widehat{\sigma}$) & APrB($\widehat{\boldsymbol{\beta}}$)  \\ 
				& & & &  $(10^{-2})$ & $(10^{-5})$ &  $(10^{-2})$ &  $(10^{-2})$ \\
				\hline
				LS-LASSO & 7.82 & 0.55 & 0.99 & 305.11 & 74.46 & 875.57 & 47.54 \\ 
				Ad-LS-LASSO & 5.55 & 0.47 & 1.00 & 271.78 & 391.65 & 708.74 & 43.05 \\
				LS-SCAD & 22.64 & 0.69 & 0.97 & 246.17 & 784.24 & 535.66 & 35.60 \\ 
				LS-MCP & 10.57 & 0.56 & 0.99 & 254.04 & 583.06 & 574.63 & 35.31 \\
				\hline 
				LAD-LASSO & 24.60 & 0.93 & 0.97 & 107.66 & 371.98 & 287.10 & 29.44 \\
\hline
		\end{tabular}}
		\label{p500e1signal3}
	\end{table}
	
	\begin{table}[h]
		\centering
		\caption{Performance measures obtained by different methods for $p=500,$ Setting A and $\boldsymbol{X}$-outliers}
		\resizebox{\textwidth}{!}{\begin{tabular}{l|rrr|rrr|r}
				\hline
				Method	& MS($\widehat{\boldsymbol{\beta}}$) & TP($\widehat{\boldsymbol{\beta}}$) & TN($\widehat{\boldsymbol{\beta}}$) & MSES($\widehat{\boldsymbol{\beta}}$) &  MSEN($\widehat{\boldsymbol{\beta}}$)  &  EE($\widehat{\sigma}$) & APrB($\widehat{\boldsymbol{\beta}}$)  \\ 
				& & & &  $(10^{-2})$ & $(10^{-5})$ &  $(10^{-2})$ &  $(10^{-2})$ \\
				\hline
				LS-LASSO & 1.26 & 0.34 & 1.00 & 385.24 & 3.82 & 620.33 & 31.72 \\ 
				Ad-LS-LASSO & 3.00 & 1.00 & 1.00 & 1.94 & 0.00 & 28.09 & 4.96 \\
				LS-SCAD & 12.41 & 0.73 & 0.98 & 191.12 & 334.02 & 326.86 & 20.89 \\ 
				LS-MCP & 5.09 & 0.62 & 0.99 & 201.85 & 227.96 & 336.84 & 21.07 \\ 
				\hline
				LAD-LASSO & 4.82 & 0.99 & 1.00 & 75.46 & 20.21 & 234.18 & 14.35 \\ 
				RLARS & 3.99 & 1.00 & 1.00 & 0.37 & 1.28 & 7.03 & 4.47 \\
\hline
		\end{tabular}}
		\label{p500ex1signal1}
	\end{table}
	
	\begin{table}[h]
		\centering
		\caption{Performance measures obtained by different methods for $p=500$, Setting B and $\boldsymbol{X}$-outliers}
		\resizebox{\textwidth}{!}{\begin{tabular}{l|rrr|rrr|r}
				\hline
				Method	& MS($\widehat{\boldsymbol{\beta}}$) & TP($\widehat{\boldsymbol{\beta}}$) & TN($\widehat{\boldsymbol{\beta}}$) & MSES($\widehat{\boldsymbol{\beta}}$) &  MSEN($\widehat{\boldsymbol{\beta}}$)  &  EE($\widehat{\sigma}$) & APrB($\widehat{\boldsymbol{\beta}}$)  \\ 
				& & & &  $(10^{-2})$ & $(10^{-5})$ &  $(10^{-2})$ &  $(10^{-2})$ \\
				\hline
				LS-LASSO & 7.82 & 0.55 & 0.99 & 305.11 & 74.46 & 875.57 & 47.54 \\
				Ad-LS-LASSO & 9.00 & 1.00 & 1.00 & 1.38 & 0.00 & 30.38 & 5.75 \\  
				LS-SCAD & 22.64 & 0.69 & 0.97 & 246.17 & 784.24 & 535.66 & 35.60 \\ 
				LS-MCP & 10.57 & 0.56 & 0.99 & 254.04 & 583.06 & 574.63 & 35.31 \\
				\hline 
				LAD-LASSO & 24.60 & 0.93 & 0.97 & 107.66 & 371.98 & 287.10 & 29.44 \\ 
				RLARS & 19.93 & 0.86 & 0.98 & 53.56 & 217.95 & 111.83 & 15.45 \\
\hline 
		\end{tabular}}
		\label{p500ex1signal3}
	\end{table}
	
	\begin{table}[ht]
		\centering
		\caption{Performance measures obtained by different methods for $p=500$, Setting A and heavy-tailed  $Y$-outliers}
		\label{p500he1signal1}
		\resizebox{\textwidth}{!}{ {\color{black}
				\begin{tabular}{l|rrr|rrr|r}
					\hline
					Method	& MS($\widehat{\boldsymbol{\beta}}$) & TP($\widehat{\boldsymbol{\beta}}$) & TN($\widehat{\boldsymbol{\beta}}$) & MSES($\widehat{\boldsymbol{\beta}}$) &  MSEN($\widehat{\boldsymbol{\beta}}$)  &  EE($\widehat{\sigma}$) & APrB($\widehat{\boldsymbol{\beta}}$)  \\ 
					& & & &  $(10^{-2})$ & $(10^{-5})$ &  $(10^{-2})$ &  $(10^{-2})$ \\
					\hline
					LS-LASSO & 1.39 & 0.37 & 1.00 & 375.27 & 4.79 & 608.93 & 32.12 \\ 
					LS-SCAD & 11.38 & 0.75 & 0.98 & 190.26 & 276.32 & 332.72 & 20.10 \\ 
					LS-MCP & 5.38 & 0.63 & 0.99 & 202.17 & 311.82 & 336.04 & 20.96 \\ 
					\hline 
					LAD-LASSO & 4.98 & 0.99 & 1.00 & 75.77 & 18.39 & 235.36 & 14.43 \\ 
					RLARS & 3.74 & 1.00 & 1.00 & 0.37 & 1.15 & 6.83 & 3.53 \\ 
					sLTS & 9.63 & 1.00 & 0.99 & 2.77 & 2.15 & 14.83 & 4.66 \\
\hline	
				\end{tabular}
		} }
	\end{table}

	\begin{table}[ht]
		\centering
		\caption{Performance measures obtained by different methods for $p=500$, Setting B and heavy-tailed $Y$-outliers}
		\label{p500he1signal3}
		\resizebox{\textwidth}{!}{{\color{black}
				\begin{tabular}{l|rrr|rrr|r}
					\hline
					Method	& MS($\widehat{\boldsymbol{\beta}}$) & TP($\widehat{\boldsymbol{\beta}}$) & TN($\widehat{\boldsymbol{\beta}}$) & MSES($\widehat{\boldsymbol{\beta}}$) &  MSEN($\widehat{\boldsymbol{\beta}}$)  &  EE($\widehat{\sigma}$) & APrB($\widehat{\boldsymbol{\beta}}$)  \\ 
					& & & &  $(10^{-2})$ & $(10^{-5})$ &  $(10^{-2})$ &  $(10^{-2})$ \\
					\hline
					LS-LASSO & 8.33 & 0.57 & 0.99 & 298.33 & 132.06 & 867.00 & 51.28 \\ 
					LS-SCAD & 22.58 & 0.69 & 0.97 & 241.12 & 795.59 & 531.53 & 39.61 \\ 
					LS-MCP & 10.56 & 0.55 & 0.99 & 257.76 & 887.80 & 586.36 & 40.71 \\ 
					\hline
					LAD-LASSO & 23.43 & 0.94 & 0.97 & 105.29 & 279.06 & 286.05 & 31.37 \\ 
					RLARS & 21.76 & 0.89 & 0.97 & 46.03 & 206.57 & 93.98 & 13.84 \\ 
					sLTS & 37.06 & 0.99 & 0.94 & 22.36 & 66.30 & 13.53 & 13.86 \\
\hline	
				\end{tabular}
		} }
	\end{table}

	\begin{table}[ht]
		\centering
		\caption{Performance measures obtained by different methods for $p=500$, Setting A and heavy-tailed $\boldsymbol{X}$-outliers}
		\label{p500hes1signal1}
		\resizebox{\textwidth}{!}{ {\color{black}
				\begin{tabular}{l|rrr|rrr|r}
					\hline
					Method	& MS($\widehat{\boldsymbol{\beta}}$) & TP($\widehat{\boldsymbol{\beta}}$) & TN($\widehat{\boldsymbol{\beta}}$) & MSES($\widehat{\boldsymbol{\beta}}$) &  MSEN($\widehat{\boldsymbol{\beta}}$)  &  EE($\widehat{\sigma}$) & APrB($\widehat{\boldsymbol{\beta}}$)  \\ 
					& & & &  $(10^{-2})$ & $(10^{-5})$ &  $(10^{-2})$ &  $(10^{-2})$ \\
					\hline
					LS-LASSO & 4.45 & 1.00 & 1.00 & 2.48 & 0.25 & 32.19 & 4.58 \\ 
					LS-SCAD & 3.05 & 1.00 & 1.00 & 0.33 & 0.00 & 22.39 & 3.51 \\ 
					LS-MCP & 3.03 & 1.00 & 1.00 & 0.33 & 0.01 & 22.36 & 3.52 \\ 
					\hline
					LAD-LASSO & 3.44 & 1.00 & 1.00 & 3.63 & 0.29 & 35.18 & 4.90 \\ 
					RLARS & 7.48 & 1.00 & 0.99 & 0.52 & 8.62 & 8.37 & 3.77 \\ 
					sLTS & 5.89 & 1.00 & 0.99 & 3.37 & 0.73 & 8.10 & 4.78 \\
\hline 
				\end{tabular}
		} }
	\end{table}

	\begin{table}[ht]
		\centering
		\caption{Performance measures obtained by different methods for $p=500$, Setting B and heavy-tailed $\boldsymbol{X}$-outliers}
		\label{p500hex1signal3}
		\resizebox{\textwidth}{!}{ {\color{black}
				\begin{tabular}{l|rrr|rrr|r}
					\hline
					Method	& MS($\widehat{\boldsymbol{\beta}}$) & TP($\widehat{\boldsymbol{\beta}}$) & TN($\widehat{\boldsymbol{\beta}}$) & MSES($\widehat{\boldsymbol{\beta}}$) &  MSEN($\widehat{\boldsymbol{\beta}}$)  &  EE($\widehat{\sigma}$) & APrB($\widehat{\boldsymbol{\beta}}$)  \\ 
					& & & &  $(10^{-2})$ & $(10^{-5})$ &  $(10^{-2})$ &  $(10^{-2})$ \\
					\hline
					LS-LASSO & 18.44 & 1.00 & 0.98 & 2.02 & 2.36 & 32.21 & 5.52 \\ 
					LS-SCAD & 9.01 & 1.00 & 1.00 & 0.34 & 0.00 & 19.66 & 3.56 \\ 
					LS-MCP & 9.01 & 1.00 & 1.00 & 0.34 & 0.00 & 19.71 & 3.56 \\ 
					\hline
					LAD-LASSO & 14.20 & 1.00 & 0.99 & 5.76 & 4.10 & 49.74 & 7.79 \\ 
					RLARS & 12.06 & 1.00 & 0.99 & 0.43 & 4.33 & 7.30 & 3.74 \\ 
					sLTS & 47.67 & 0.94 & 0.92 & 91.55 & 264.41 & 11.30 & 28.93 \\
\hline
				\end{tabular}
		} }
	\end{table}
	
	\begin{table}[htb]
		\centering
		\caption{Runtimes of the different methods for the whole fitness process}
		\label{table:runtime}
		{\color{black}
			\begin{tabular}{cccc}
				\hline
				& Pure data & $10\%$ of $Y$-outliers & $10\%$ of $\boldsymbol{X}$-outliers \\ 
				\hline
				\multicolumn{4}{c}{Setting A}\\
				\hline
				AW-DPD-LASSO $\gamma$ =0.1 & 7.639 & 7.806 & 7.210 \\ 
				AW-DPD-LASSO $\gamma$ =0.3 & 7.575 & 8.073 & 7.817 \\ 
				AW-DPD-LASSO $\gamma$ =0.5 & 7.581 & 7.778 & 7.428 \\ 
				AW-DPD-LASSO $\gamma$ =0.7 & 7.188 & 7.611 & 7.222 \\ 
				AW-DPD-LASSO $\gamma$ =1 & 7.076 & 8.184 & 6.907 \\ 
				\hline
				DPD-ncv $\gamma$ =0.1 & 7.793 & 7.869 & 8.453 \\ 
				DPD-ncv $\gamma$ =0.3 & 7.365 & 7.746 & 7.438 \\ 
				DPD-ncv $\gamma$ =0.5 & 7.563 & 7.654 & 7.341 \\ 
				DPD-ncv $\gamma$ =0.7 & 7.665 & 8.229 & 7.395 \\ 
				DPD-ncv $\gamma$ =1 & 7.479 & 8.313 & 7.591 \\ 
				\hline
				\multicolumn{4}{c}{Setting B}\\
				\hline
				AW-DPD-LASSO $\gamma$ =0.1 & 7.520 & 8.011 & 7.388 \\ 
				AW-DPD-LASSO $\gamma$ =0.3 & 7.269 & 7.811 & 7.220 \\ 
				AW-DPD-LASSO $\gamma$ =0.5 & 7.130 & 7.769 & 7.377 \\ 
				AW-DPD-LASSO $\gamma$ =0.7 & 7.294 & 7.560 & 7.070 \\ 
				AW-DPD-LASSO $\gamma$ =1 & 7.085 & 7.340 & 7.185 \\ 
				DPD-ncv $\gamma$ =0.1 & 7.895 & 8.473 & 7.396 \\ 
				DPD-ncv $\gamma$ =0.3 & 7.510 & 7.970 & 7.622 \\ 
				DPD-ncv $\gamma$ =0.5 & 7.649 & 8.578 & 7.453 \\ 
				DPD-ncv $\gamma$ =0.7 & 7.631 & 9.147 & 7.666 \\ 
				DPD-ncv $\gamma$ =1 & 7.781 & 9.451 & 7.815 \\ 
				\hline
			\end{tabular}
		}
	\end{table}

	\begin{table}[!h]
		\centering
		\caption{Error measures for X-ray microanalysis (EPXMA) with the different methods.}
		\small{\begin{tabular}{rrrrrr}
				\hline
				& MS & MSE & MSPE(0.9) & MIN & MAX \\ 
				\hline
				Ad-LS-LASSO & 12.00 & 10.89 & 7.36 & 13.07 & 44.73 \\ 
				LS-SCAD & 8.00 & 15.47 & 10.46 & 11.53 & 26.59 \\ 
				Ad-DPD-LASSO	 $\gamma = $ 0.1 & 12.00 & 6.54 & 3.74 & 14.73 & 51.16 \\ 
				Ad-DPD-LASSO  $\gamma = $ 0.3 & 11.00 & 7.73 & 3.82 & 17.45 & 17.94 \\ 
				Ad-DPD-LASSO  $\gamma = $ 0.5 & 8.00 & 10.69 & 5.10 & 14.44 & 5.16 \\ 
				Ad-DPD-LASSO  $\gamma = $ 0.7 & 8.00 & 11.10 & 5.08 & 15.14 & 0.76 \\ 
				Ad-DPD-LASSO  $\gamma = $ 1 & 8.00 & 11.33 & 5.12 & 15.40 & 11.82 \\
\hline
		\end{tabular}}
		\label{TAB:data1}
	\end{table}

	\begin{table}[h]
		\centering
		\caption{Median of $\tau-$scale of the residuals of each of the estimators for the electronprobe X-ray microanalysis data}
		\small{\begin{tabular}{rr}
				\hline
				& $\tau-$escale \\
				& ($\times 10^{2}$)  \\ 
				Ad-LS-LASSO & 3.84\\ 
				LS-SCAD & 4.67\\ 
				\hline
				AW-DPD-LASSO $\gamma = $ 0.1 & 3.67 \\ 
				AW-DPD-LASSO $\gamma = $ 0.3 & 3.87 \\ 
				AW-DPD-LASSO $\gamma = $ 0.5 & 4.08 \\ 
				AW-DPD-LASSO $\gamma = $ 0.7 &  4.21 \\ 
				\hline
				Ad-DPD-LASSO  $\gamma = $ 0.1 & 3.11 \\ 
				Ad-DPD-LASSO  $\gamma = $ 0.3 & 3.44 \\ 
				Ad-DPD-LASSO  $\gamma = $ 0.5 & 3.63\\ 
				Ad-DPD-LASSO  $\gamma = $ 0.7 &  3.57 \\ 
				\hline
				DPD-ncv   $\gamma = $ 0.1 &  3.87 \\ 
				DPD-ncv  $\gamma = $ 0.3 & 3.86\\ 
				DPD-ncv   $\gamma = $ 0.5 & 3.93 \\ 
				DPD-ncv  $\gamma = $ 0.7 &  4.11 \\ 
				\hline
		\end{tabular}}
		\label{tauscale}
	\end{table}
	
	\section{Comparing Ad-DPD-LASSO and AW-DPD-LASSO estimators in terms of prediction performance}
	
	We have seen from our simulations presented in Section 6 that the Ad-DPD-LASSO estimator selects pretty better the model size than AW-DPD-LASSO. 
	In contrast, AW-DPD-LASSO is more accurate in estimating the regression coefficient, which is shown by generally lower MSES and MSEN values. 
	Finally, both performs similarly in estimating $\sigma$ which is significantly better than other existing algorithms considered. 
	We further compare the performance of Ad-DPD-LASSO and AW-DPD-LASSO, in terms of their prediction accuracy, 
	for varying  number of covariates and  contamination on proportions in the data. 
	
	For this purpose, we repeat the previous simulation exercise for Setting A, with varying $p$ and contamination proportion,
	and including similarly generated test sample to examine the accuracy of the prediction. 
	Once the model is fitted and the regression parameters are estimated based on the training data, 
	we evaluate the mean root square error of prediction based on the test data as $\text{RMSE} = \sqrt{\frac{1}{n}\sum_{i=1}^n\left(y_{\text{test},i}-\boldsymbol{x}_{\text{test},i}^T\widehat{\boldsymbol{\beta}}\right)^2}$. 
	for both Ad-DPD-LASSO and AW-DPD-LASSO estimators. 
	Figure \ref{covVSerror} shows the RSME values against the number of covariates for pure data, 
	$10\%$ of $Y$-outliers and $10\%$ of $\boldsymbol{X}$-outliers.
	Both methods perform similarly in all settings, with or without contamination, with the Ad-DPD-LASSO being slightly more steady. 
	By contrast AW-DPD-LASSO  produces lower RMSE when the number of covariates is lower, showing a greater  prediction accuracy.  
	
	\begin{figure}[htb]
		\centering
		\begin{subfigure}[]{\textwidth}
			\centering
			\includegraphics[height=3.9cm, width=5.6cm]{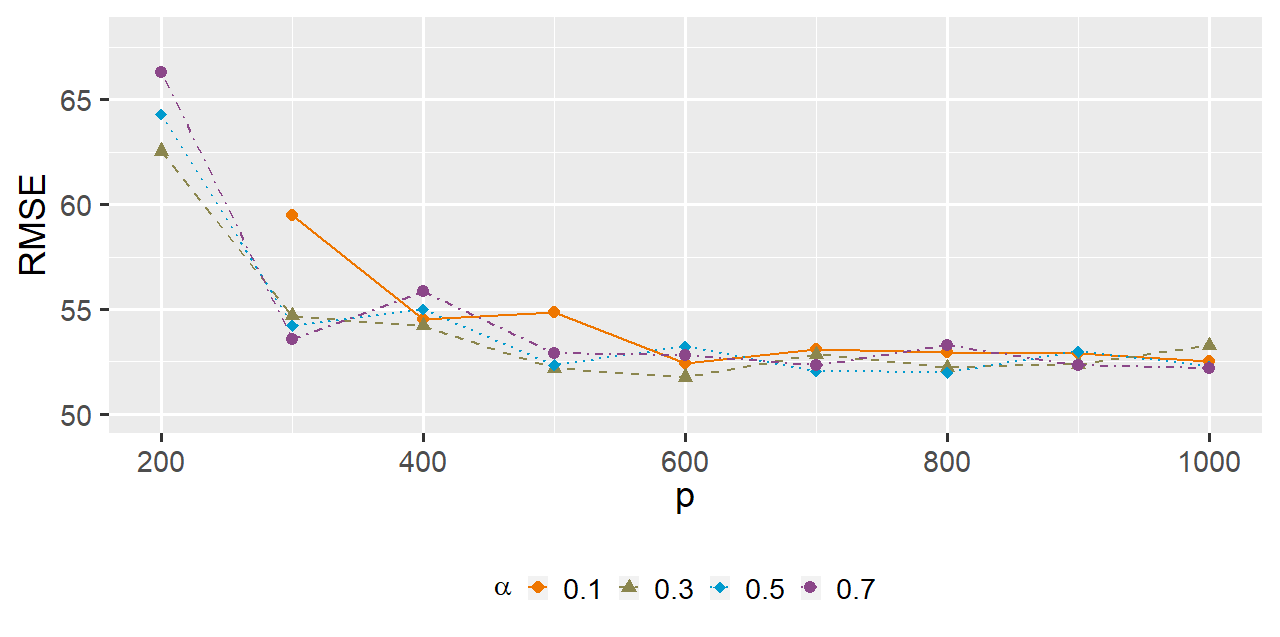}
			\includegraphics[height=3.9cm, width=5.6cm]{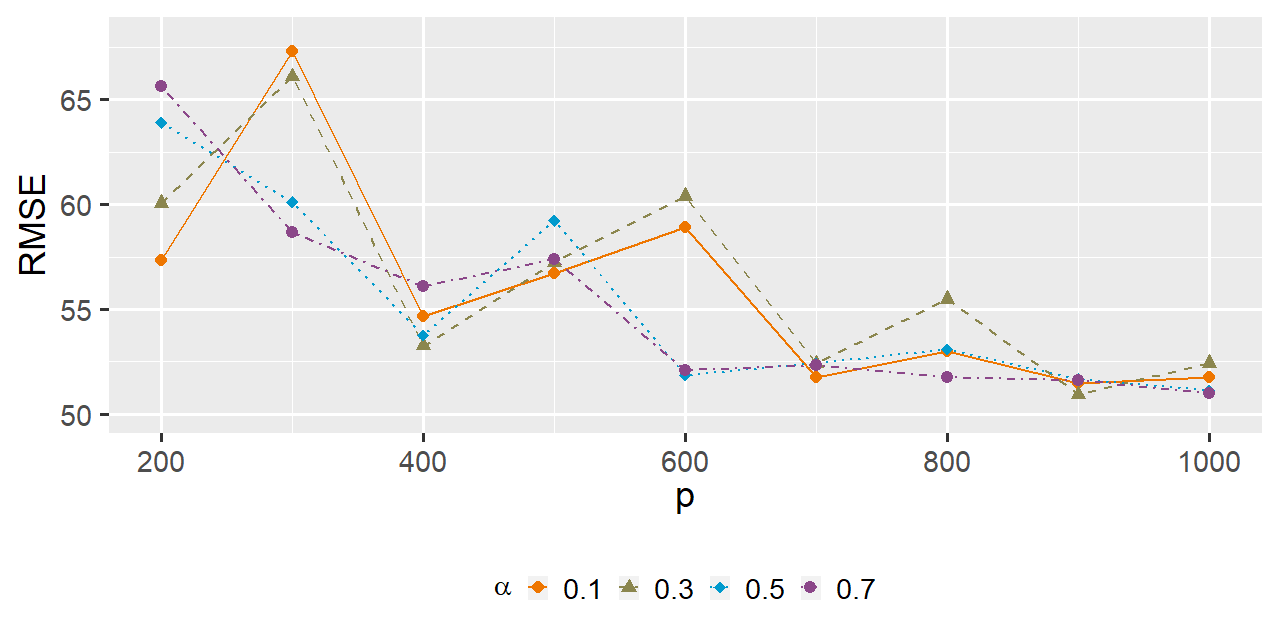}
			\caption{Pure data}
		\end{subfigure}
		\begin{subfigure}[b]{\textwidth}
			\centering
			\includegraphics[height=3.9cm, width=5.6cm]{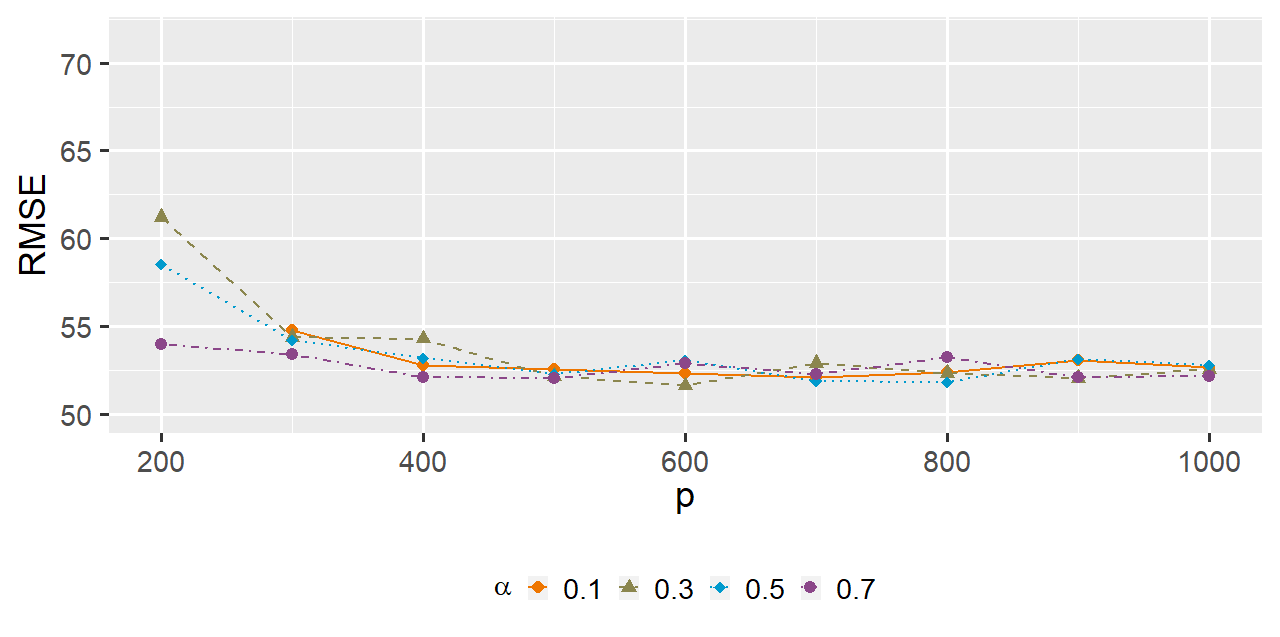}
			\includegraphics[height=3.9cm, width=5.6cm]{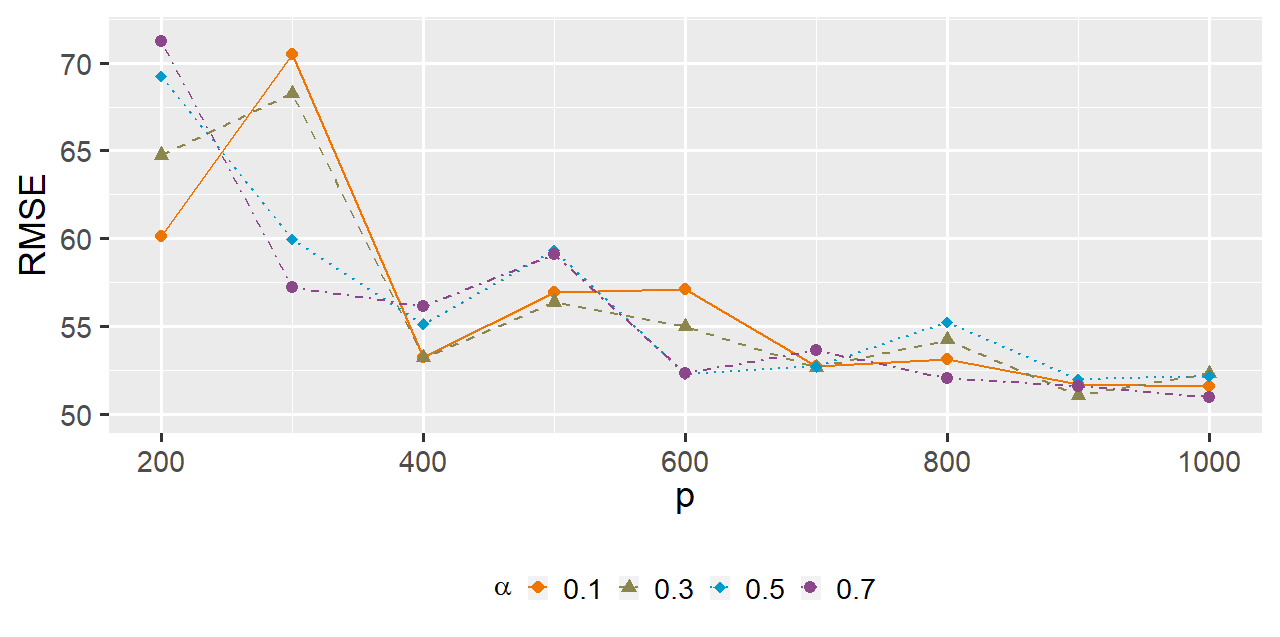}
			\subcaption{10$\%$ of $Y$-outliers}
		\end{subfigure}
		\begin{subfigure}[b]{\textwidth}
			\centering
			\includegraphics[height=3.9cm, width=5.6cm]{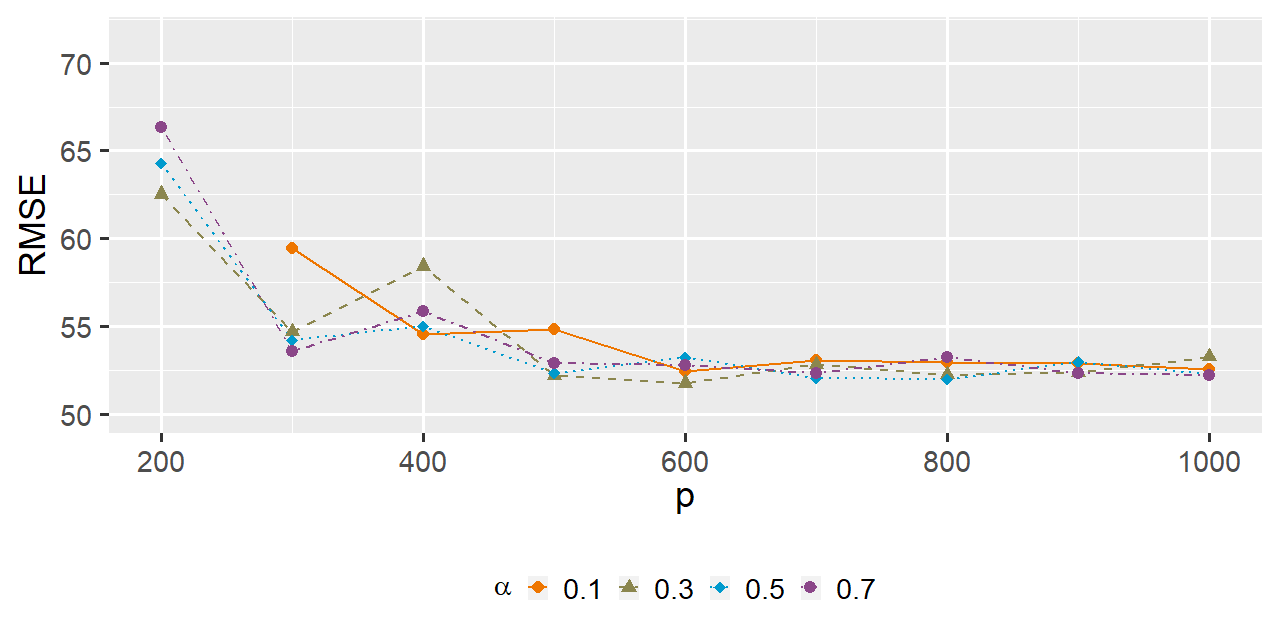}
			\includegraphics[height=3.9cm, width=5.6cm]{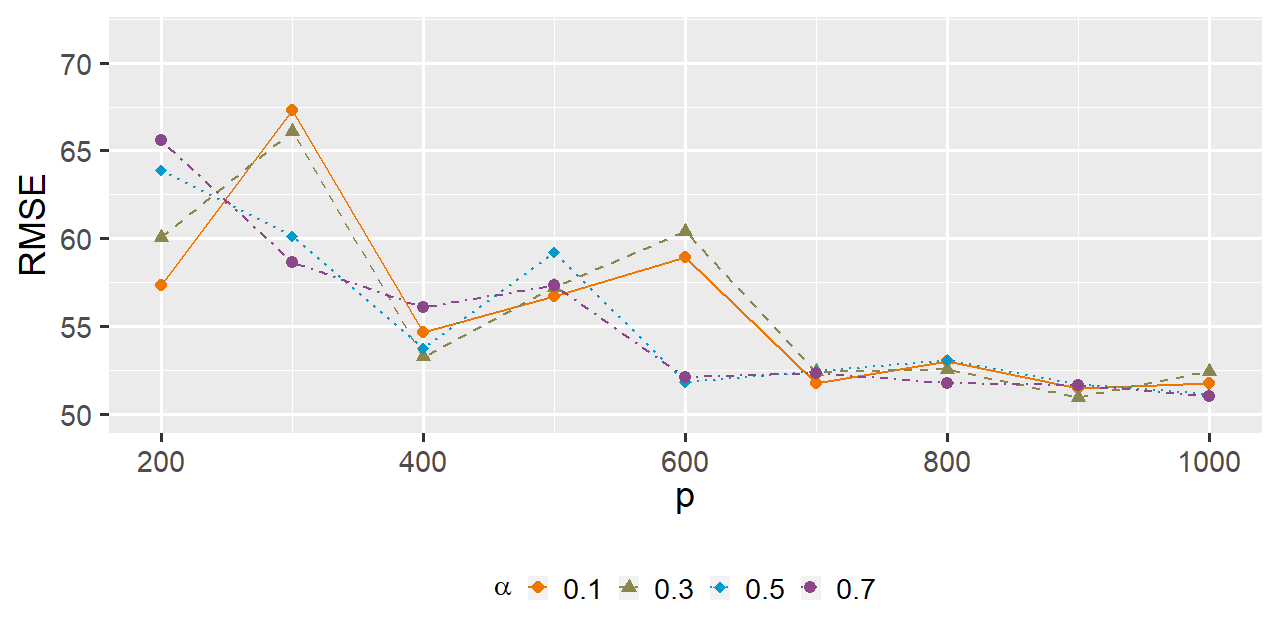}
			\subcaption{10$\%$ of $\boldsymbol{X}$-outliers}
		\end{subfigure}
		\caption{Number of covariates against RMSE obtained with Ad-DPD-LASSO (left panel) and AW-DPD-LASSO (right panel) methods under different simulation settings.}
		\label{covVSerror}
	\end{figure}
	
	Finally, we plot the RMSE  over the contamination level on $Y$-outliers from $0\%$ (pure data) to $25\%$, 
	for the case with  $p=500$ and Setting A in Figure \ref{contVSerror}; 
	the same for the $\boldsymbol{X}$-outliers are shown in Figure \ref{contVSerrorex1}. 
	For both methods, the robustness increases with greater values of $\gamma$, 
	but they perform quite similarly in terms of prediction error for the same values of $\gamma$. 
	For a greater contamination level (more than $15\%$ contamination), the RMSE at $\gamma=0.1$ grows exponentially.
	For covariate contamination, low values of $\gamma$ produce greater RMSE even under small contamination percentages.
	
	\begin{figure}[hbt]
		\begin{subfigure}[b]{\textwidth}
			\centering
			\includegraphics[height=5cm, width=5.6cm]{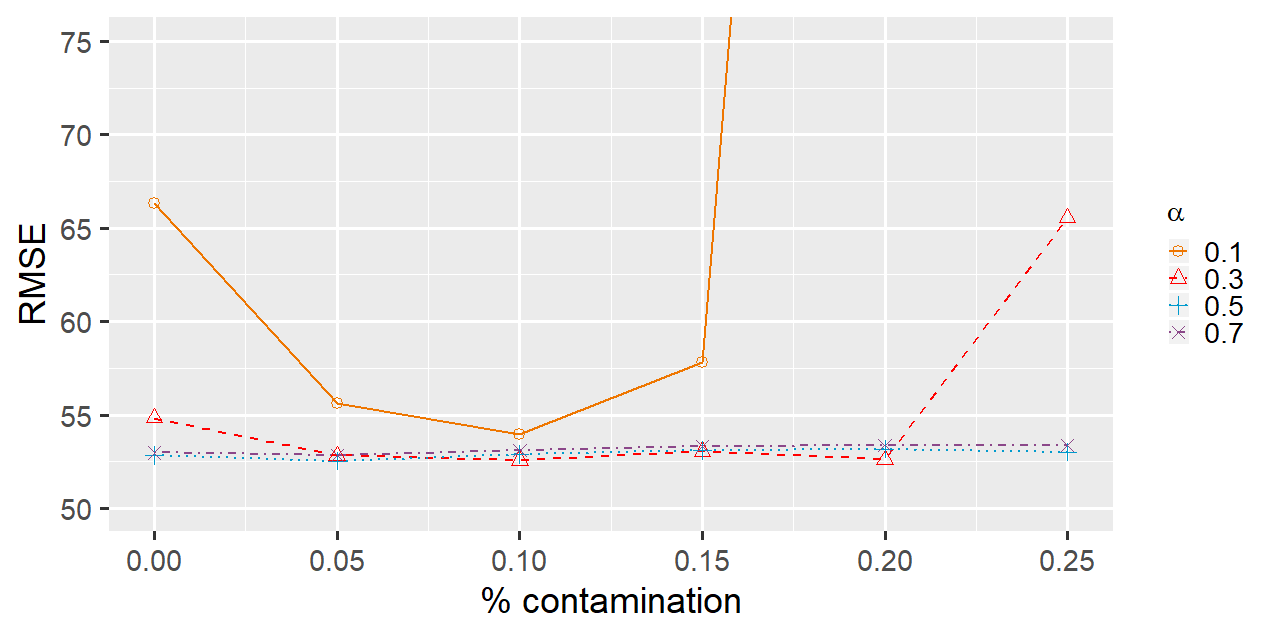}
			\includegraphics[height=5cm, width=5.6cm]{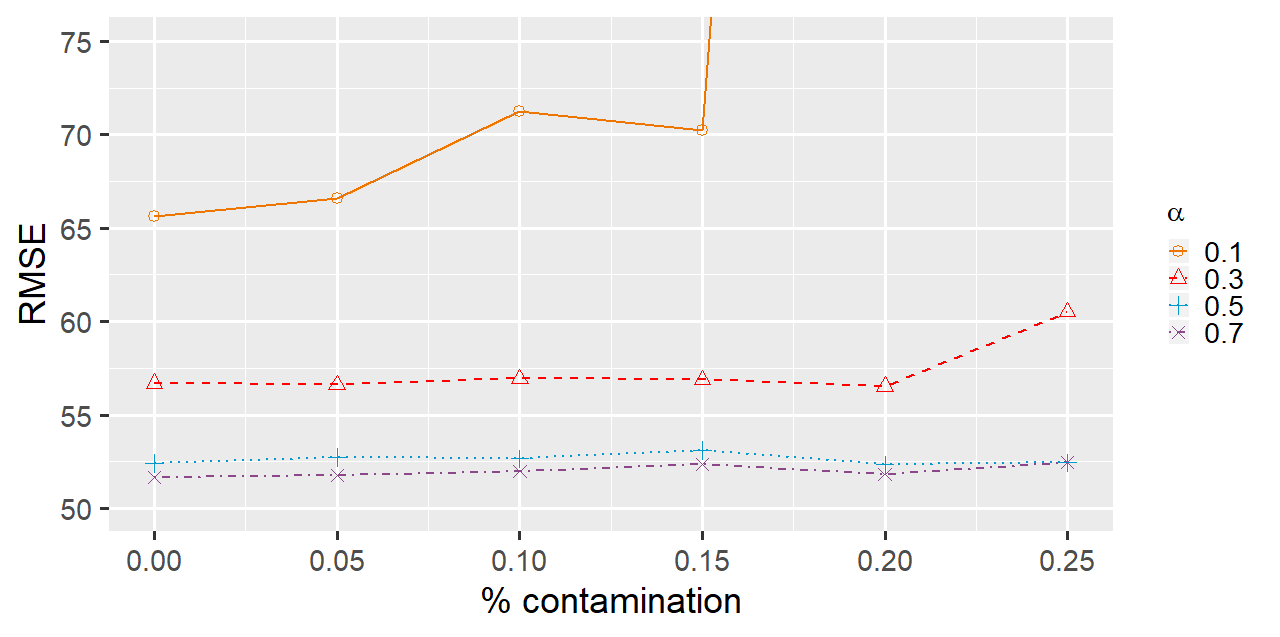}
			\caption{$Y$-outliers}
			\label{contVSerror}
		\end{subfigure}
	%
		\begin{subfigure}[b]{\textwidth}
			\centering
			\includegraphics[height=5cm, width=5.6cm]{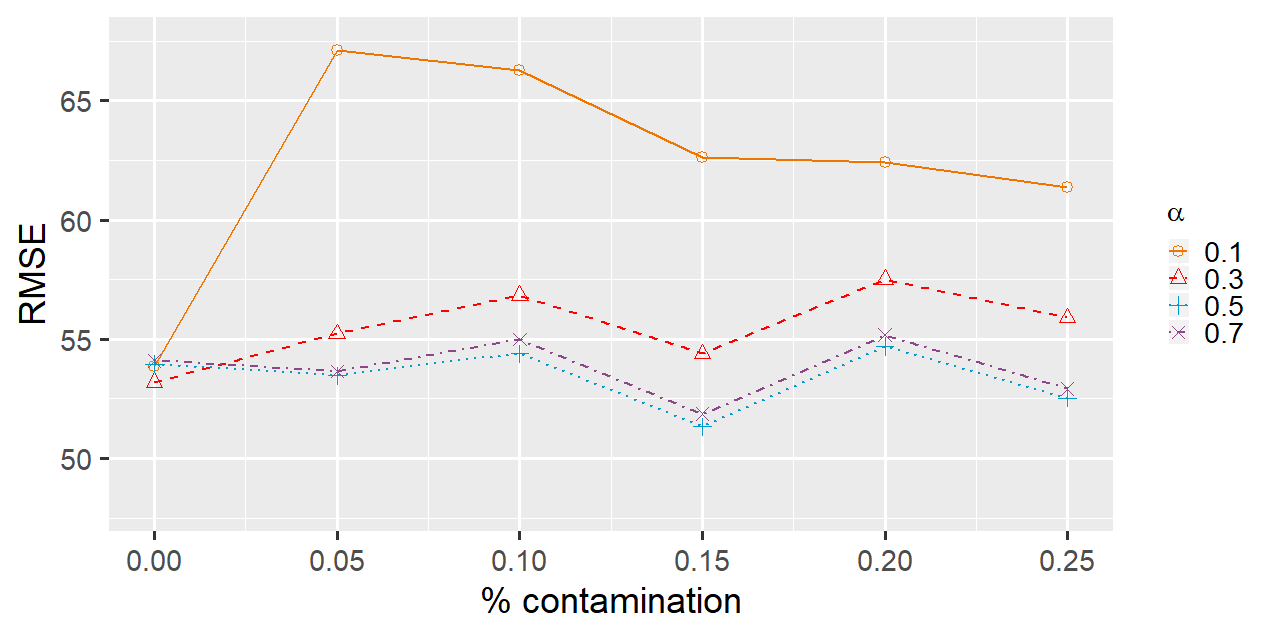}
			\includegraphics[height=5cm, width=5.6cm]{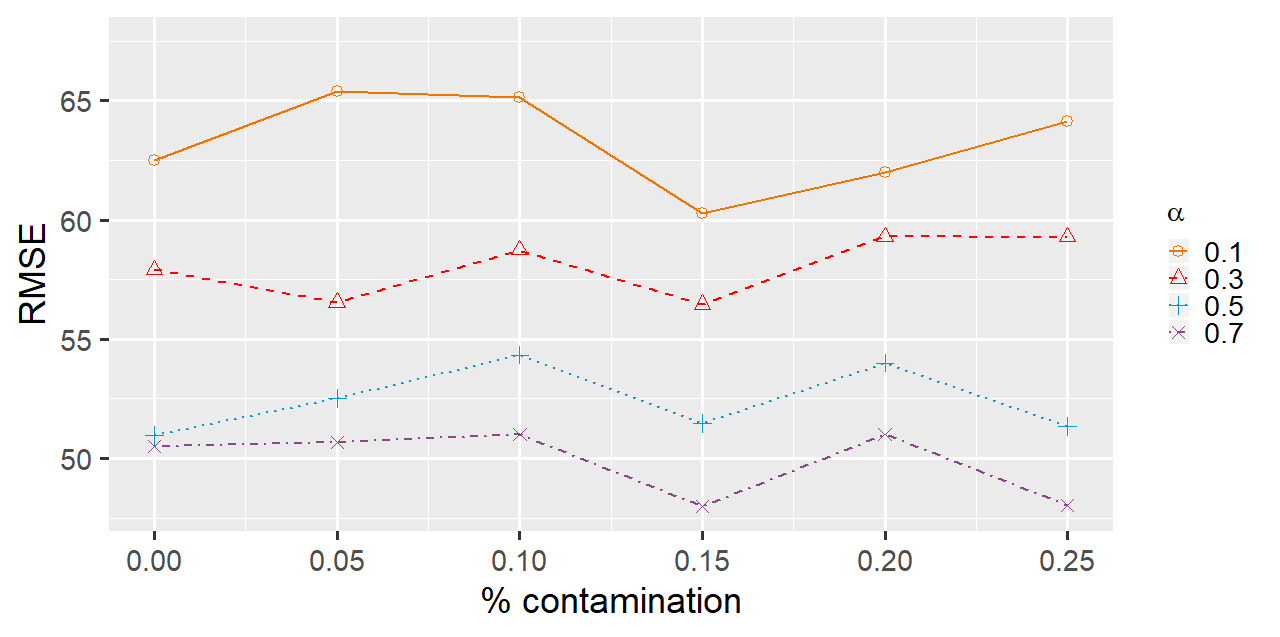}
			\caption{ $\boldsymbol{X}$-outliers}
			\label{contVSerrorex1}
		\end{subfigure}
		\caption{Contamination level of $Y$ and $\boldsymbol{X}$-outliers against RMSE with Ad-DPD-LASSO (left) and AW-DPD-LASSO (right) methods}
		\label{contVSerrore}
	\end{figure}

\end{document}